\documentclass[conference]{IEEEtran}
\def\techreport{}

\usepackage{hyperref}
\usepackage[ignoreall]{linkstodefs}
\usepackage{latexsym}
\usepackage{tikz}
\usetikzlibrary{arrows}
\usepackage{amsmath}
\usepackage{amsthm}
\usepackage{txfonts}
\usepackage[ruled,vlined,linesnumbered,lined]{algorithm2e}

\tikzstyle{distr}=[inner sep=0mm, minimum size=1mm, draw, circle, fill]
\tikzstyle{state}=[inner sep=0.3mm, thick, minimum size=4mm, draw]
\tikzstyle{trarr}=[semithick, -latex, rounded corners]

\def\newcommandwl#1#2{\newcommand{#1}{#2}\withlinks{#1}}%

\newtheorem{example}{Example}
\newtheorem{lemma}{Lemma}
\newtheorem{theorem}{Theorem}
\newtheorem{remark}{Remark}
\newtheorem{corollary}{Corollary}
\newtheorem{proposition}{Proposition}
\newtheorem{claim}{Claim}

\newcommand{\violating}[2]{{\cal C}(#1,#2)} %
\newcommand{\exact}[2]{{\cal D}(#1,#2)} %

\newcommand{\staymec}[1]{R_{#1}} %
\withlinksparam{\staymec}
\newcommandwl{\lraname}{\mathit{mp}} %
\newcommandwl{\lrvlname}{\mathit{lv}} %
\newcommandwl{\lrvhname}{\mathit{hv}} %
\newcommand{\lra}[1]{\lraname(#1)} %
\withlinksparam{\lra}
\newcommand{\lrvl}[1]{\lrvlname(#1)} %
\withlinksparam{\lrvl}
\newcommand{\lrvh}[3]{\lrvhname^{#1,#2}(#3)} %
\withlinksthreeparam{\lrvh}

\newcommandwl{\MECuni}{\mathcal{C}}
\newcommandwl{\MEC}{C}
\newcommand{\MECact}[1]{#1\cap A}
 \withlinksparam{\MECact}
\newcommand{\MECstate}[1]{#1\cap S}
 \withlinksparam{\MECstate}
\newcommand{\BSCCuni}[1]{\mathit{BSCC}(#1)}
 \withlinksparam{\BSCCuni}
\newcommand{\BSCCact}[1]{#1\cap A}
 \withlinksparam{\BSCCact}
\newcommand{\BSCCstate}[1]{#1\cap S}
 \withlinksparam{\BSCCstate}
\newcommandwl{\lrvname}{\mathit{lv}} %
\newcommandwl{\coBuchi}{\mathsf{coBuchi}}

\newcommandwl{\rvend}{X^{\mathit{end}}}
\newcommandwl{\rvswitch}{X^{\mathit{switch}}}
\newcommandwl{\rvsim}{X^{\mathit{sim}}}
\newcommandwl{\rvoth}{X^{\mathit{oth}}}

\newcommandwl{\counter}{\mathit{cntr}}
\newcommandwl{\rmdef}{\mathop{\stackrel{\mbox{\rm {\tiny def}}}{=}}}

\newcommand{\ceil}[1]{\lceil #1\rceil}
 \withlinksparam{\ceil}
\newcommand{\tp}[1]{\langle #1\rangle}
 \withlinksparam{\tp}
\newcommandwl{\last}{\mathit{last}}
\newcommandwl{\init}{\mathit{init}}

\newcommandwl{\dist}{\mathit{dist}}
\newcommandwl{\St}{\Sigma}
\newcommandwl{\StFL}{\St^{FL}}
\newcommandwl{\StFRM}{\St^{FRM}}
\newcommandwl{\StM}{\St^{M}}
\newcommand{\seq}[1]{\langle{#1}\rangle}
 \withlinksparam{\seq}
\newcommand{\run}[1]{\mathit{Run}(#1)}
 \withlinksparam{\run}
\newcommandwl{\Prb}{\mathbb{P}}
\renewcommand{\Pr}[3]{\Prb^{#1}_{#2}\hspace{-0.16em}\left[{#3}\right]}   %
\newcommandwl{\Exp}{\mathbb{E}}
\newcommandwl{\Var}{\mathbb{V}}
\newcommand{\Va}[3]{\Var^{#1}_{#2}\hspace{-0.16em}\left[{#3}\right]}   %
 \withlinksthreeparam{\Va}
\newcommand{\Ex}[3]{\Exp^{#1}_{#2}\hspace{-0.16em}\left[{#3}\right]}   %
 \withlinksthreeparam{\Ex}
\newcommandwl{\reach}{\mathit{Reach}}
\newcommand{\lrI}[2]{\mathrm{lr}_{\mathrm{inf}}(#1,#2)}  %
 \withlinkstwoparam{\lrI}
\newcommand{\lrS}[2]{\mathrm{lr}_{\mathrm{sup}}(#1,#2)}  %
 \withlinkstwoparam{\lrS}
\newcommand{\lr}[2]{\mathrm{lr}(#1,#2)}  %
 \withlinkstwoparam{\lr}
\newcommand{\lrLim}[1]{\mathrm{lr}(#1)}  %
 \withlinksparam{\lrLim}
\newcommand{\freq}[3]{\mathrm{freq}(#1,#2,#3)}  %
 \withlinksthreeparam{\freq}
\newcommand{\freqI}[3]{\mathrm{freq}_{\mathrm{inf}}(#1,#2,#3)}
 \withlinksthreeparam{\freqI}
\newcommand{\freqS}[3]{\mathrm{freq}_{\mathrm{sup}}(#1,#2,#3)}
 \withlinksthreeparam{\freqS}
\newcommand{\comp}[2]{#1|#2}
 \withlinkstwoparam{\comp}
\newcommand{\src}[1]{\mathit{Src}(#1)}
 \withlinksparam{\src}
\newcommand{\act}[1]{\mathit{Act}(#1)}
 \withlinksparam{\act}
\newcommand{\calF}{\mathcal{F}}
\newcommandwl{\indi}{\mathbf{1}}
\newcommand{\ind}[1]{\indi_{#1}}
 \withlinksparam{\ind}
\newcommand{\Prob}[1]{\mathit{Prob}(#1)}
 \withlinksparam{\Prob}
\newcommandwl{\Mec}{\mathit{MEC}}
\newcommandwl{\Qset}{\mathbb{Q}}
\newcommandwl{\Nset}{\mathbb{N}}
\newcommandwl{\Rset}{\mathbb{R}}
\newcommandwl{\Zset}{\mathbb{Z}}
\newcommandwl{\PSPACE}{\mathbf{PSPACE}}
\newcommandwl{\NP}{\mathbf{NP}}

\newcommand{\tran}[1]{{}\mathchoice%
    {\stackrel{#1}{\rightarrow}}
    {\mathop {\smash\rightarrow}\limits^{\vrule width 0pt height 0pt depth 4pt\smash{#1}}}
    {\stackrel{#1}{\rightarrow}}
    {\stackrel{#1}{\rightarrow}}
{}}

\newcommandwl{\pat}{\omega}
\newcommandwl{\Pat}{\mathit{Runs}}
\newcommandwl{\fpat}{w}
\newcommandwl{\mem}{\mathsf{M}}
\newcommandwl{\Cone}{\mathit{Cone}}
\newcommandwl{\reals}{\mathbb{R}}
\newcommand{\lrIf}[1]{\mathrm{lr}_{\mathrm{inf}}(#1)}  %
 \withlinksparam{\lrIf}
\newcommand{\lrSf}[1]{\mathrm{lr}_{\mathrm{sup}}(#1)}  %
 \withlinksparam{\lrSf}

\newcommandwl{\AcEx}{\mathsf{AcEx}}
\newcommandwl{\AcPr}{\mathsf{AcSt}}

\begin{document}

\title{Trading Performance for Stability in Markov Decision Processes}

\author{\IEEEauthorblockN{Tom\'{a}\v{s} Br\'{a}zdil\IEEEauthorrefmark{1},
Krishnendu~Chatterjee\IEEEauthorrefmark{2},
Vojt\v{e}ch Forejt\IEEEauthorrefmark{3}, and
Anton\'{\i}n Ku\v{c}era\IEEEauthorrefmark{1}}
\IEEEauthorblockA{\IEEEauthorrefmark{1}Faculty of Informatics, Masaryk University \quad (\{xbrazdil,kucera\}@fi.muni.cz)}
\IEEEauthorblockA{\IEEEauthorrefmark{2}IST Austria \quad({krish.chat@gmail.com})}
\IEEEauthorblockA{\IEEEauthorrefmark{3}Department of Computer Science, University of Oxford \quad ({vojfor@cs.ox.ac.uk})}
}

\maketitle
\begin{abstract}
  We study the complexity of central controller synthesis 
problems for finite-state Markov decision processes, where the objective
is to optimize \emph{both} the expected mean-payoff performance of the system 
and its stability.
We argue that the basic theoretical notion of expressing the stability in
terms of the variance of the mean-payoff (called {\em global} variance in our 
paper) is not always sufficient, since it
ignores possible instabilities on respective runs. For this reason we
propose alernative definitions of stability, which we call {\em local}
and {\em hybrid} variance, and which express how rewards on each run deviate from 
the run's own mean-payoff and from the expected mean-payoff, respectively.

We show that a strategy ensuring both the expected mean-payoff and the variance
below given bounds requires randomization and memory, under all the above semantics of
variance. We then look at the problem of determining whether there is a such
a strategy. For the global variance, we show that the problem is in PSPACE, and that the answer can be
approximated in pseudo-polynomial time. For the hybrid variance, the analogous
decision problem is in NP, and a polynomial-time approximating algorithm also exists.
For local variance, we show that the decision problem is in NP.
Since the overall performance can be traded for stability (and vice versa),
we also present algorithms for approximating the associated Pareto curve in 
all the three cases.

Finally, we study a special case of the decision problems, where we require a 
given expected mean-payoff together with {\em zero} variance. Here we show 
that the problems can be all solved in polynomial time.

\end{abstract}

\section{Introduction}
\label{sec-intro}

Markov decision processes (MDPs) are a standard model for stochastic
dynamic optimization. Roughly speaking, an MDP consists of a finite
set of states, where in each state, one of the finitely many actions 
can be chosen by a controller. For every action, there is a fixed  
probability distribution over the states. The execution begins in
some initial state where the controller selects an outgoing action, 
and the system evolves into another state according
to the distribution associated with the chosen action. 
Then, another action is chosen by the controller, 
and so on.  A \emph{strategy} is a recipe for choosing actions. 
In general, a strategy may depend on the execution history
(i.e., actions may be chosen differently when revisiting the same 
state) and the choice of actions can be randomized (i.e., the 
strategy specifies a probability distribution over the available
actions). Fixing a strategy for the controller makes the behaviour
of a given MDP fully probabilistic and determines the usual 
probability space over its \emph{runs}, i.e., infinite
sequences of states and actions.

A fundamental concept of performance and dependability analysis
based on MDP models is \emph{mean-payoff}. Let us assume that every
action is assigned some rational \emph{reward},
which corresponds to some costs (or gains) caused by the action.
The mean-payoff of a given run is then defined as the long-run
average reward per executed action, i.e., the limit of partial
averages computed for longer and longer prefixes of a given run. 
For every strategy~$\sigma$, the overall performance (or throughput) of 
the system controlled by~$\sigma$ then corresponds to the 
expected value of mean-payoff, i.e., the \emph{expected mean-payoff}. 
It is well known (see, e.g., \cite{Puterman:book}) that optimal 
strategies for minimizing/maximizing the expected mean-payoff are
positional (i.e., deterministic and independent of execution history),
and can be computed in polynomial time. However, the quality of
services provided by a given system often depends not only on
its overall performance, but also on its \emph{stability}. 
For example, an optimal controller for a live video streaming 
system may achieve the expected throughput of approximately $2$~MBits/sec.
That is, if a user connects to the server many times, he gets
$2$~Mbits/sec connection \emph{on average}. If an acceptable video
quality requires at least $1.8$~Mbits/sec, the user is also
interested in the likelihood that he gets at least 
$1.8$~Mbits/sec. That is, he requires a certain level of
\emph{overall stability} in service quality, which can be measured 
by the \emph{variance} of mean-payoff, called \emph{global variance} 
in this paper. The basic computational question is \emph{``given
rationals $u$ and $v$, is there a strategy that achieves the
expected mean-payoff $u$ (or better) and variance $v$ (or better)?''}.
Since the expected mean-payoff can be ``traded'' for smaller global 
variance, we are also interested in approximating the associated 
\emph{Pareto curve} consisting of all points $(u,v)$ such that 
(1)~there is a strategy achieving the expected mean-payoff $u$ and
global variance $v$; and (2)~no strategy can improve $u$ or $v$ 
without worsening the other parameter. 

The global variance says how much the actual mean-payoff of a run tends
to deviate from the expected mean-payoff. However, it does not say 
\emph{anything} about the stability of individual runs. To see this, 
consider again the video streaming system example, where we now assume
that although the connection is guaranteed to be fast on average, 
the amount of data delivered per second may change
substantially along the executed run for example due to a faulty network infrastructure.
For simplicity, let us suppose that performing 
one action in the underlying MDP model
takes one second, and the reward assigned to a given action corresponds to 
the amount of transferred data. The above scenario can be modeled by saying that
$6$~Mbits are downloaded every third action, and $0$ Mbits are downloaded
in other time frames. Then the user gets $2$~Mbits/sec connection almost 
surely, but since the individual runs are apparently ``unstable'',
he may still see a lot of stuttering in the video stream.
As an appropriate measure
for the stability of individual runs, we propose \emph{local variance},
which is defined as the long-run average of $(r_i(\pat) - \lra{\pat})^2$,
where $r_i(\pat)$ is the reward of the $i$-th action executed in a run 
$\pat$ and $\lra{\pat}$ is the mean-payoff of~$\pat$. Hence, local variance 
says how much the rewards of the actions executed along a given 
run deviate from the mean-payoff of the run on average. For example,
if the mean-payoff of a run is $2$~Mbits/sec and all of the executed
actions deliver $2$~Mbits, then the run is ``absolutely smooth'' and its
local variance is zero.  The level of ``local stability''
of the whole system (under a given strategy) then corresponds to the
\emph{expected local variance}. The basic algorithmic problem for
local variance is similar to the one for global variance, i.e., 
\emph{``given rationals $u$ and $v$, is there a strategy that achieves the
expected mean-payoff $u$ (or better) and the expected local variance 
$v$ (or better)?''}. We are also interested in the underlying
Pareto curve. 

Observe that the global variance and the expected local variance capture
different and to a large extent \emph{independent} forms of systems' 
(in)stability. Even if the global variance is small, the expected local
variance may be large, and vice versa. In certain situations, we might
wish to minimize \emph{both} of them at the same.
Therefore, we propose 
another notion of \emph{hybrid variance} as a measure for ``combined''
stability of a given system. Technically, the hybrid variance of
a given run $\pat$ is defined as the long-run average of 
$(r_i(\pat) - \Ex{}{}{\lraname})^2$, where $\Ex{}{}{\lraname}$ is the
expected mean-payoff. That is, hybrid variance says how much
the rewards of individual actions executed along a given 
run deviate from the expected mean-payoff on average. The combined
stability of the system then corresponds to the
\emph{expected hybrid variance}. One of the most crucial properties that motivate the
definition of hybrid variance is that the expected hybrid
variance is small iff both the global variance and the expected local
variance are small (in particular, for a prominent class of strategies
the expected hybrid variance is a sum of expected local and global variances). The studied algorithmic problems for hybrid
variance are analogous to the ones for global and local variance.

\smallskip\noindent\textbf{The Results.} 
Our results are as follows:
\begin{enumerate}

\item \emph{(Global variance).} 
The global variance problem was considered before but only under the
restriction of memoryless strategies~\cite{Sobel94}. We first show that in 
general randomized memoryless strategies are not sufficient for Pareto optimal
points for global variance (Example~\ref{ex:global-mem}).
We then establish that 2-memory strategies are sufficient.
We show that the basic algorithmic problem for global variance is
in PSPACE, and the approximate version can be solved in pseudo-polynomial
time.

\item \emph{(Local variance).} 
The local variance problem comes with new conceptual challenges. 
For example, for unichain MDPs, deterministic memoryless strategies are
sufficient for global variance, whereas we show (Example~\ref{ex:local-mem}) 
that even for unichain MDPs both randomization and memory is required for local
variance.
We establish that 3-memory strategies are
sufficient for Pareto optimality for local variance. 
We show that the basic algorithmic problem (and hence also the approximate version)
is in NP.

\item \emph{(Hybrid variance).} 
After defining hybrid variance, we establish that for Pareto optimality 
2-memory strategies are  sufficient, and in general randomized memoryless 
strategies are not.
We show the basic algorithmic problem for hybrid variance is
in NP, and the approximate version can be solved in polynomial time.

\item \emph{(Zero variance).} 
Finally, we consider the problem where the variance is optimized
to zero (as opposed to a given non-negative number in the general case).
In this case, we present polynomial-time algorithms to compute the optimal
mean-payoff that can be ensured with zero variance (if zero variance can 
be ensured) for all the three cases.
The polynomial-time algorithms for zero variance for mean-payoff objectives 
is in sharp contrast to the NP-hardness for cumulative reward MDPs~\cite{ICML2011Mannor_156}.
\end{enumerate}

To prove the above results, one has to overcome various obstacles.
For example, although at multiple places
we build on the techniques of \cite{EKVY:multi-objectives} and \cite{BBCFK:MDP-two-views} which allow us to deal with
maximal end components of an MDP separately, we often need to extend these techniques, since unlike the above
works which study multiple ``independent'' objectives, in the case of global and hybrid variance
any change of value in the expected mean payoff implies a change of value of the variance.
Also, since we do not
impose any restrictions on the structure of the strategies, we cannot even assume
that the limits defining the mean-payoff and the respective variances exist; this becomes
most apparent in the case of local and hybrid variance, where we need to rely on delicate
techniques of selecting runs from which the limits can be extracted.
Another complication is that while most of the work on multi-objective verification deals
with objective functions which are linear, our objective functions are inherently quadratic due to the definition of variance.

The summary of our results is presented in Table~\ref{table:results}.
A simple consequence of our results is that the Pareto curves can be approximated in pseudo-polynomial
time in the case of global and hybrid variance, and in exponential time for local variance.

\noindent
\begin{table*}
\begin{center}
\begin{tabular}{|l|l|l|l|l|}
\hline
& Memory size & Complexity & Approx. complexity & Zero-var. complexity \\
\hline
\hline
Global & 2-memory & PSPACE (Theorem~\ref{thm:global}) & Pseudo-polynomial (Theorem~\ref{thm:global}) & PTIME (Theorem~\ref{thm:zero}) \\
&  LB: Example~\ref{ex:global-mem}, UB: Theorem~\ref{thm:global}  &  &  & \\
\hline
Local & LB: 2-memory (Example~\ref{ex:local-mem}) & NP (Theorem~\ref{thm:local-np-alg}) & NP & PTIME (Theorem~\ref{thm:zero})\\
& UB: 3-memory (Theorem~\ref{thm:local-np-alg}) & & & \\
\hline
Hybrid & 2-memory & NP (Theorem~\ref{thm:hybrid}) & PTIME (Theorem~\ref{thm:hybrid}) & Quadratic (Theorem~\ref{thm:zero})\\
 & LB: Example~\ref{ex:hybrid-mem-ran}, UB: Theorem~\ref{thm:hybrid} & & & \\
\hline
\end{tabular}
\end{center}
\caption{Summary of the results, where LB and UB denotes lower- and upper-bound, respectively.\label{table:results}}
\end{table*}

\noindent\textbf{Related Work.} %
Studying the trade-off between multiple objectives in an MDP
has attracted significant attention in the recent years 
(see~\cite{Altman} for overview).
In the verification area,
MDPs with multiple mean-payoff objectives~\cite{BBCFK:MDP-two-views}, 
discounted objectives~\cite{CHMH:multi-objectives}, 
cumulative reward objectives~\cite{FKP12}, and 
multiple $\omega$-regular objectives~\cite{EKVY:multi-objectives} 
have been studied.
As for the stability of a system,
the variance penalized mean-payoff problem (where the mean-payoff is penalized by a 
constant times the variance) under memoryless (stationary) strategies was studied in~\cite{FKL89}.
The mean-payoff variance trade-off problem for unichain MDPs was considered
in~\cite{Chung94}, where a solution using quadratic programming was designed;
under memoryless (stationary) strategies the problem was considered in~\cite{Sobel94}.
All the above works for mean-payoff variance trade-off consider the global 
variance, and are restricted to memoryless strategies.
The problem for 
general strategies and global variance was not solved before.
Although restrictions to unichains or memoryless strategies
are feasible in some areas, many systems modelled as MDPs might require more general approach. For example,
a decision of a strategy to shut the system down might make it impossible to return the running state again, yielding
in a non-unichain MDP.
Similarly, it is natural to synthesise strategies that change their decisions over time.

As regards other types of objectives, no work considers the local and hybrid variance problems.
The variance problem for {\em discounted} reward MDPs was studied in~\cite{Sobel82}.
The trade-off of expected value and variance of {\em cumulative} reward in MDPs was studied 
in~\cite{ICML2011Mannor_156}, showing the zero variance problem to be NP-hard.
This contrasts with our results, since in our setting we present polynomial-time algorithms for zero variance.

\section{Preliminaries}
\label{sec-prelim}

We use $\Nset$, $\Zset$, $\Qset$, and $\Rset$ to denote the sets of
positive integers, integers, rational numbers, and real numbers,
respectively. 
We assume familiarity with basic notions of probability
theory, e.g., \emph{probability space}, \emph{random variable}, or 
\emph{expected value}.
As usual, a \emph{probability distribution} over a finite or 
countable set $X$ is a function
$f : X \rightarrow [0,1]$ such that \mbox{$\sum_{x \in X} f(x) = 1$}. 
We call $f$ \emph{positive} if 
$f(x) > 0$ for every $x \in X$, \emph{rational} if $f(x) \in
\Qset$ for every $x \in X$, and \emph{Dirac} if $f(x) = 1$ for some 
$x \in X$. The set of all distributions over $X$ is denoted by
$\dist(X)$.

For our purposes, a \emph{Markov chain} is a triple 
\mbox{$M = (L,\tran{},\mu)$} where $L$ is a finite or countably 
infinite set of \emph{locations}, 
\mbox{${\tran{}} \subseteq L \times (0,1] \times L$} is a 
\emph{transition relation}
such that for each fixed $\ell \in L$, $\sum_{\ell \tran{x} \ell'} x = 1$, and
$\mu$ is the \emph{initial probability distribution} on~$L$.
A \emph{run} in $M$ is an infinite sequence $\pat = \ell_1 \ell_2 \ldots$
of locations such that $\ell_i \tran{x} \ell_{i{+}1}$ for every $i \in \Nset$.
A \emph{finite path} in $M$ is a finite prefix of a run. Each finite
path $\fpat$ in $M$ determines the set $\Cone(\fpat)$ consisting of
all runs that start with $\fpat$. To $M$ we associate the probability 
space $(\Pat_M,\calF,\mathbb{P})$, where $\Pat_M$ is the set of all 
runs in $M$, $\calF$ is the $\sigma$-field generated by all $\Cone(\fpat)$ for finite paths $\fpat$,
and $\mathbb{P}$ is the unique probability measure such that
$\mathbb{P}(\Cone(\ell_1,\ldots,\ell_k)) = 
\mu(\ell_1) \cdot \prod_{i=1}^{k-1} x_i$, where
$\ell_i \tran{x_i} \ell_{i+1}$ for all \mbox{$1 \leq i < k$} (the empty
product is equal to $1$).

\smallskip\noindent{\bf Markov decision processes.} 
A \emph{Markov decision process} (MDP) is a tuple $G=(S,A,\mathit{Act},\delta)$ 
where $S$ is a \emph{finite} set of states, $A$ is a \emph{finite} set 
of actions, $\mathit{Act} : S\rightarrow 2^A\setminus \{\emptyset\}$ is 
an action enabledness 
function that assigns to each state $s$ the set $\act{s}$ of actions enabled 
at $s$, and $\delta : S\times A\rightarrow \dist(S)$ is a probabilistic 
transition function  that given a state $s$ and an action 
$a \in \act{s}$ enabled at $s$ gives a probability distribution over the 
successor states.
For simplicity, we assume that every action is enabled in exactly one state, 
and we denote this state $\src{a}$. Thus, henceforth we will assume that 
$\delta: A \rightarrow \dist(S)$.

A \emph{run} in $G$ is an infinite alternating sequence of states
and actions $\pat=s_1 a_1 s_2 a_2\ldots$
such that for all $i \geq 1$, $\src{a_i}=s_i$ and
$\delta(a_i)(s_{i+1}) > 0$. 
We denote by $\Pat_G$ the set of all runs in~$G$.
A \emph{finite path} of length~$k$ in~$G$ is a finite prefix
$\fpat = s_1 a_1\ldots a_{k-1} s_k$ of a run, and we use
$\last(\fpat)=s_k$ for the last state of $\fpat$.
Given a run $\pat \in \Pat_G$,  we denote by $A_i(\pat)$ the 
$i$-th action $a_i$ of~$\pat$.

A pair $(T,B)$ with $\emptyset\neq T\subseteq S$ and $B\subseteq \bigcup_{t\in T}\act{t}$
is an \emph{end component} of $G$
if (1) for all $a\in B$, if $\delta(a)(s')>0$ then $s'\in T$;
and (2) for all $s,t\in T$ there is a finite path 
$\fpat = s_1 a_1\ldots a_{k-1} s_k$ such that $s_1 = s$, $s_k=t$, and all states
and actions that appear in $\fpat$ belong to $T$ and $B$, respectively.
An end component $(T,B)$ is a \emph{maximal end component (MEC)}
if it is maximal wrt.\ pointwise subset ordering. The set of all MECs of
$G$ is denoted by $\Mec(G)$. Given an end component
$C=(T,B)$, we sometimes abuse notation by considering $C$ as the disjoint
union of $T$ and $B$ (for example, we write $S\cap C$ to denote the set $T$).
For a given $C \in \Mec(G)$, we use $\staymec{C}$ to denote the set of all runs 
$\pat =s_1 a_1 s_2 a_2\ldots$ that eventually \emph{stay} in $C$, i.e., 
there is $k \in \Nset$ such that for all $k' \geq k$ we have that
$s_{k'},a_{k'} \in C$.

\smallskip\noindent{\bf Strategies and plays.} 
Intuitively, a strategy in an MDP $G$ is a ``recipe'' to choose actions.
Usually, a strategy is formally defined as a function 
$\sigma : (SA)^*S \to \dist(A)$ that given a finite path~$\fpat$, representing 
the execution history, gives a probability distribution over the 
actions enabled in~$\last(\fpat)$. In this paper we adopt a definition which is equivalent to the standard
one, but more convenient for our purpose.
Let $\mem$ be a finite or countably
infinite set of \emph{memory elements}. A \emph{strategy}
is a triple
$\sigma = (\sigma_u,\sigma_n,\alpha)$, where 
$\sigma_u: A\times S \times \mem \to \dist(\mem)$ and 
$\sigma_n: S \times \mem \to \dist(A)$ are \emph{memory update}
and \emph{next move} functions, respectively, and $\alpha$ is
an initial distribution on memory elements. We require that for 
all $(s,m) \in S \times \mem$, the distribution $\sigma_n(s,m)$ assigns a
positive value only to actions enabled at~$s$. The set of all
strategies is denoted by $\St$ (the underlying MDP~$G$ will be always
clear from the context).

A \emph{play} of $G$ determined by an initial state
$s\in S$ and a strategy $\sigma$ is a Markov chain 
$G^\sigma_s$ (or $G^\sigma$ if $s$ is clear from the context)
where the set of locations is $S \times \mem \times A$,
the initial distribution $\mu$ is positive only on (some) elements 
of $\{s\} \times  \mem \times A$ where
$\mu(s,m,a) = \alpha(m) \cdot \sigma_n(s,m)(a)$, 
and $(t,m,a) \tran{x} (t',m',a')$ iff 
$
   x  = 
   \delta(a)(t')  \cdot \sigma_u(a,t',m)(m') \cdot \sigma_n(t',m')(a')
      >   0
$.
Hence, $G^\sigma_s$ starts in a location chosen randomly according
to $\alpha$ and $\sigma_n$. In a current location $(t,m,a)$, 
the next action to be performed is $a$, hence the probability of entering
$t'$ is $\delta(a)(t')$. The probability of updating the memory to $m'$
is $\sigma_u(a,t',m)(m')$, and the probability of selecting $a'$ as 
the next action is $\sigma_n(t',m')(a')$. Since these choices
are independent (in the probability theory sense), we obtain the product above. 

Note that every run in $G^\sigma_s$ determines a unique run in~$G$.
Hence, every notion originally defined for the runs in~$G$ can
also be used for the runs in $G^\sigma_s$, and we use this fact
implicitly at many places in this paper. For example, we use the 
symbol $\staymec{C}$ to denote the set of all runs in~$G^\sigma_s$ that 
eventually stay in~$C$, certain functions originally defined over $\Pat_G$
are interpreted as random variables over the runs in~$G^\sigma_s$, etc.

\smallskip\noindent{\bf Strategy types.}
In general, a strategy may use infinite memory, and both 
$\sigma_u$ and $\sigma_n$ may randomize. A strategy is 
\emph{pure} (or \emph{deterministic}) if $\alpha$ is Dirac and
  both the memory update and the next move functions give a Dirac
  distribution for every argument, and
 \emph{stochastic-update} if $\alpha$, $\sigma_u$, and $\sigma_n$
  are unrestricted.
Note that every pure strategy is 
stochastic-update. 
A \emph{randomized} strategy is a strategy which is not necessarily pure.
We also classify the strategies according to the size of memory
they use. Important subclasses are 
\emph{memoryless} strategies, in which $\mem$ is a singleton,
\emph{$n$-memory} strategies, in which $\mem$ has exactly $n$~elements, and
\emph{finite-memory} strategies, in which $\mem$ is finite.

 For a finite-memory strategy $\sigma$, a 
 \emph{bottom strongly connected component} (BSCC) of $G^\sigma_s$ is 
 a subset of locations \mbox{$W \subseteq S\times \mem \times A$} 
 such that for all $\ell_1 \in W$ and \mbox{$\ell_2 \in S\times \mem\times A$}
 we have that (i) if $\ell_2$ is reachable from $\ell_1$, then 
 $\ell_2 \in W$, and (ii) for all $\ell_1,\ell_2 \in W$ we have that
 $\ell_2$ is reachable from $\ell_1$.
 Every BSCC $W$ determines a unique end component 
 $(\{s\mid (s,m,a)\in W\},\{a\mid (s,m,a)\in W\})$,
 and we sometimes do not distinguish between $W$ 
 and its associated end component.

An MDP is {\em strongly connected} if all its states form a
single (maximal) end component.
A strongly connected  MDP is a {\em unichain}
if for all end components $(T,B)$ we have $T=S$.

Throughout this paper we will use the following standard result
about MECs.

\begin{lemma}[{\cite[Proposition~3.1]{CY:MDP-regular-TAC}}]\label{lemma:stay-mec}
Almost all runs eventually end in a MEC, i.e. $\Pr{\sigma}{s}{\bigcup_{C\in Mec(G)} \staymec{C}}=1$ for all $\sigma$ and $s$.
\end{lemma}

\smallskip\noindent{\bf Global, local, and hybrid variance.}
Let $G=(S,A,\mathit{Act},\delta)$ be an MDP, and $r : A \to \Qset$ a
\emph{reward function}.  We define the \emph{mean-payoff} of a run 
$\pat \in \Pat_G$ by
\[
\lra{\pat} = \limsup_{n\rightarrow \infty} \frac{1}{n}\sum_{i=0}^{n-1} r(A_i(\pat)) \,.
\]
The expected value and variance of $\lraname$ in~$G^\sigma_s$ are denoted by
$\Ex{\sigma}{s}{\lraname}$ and $\Va{\sigma}{s}{\lraname}$, respectively
(recall that $\Va{\sigma}{s}{\lraname} =  
\Ex{\sigma}{s}{(\lraname - \Ex{\sigma}{s}{\lraname})^2} =
\Ex{\sigma}{s}{\lraname^2} - (\Ex{\sigma}{s}{\lraname})^2$).
Intuitively, $\Ex{\sigma}{s}{\lraname}$ corresponds to the ``overall
performance'' of~$G^\sigma_s$, and $\Va{\sigma}{s}{\lraname}$ is a measure
of ``global stability'' of~$G^\sigma_s$ indicating how
much the mean payoffs of runs in~$G^\sigma_s$ tend to deviate from
$\Ex{\sigma}{s}{\lraname}$ (see Section~\ref{sec-intro}).
In the rest of this paper, we refer to $\Va{\sigma}{s}{\lraname}$ 
as \emph{global variance}.

The stability of a given run $\pat \in \Pat_G$ (see Section~\ref{sec-intro})
is measured by its \emph{local variance} defined as follows:
\[
  \lrvl{\pat}  =
  \limsup_{n\rightarrow \infty} 
     \frac{1}{n}\sum_{i=0}^{n-1} \big( r(A_i(\pat)) - \lra{\pat} \big)^2
\]
Note that $\lrvl{\pat}$ is not really a ``variance'' in the usual sense
of probability theory\footnote{By investing some effort, one could perhaps 
find a random variable $X$ such that $\lrvl{\pat}$ is the variance of~$X$,
but this question is not really relevant---we only use $\lrvlname$ 
as a \emph{random variable which measures the level of local stability 
of runs}. One could perhaps study the variance of $\lrvlname$, but this
is beyond the scope of this paper. The same applies to the function
$\lrvhname$.}.
We call the function $\lrvl{\pat}$  ``local variance'' 
because we find this name suggestive; $\lrvl{\pat}$ is the long-run 
average square of the distance 
from~$\lra{\pat}$. The expected value of $\lrvlname$ in~$G^\sigma_s$ 
is denoted by $\Ex{\sigma}{s}{\lrvlname}$.

Finally, given a run $\pat$ in~$G^\sigma_s$, we define the 
\emph{hybrid variance} of $\omega$ in~$G^\sigma_s$ as follows:
\[
  \lrvhname(\pat)  =
  \limsup_{n\rightarrow \infty} 
     \frac{1}{n}\sum_{i=0}^{n-1} \left( r(A_i(\pat)) - 
     \Ex{\sigma}{s}{\lraname} \right)^2
\] 
Note that the definition of $\lrvhname(\pat)$ depends on the expected
mean payoff, and hence it makes sense only after fixing a strategy
$\sigma$ and an initial state~$s$. Sometimes we also write
$\lrvh{\sigma}{s}{\pat}$ instead of $\lrvhname(\pat)$ to prevent
confusions about the underlying $\sigma$ and~$s$. The expected value
of $\lrvhname$ in~$G^\sigma_s$ is denoted by
$\Ex{\sigma}{s}{\lrvhname}$. Intuitively, $\Ex{\sigma}{s}{\lrvhname}$
measures the ``combined'' stability of $G^\sigma_s$ (see 
Section~\ref{sec-intro}). 

\smallskip\noindent{\bf Pareto optimality.}
We say that a strategy $\sigma$ is \emph{Pareto optimal} in~$s$
wrt.\ global variance if for every 
strategy $\zeta$ we have that 
$(\Ex{\sigma}{s}{\lraname},\Va{\sigma}{s}{\lraname}) \ge
 (\Ex{\zeta}{s}{\lraname},\Va{\zeta}{s}{\lraname})$
implies $(\Ex{\sigma}{s}{\lraname},\Va{\sigma}{s}{\lraname}) =
 (\Ex{\zeta}{s}{\lraname},\Va{\zeta}{s}{\lraname})$, where
$\ge$ is the standard component-wise ordering.
Similarly,
we define Pareto optimality of $\sigma$ wrt.\ local and hybrid variance
by replacing $\Va{\alpha}{s}{\lraname}$ with  $\Ex{\alpha}{s}{\lrvlname}$
and $\Ex{\alpha}{s}{\lrvhname}$, respectively.
We choose the order $\ge$ for technical convenience,
if one wishes to maximize the expected value while minimizing the variance, it suffices to multiply all rewards by $-1$.
The \emph{Pareto curve}
for $s$ wrt.\ global, local, and hybrid variance consists of
all points of the form
$(\Ex{\sigma}{s}{\lraname},\Va{\sigma}{s}{\lraname})$, 
$(\Ex{\sigma}{s}{\lraname},\Ex{\sigma}{s}{\lrvlname})$, and
$(\Ex{\sigma}{s}{\lraname},\Ex{\sigma}{s}{\lrvhname})$, where $\sigma$
is a Pareto optimal strategy wrt.\  global, local, and hybrid variance,
respectively.

\smallskip\noindent{\bf Frequency functions.}
Let $C$ be a MEC. We say that $f:\MECact{C}\rightarrow [0,1]$ is a {\em frequency function on $C$} if
\begin{itemize}
\item $\sum_{a\in \MECact{C}} f(a)=1$
\item $\sum_{a\in \MECact{C}} f(a)\cdot \delta(a)(s)=\sum_{a\in \mathit{Act}(s)} f(a)$ for every $s\in \MECstate{C}$
\end{itemize}
Define $\lraname[f]:=\sum_{a\in C} f(a)\cdot r(a)$ and $\lrvname[f]:=\sum_{a\in C} f(a)\cdot (r(a)-\lraname[f])^2$.

\smallskip\noindent{\bf The studied problems.}
In this paper, we study the following basic problems connected
to the three stability measures introduced above (below $V_s^\sigma$
is either
$\Va{\sigma}{s}{\lraname}$, $\Ex{\sigma}{s}{\lrvlname}$, or $\Ex{\sigma}{s}{\lrvhname}$):
\begin{itemize}
\item \emph{Pareto optimal strategies and their memory}. Do Pareto optimal strategies exist for all points on the Pareto curve?
 Do Pareto optimal strategies require memory and randomization in general? Do strategies achieving non-Pareto points require memory and randomization in general?
\item \emph{Deciding strategy existence}.
  For a given MDP $G$, an initial
  state~$s$, a rational reward function $r$, and a point $(u,v) \in
  \Qset^2$, we ask whether there exists a strategy $\sigma$ such that
  $(\Ex{\sigma}{s}{\lraname},V_s^\sigma) \leq (u,v)$.
\item \emph{Approximation of strategy existence}. 
  For a given MDP $G$, an initial
  state~$s$, a rational reward function $r$, a number $\varepsilon$ and a point $(u,v) \in
  \Qset^2$, we want to get an algorithm which (a) outputs ``yes'' if there is a strategy $\sigma$ such that
  $(\Ex{\sigma}{s}{\lraname},V_s^\sigma) \leq (u -\varepsilon,v - \varepsilon)$;
  (b) outputs ``no'' if there is no strategy such that
  $(\Ex{\sigma}{s}{\lraname},V_s^\sigma) \leq (u,v)$.
\item \emph{Strategy synthesis.} If there exists a strategy  $\sigma$ such that
  $(\Ex{\sigma}{s}{\lraname},V_s^\sigma) \leq (u,v)$,
  we wish to \emph{compute} such strategy. Note that it is 
  not \emph{a priori} clear that $\sigma$ is finitely representable,
  and hence we also need to answer the question what \emph{type}
  of strategies is needed to achieve Pareto optimal points.
\item \emph{Optimal performance with zero-variance.} 
  Here we are interested in deciding if there exists a Pareto point
  of the form $(u,0)$ and computing the value of $u$,
  i.e., the optimal expected mean payoff
  achievable with ``absolute stability'' (note that the variance is always non-negative
  and its value $0$ corresponds to stable behaviours).
\end{itemize}

\begin{remark}\label{remark1}
If the approximation of strategy existence problem is decidable, 
we design the following algorithm to approximate the Pareto curve up to an arbitrarily
small given $\varepsilon > 0$. We compute a finite
set of points $P \subseteq \Qset^2$ such that (1) for every Pareto point
$(u,v)$ there is $(u',v') \in P$ with 
$(|u-u'|,|v-v'|) \leq (\varepsilon,\varepsilon)$, and (2) for every
$(u',v') \in P$ there is a Pareto point $(u,v)$ such that
$(|u-u'|,|v-v'|) \leq (\varepsilon,\varepsilon)$. Let 
$R = \max_{a \in A} |r(a)|$. Note that $|\Ex{\sigma}{s}{\lraname}| \leq R$
and $V_s^\sigma \leq R^2$ for an arbitrary strategy $\sigma$.
Hence, the set $P$ is computable by a naive algorithm which decides the 
approximation of strategy existence for $\mathcal{O}(|R|^3/\varepsilon^2)$ 
points in the corresponding $\varepsilon$-grid and puts
$\mathcal{O}(|R|^2/\varepsilon)$ points into~$P$. The question
whether the three Pareto curves can be approximated more efficiently by 
sophisticated methods based on deeper analysis of their properties
is left for future work.
\end{remark}

\section{Global variance}
\label{sec-global}

In the rest of this paper, unless specified otherwise, we suppose we work with a fixed MDP $G=(S,A,\mathit{Act},\delta)$
and a reward function $r : A \to \Qset$.
We start by proving that  both memory and randomization is needed
even for achieving non-Pareto points; this implies that memory and randomization is needed even to approximate
the value of Pareto points. Then we show that
2-memory stochastic update strategies are sufficient, which gives a tight
bound. 
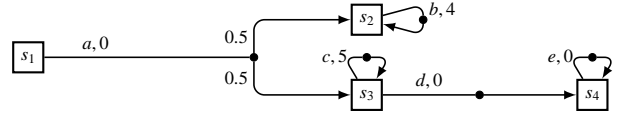
\begin{figure}[t]
\centering
\begin{tikzpicture}[every node/.style={inner sep=.7mm},x=1.5cm,y=1cm,font=\scriptsize]
\node[state] at (-1,0) (s1) {$s_1$};
\node[distr] at (1,0) (s1a) {};
\node[state] at (2,0.5) (s2) {$s_2$};
\node[state] at (2,-0.5) (s3) {$s_3$};
\node[distr] at (2.5,0.5) (s2a) {};
\node[distr] at (2,0) (s3a) {};
\node[distr] at (3,-0.5) (s3b) {};
\node[state] at (4,-0.5) (s4) {$s_4$};
\node[distr] at (4,0) (s4a) {};

\draw[trarr] (s1) -- node[pos=0.25,above] {$a,0$} (s1a)|- node[pos=0.25,left] {$0.5$} (s2);
\draw[trarr] (s1a) |- node[pos=0.25,left] {$0.5$} (s3);
\draw[trarr] (s2) -- +(0.5,0.2) -- node[pos=0.5,right] {$b,4$} (s2a) -- +(0,-0.2) -- (s2);
\draw[trarr] (s3) -- +(-0.2,0.5) -- node[pos=0.5,left] {$c,5$} (s3a) -- +(0.2,0) -- (s3);
\draw[trarr] (s4) -- +(-0.2,0.5) -- node[pos=0.5,left] {$e,0$}(s4a) -- +(0.2,0) -- (s4);
\draw[trarr] (s3) --  node[pos=0.5,above] {$d,0$} (s3b) -- (s4);
\end{tikzpicture}
\caption{An MDP witnessing the need for memory and randomization in Pareto
optimal strategies for global variance.}
\label{fig-MDP-mem-rand}  
\end{figure}

\begin{example}\label{ex:global-mem}
Consider the MDP of Fig.~\ref{fig-MDP-mem-rand}. 
Observe that the point $(4,2)$ is achievable by a strategy 
$\sigma$ which selects 
$c$ with probability $\frac{4}{5}$ and $d$ with probability 
$\frac{1}{5}$ upon the \emph{first} visit to $s_3$; 
in every other visit
to $s_3$, the strategy $\sigma$ selects $c$ with probability~$1$. Hence,
$\sigma$ is a 2-memory randomized strategy which stays in MEC
$C=(\{s_3\},\{c\})$ with probability 
$\frac{1}{2}\cdot \frac{4}{5} = \frac{2}{5}$. Clearly, 
$\Ex{\sigma}{s_1}{\lraname} = \frac{1}{2} \cdot 4 + 
\frac{1}{2}\cdot \frac{4}{5} \cdot 5 + \frac{1}{2}\cdot \frac{1}{5} \cdot 0=
4$ and \mbox{$\Va{\sigma}{s_1}{\lraname} = \frac{1}{2} \cdot 4^2 + 
\frac{1}{2}\cdot \frac{4}{5} \cdot 5^2 + \frac{1}{2}\cdot \frac{1}{5} \cdot 
0^2 - 4^2 = 2$}. Further, note that every strategy $\bar\sigma$ which stays
in $C$ with probability $x$ satisfies
$\Ex{\bar\sigma}{s_1}{\lraname} = \frac{1}{2}\cdot 4 + x\cdot 5$
and
$\Va{\bar\sigma}{s_1}{\lraname} = \frac{1}{2} \cdot 4^2 + x\cdot 5^2 - (2 + x\cdot 5)^2$.
For $x>\frac{2}{5}$ we get $\Ex{\bar\sigma}{s_1}{\lraname} > 4$,
and for $x< \frac{2}{5}$ we get $\Va{\bar\sigma}{s_1}{\lraname} > 2$,
so $(4,2)$ is indeed 
a Pareto point. Every deterministic (resp. memoryless) strategy
can stay in $C$ with probability either $\frac{1}{2}$ or $0$, giving $\Ex{\bar\sigma}{s_1}{\lraname}=\frac{9}{2}$
or $\Va{\bar\sigma}{s_1}{\lraname}=4$.
So, both memory and randomization are needed to achieve the Pareto point $(4,2)$ or a non-Pareto point $(4.1, 2.1)$.
\end{example}

Interestingly, if the MDP is strongly connected, memoryless deterministic strategies
always suffice, because in this case a memoryless strategy that minimizes the expected mean payoff
immediately gets zero variance. This is in contrast with local and hybrid variance, where
we will show that memory and randomization is required in general already for 
unichain MDPs.
For the general case of global variance, the sufficiency of 2-memory 
strategies is captured by the following theorem.

\begin{theorem}\label{thm:global}
  If there is a strategy $\zeta$ satisfying
  $(\Ex{\zeta}{s}{\lraname}, \Va{\zeta}{s}{\lraname})
     \leq (u,v)$,
  then there is a 2-memory strategy with the same properties.
  Moreover, Pareto optimal strategies always exist, the problem whether there is a strategy achieving a point $(u,v)$ is in PSPACE,
  and approximation of the answer can be done in pseudo-polynomial time.
\end{theorem}

Note that every $C \in \Mec(G)$ can be seen as a strongly connected
MDP. By using standard linear programming methods (see, e.g., 
\cite{Puterman:book}), for every $C \in \Mec(G)$ we can compute the 
\emph{minimal} and the \emph{maximal} expected mean payoff achievable 
in $C$, denoted by $\alpha_C$ and $\beta_C$, in polynomial time (since
$C$ is strongly connected, the choice of initial state is irrelevant). 
Thus, we can also compute the system $L$ of Fig.~\ref{fig-L} in
polynomial time. We show the following:

\begin{figure}[t]
\centering
  \begin{align}\small
   \lefteqn{\!\!\!\!\!\!\!\!\!\!\!\mathbf{1}_{s}(t) + \sum_{a\in A} y_{a}\cdot \delta(a)(t) =
   \sum_{a\in \mathit{Act}(t)} y_{a} + y_t \hspace*{1em} 
        \text{ for all $t\in S$}\label{eq:gl:ya}}\\
   \sum_{\substack{C\in \Mec(G)\\t \in S\cap C}}\! y_t & =  1\label{eq:gl:yt}\\
   y_\kappa & \geq  0 \hspace*{4.5em} \text{for all $\kappa \in S \cup A$}\label{eq:gl:pos}\\
   \alpha_C & \le  x_C \hspace*{4em}
        \text{for all $C \in \Mec(G)$}\label{eq:gl:min}\\
   x_C & \le  \beta_C \hspace*{4em}
   \text{for all $C \in \Mec(G)$}\label{eq:gl:max}\\
   u & \ge  \sum_{C \in \Mec(G)} x_C \cdot \sum_{t\in S\cap C} y_t \label{eq:gl:exp}\\
   v & \ge  \Big(\!\sum_{C \in \Mec(G)}\!\!\! x_C^2 {\cdot} 
                     \sum_{t\in S\cap C}y_t\Big) - 
    \Big(\!\sum_{C \in \Mec(G)}\!\! x_C {\cdot} 
                     \sum_{t \in S \cap C}y_t\smash{\Big)^2} \label{eq:gl:var}
\end{align}
\caption{The system $L$. (Here $\mathbf{1}_{s_0}(s)=1$ 
if $s = s_0$, and $\mathbf{1}_{s_0}(s)=0$ otherwise.)}
\label{fig-L}
\end{figure}

\begin{proposition}\label{prop:global-lp}
Let $s \in S$ and $u,v\in \Rset$. 
\begin{enumerate}
\item \label{thm:global-main-stop} 
  If there is a strategy $\zeta$ satisfying
  $(\Ex{\zeta}{s}{\lraname}, \Va{\zeta}{s}{\lraname})
     \leq (u,v)$
  then the system $L$ of Fig.~\ref{fig-L} has a solution.
\item \label{thm:global-main-ptos} 
  If the system $L$ of Fig.~\ref{fig-L} has a solution, 
  then there exist a 2-memory stochastic-update strategy $\sigma$ 
  and $z \in \Rset$ such that 
  $(\Ex{\sigma}{s}{\lraname}, \Va{\sigma}{s}{\lraname}) \leq (u,v)$
  and for every $C \in \Mec(G)$ we have the following: If 
  $\alpha_C > z$, then $x_C = \alpha_C$; if $\beta_C < z$, then $x_C = \beta_C$;
  otherwise (i.e., if $\alpha_C \leq z \leq \beta_C$) $x_C = z$.
\end{enumerate} 
\end{proposition}
Observe that the existence of Pareto optimal strategies follows from the above proposition,
since we define points $(u,v)$ that some strategy can achieve by a continous function
from values $x_C$ and $\sum_{t\in S\cap C} y_t$ for $C\in\Mec(G)$ to $\Rset^2$. Because the domain is bounded (all $x_C$
and $\sum_{t\in S\cap C} y_t$ have minimal and maximal values they can achieve) and closed (the points of the domain
are expressible as a projection of feasible solutions of a linear program), it is also compact, and a continuous
map of a compact set is compact \cite{Royden88}, and hence closed.

Let us briefly sketch the proof of Proposition~\ref{prop:global-lp}, which
combines new techniques with results of \cite{BBCFK:MDP-two-views,%
EKVY:multi-objectives}. 
We start with Item~1. Let $\zeta$ be a strategy satisfying  
$(\Ex{\zeta}{s}{\lraname}, \Va{\zeta}{s}{\lraname}) \leq (u,v)$.
First, note that almost every run of $G^\zeta_s$ eventually 
stays in some MEC of $G$ by Lemma~\ref{lemma:stay-mec}.
The way how $\zeta$ determines the values of all $y_\kappa$, where
$\kappa \in S \cup A$, is exactly the same as in
\cite{BBCFK:MDP-two-views} and it is based on the ideas of
\cite{EKVY:multi-objectives}. The details are given in 
\ifthenelse{\isundefined{\techreport}}{\cite{techrepLics13}}{Appendix~\ref{app-global-stay-mecs}}.
The important property preserved is that for every $C \in \Mec(G)$
and every state $t \in S \cap C$, the value of $y_t$ corresponds
to the probability that a run stays in $C$ and enters $C$ via
the state $t$. Hence, $\sum_{t \in S\cap C} y_t$ is the probability 
that a run of $G^\zeta_s$ eventually stays in~$C$. The way how $\zeta$
determines the value of $y_a$, where $a \in A$, is explained
in
\ifthenelse{\isundefined{\techreport}}{\cite{techrepLics13}}{Appendix~\ref{app-global-stay-mecs}}.
The value of
$x_C$ is the conditional expected mean payoff under the condition
that a run stays in~$C$, i.e., $x_C = \Ex{\zeta}{s}{\lraname \mid R_C}$.
Hence, $\alpha_C \leq x_C \leq \beta_C$, which means that 
\eqref{eq:gl:min}~and~\eqref{eq:gl:max} are satisfied.
Further, $\Ex{\zeta}{s}{\lraname} = 
\sum_{C \in \Mec(G)} x_C \cdot \sum_{t\in S\cap C} y_t$, and 
hence~\eqref{eq:gl:exp} holds. Note that $\Va{\zeta}{s}{\lraname}$
is \emph{not} necessarily equal to the right-hand side of~\eqref{eq:gl:var},
and hence it is not immediately clear why~\eqref{eq:gl:var} should hold.  
Here we need the following lemma (a proof is given in 
\ifthenelse{\isundefined{\techreport}}{\cite{techrepLics13}}{Appendix~\ref{app-global-x_c-almost-surely}}):
\begin{lemma}
\label{lem-x_c-almost-surely}
  Let $C \in \Mec(G)$, and let $z_C \in [\alpha_C,\beta_C]$. Then there exists
  a memoryless randomized strategy  $\sigma_{z_C}$ such that for every
  state $t \in C\cap S$ we have that $\Pr{\sigma_{z_C}}{t}{\lraname {=} z_C} = 1$.
\end{lemma}
Using Lemma~\ref{lem-x_c-almost-surely}, we can define another strategy
$\zeta'$ from $\zeta$ such that for every $C \in \Mec(G)$ we have the following: 
(1) the probability of $R_C$ in $G^\zeta_s$ and in  $G^{\zeta'}_s$ is 
the same; (2) almost all runs $\pat \in R_C$ satisfy $\lraname(\pat) = x_C$. 
This means that $\Ex{\zeta}{s}{\lraname} = \Ex{\zeta'}{s}{\lraname}$,
and we show that 
$\Va{\zeta}{s}{\lraname} \geq \Va{\zeta'}{s}{\lraname}$ (see
\ifthenelse{\isundefined{\techreport}}{\cite{techrepLics13}}{Appendix~\ref{app-global-variance-improves}}).
Hence, $(\Ex{\zeta'}{s}{\lraname}, \Va{\zeta'}{s}{\lraname})
\leq (u,v)$, and therefore \eqref{eq:gl:ya}--\eqref{eq:gl:exp} also 
hold if we use $\zeta'$ instead of $\zeta$ to determine the values 
of all variables. Further, the right-hand side of~\eqref{eq:gl:var} is 
equal to $\Va{\zeta'}{s}{\lraname}$, and hence~\eqref{eq:gl:var} holds.
This completes the proof of Item 1.

Item~2 is proved as follows. Let $y_\kappa$, where
$\kappa \in S \cup A$, and $x_C$, where $C \in \Mec(G)$, be
a solution of~$L$. For every $C \in \Mec(G)$, we put
$y_C = \sum_{t \in S \cap C} y_t$. By using the results of Sections~3 and~5
of~\cite{EKVY:multi-objectives} and the modifications 
presented in \cite{BBCFK:MDP-two-views}, we first construct a 
finite-memory stochastic update strategy $\varrho$ such that the
probability of $R_C$ in $G^\varrho_s$ is equal to $y_C$. Then,
we construct a strategy
$\hat{\sigma}$ which plays according to $\varrho$ until a bottom 
strongly connected component $B$ of $G^\varrho_s$
is reached. %
Observe that the set of all states and actions which appear in 
$B$ is a subset of some $C \in \Mec(G)$. From that point on,
the strategy $\hat{\sigma}$ ``switches'' to the memoryless randomized strategy 
$\sigma_{x_C}$ of~Lemma~\ref{lem-x_c-almost-surely}. Hence,
$\Ex{\varrho}{s}{\lraname}$ and $\Va{\varrho}{s}{\lraname}$ are equal
to the right-hand sides of \eqref{eq:gl:exp} and \eqref{eq:gl:var},
respectively, and thus we get
$(\Ex{\varrho}{s}{\lraname},\Va{\varrho}{s}{\lraname}) \leq (u,v)$.
Note that $\hat{\sigma}$ may use more than 2-memory 
elements. A 2-memory strategy is obtained by modifying the initial part of $\hat{\sigma}$
(i.e., the part before the switch) into a memoryless strategy in the
same way as in~\cite{BBCFK:MDP-two-views}. Then, $\hat{\sigma}$ only needs
to remember whether a switch has already been performed or not, and hence
2~memory elements are sufficient. Finally, we transform $\hat{\sigma}$\label{page:hatsigma}
into another 2-memory stochastic update strategy $\sigma$ which satisfies
the extra conditions of Item~2 for a suitable~$z$. This is achieved 
by modifying the behaviour of $\hat{\sigma}$ in some MECs so that the
probability of staying in every MEC is preserved, the expected mean payoff
is also preserved, and the global variance can only decrease. This part
is somewhat tricky and the details are given in
\ifthenelse{\isundefined{\techreport}}{\cite{techrepLics13}}{Appendix~\ref{app-global}}.

We can solve
the strategy existence problem by encoding the existence of a solution
to~$L$ as a closed formula $\Phi$ of the existential fragment 
of $(\Rset,+,*,\leq)$. Since $\Phi$ is computable in polynomial time
and the existential fragment of $(\Rset,+,*,\leq)$ is decidable in
polynomial space \cite{Canny:Tarski-exist-PSPACE}, we obtain Theorem~\ref{thm:global}.

The pseudo-polynomial-time approximation algorithm is obtained as follows.
First note that if we had the number $z$ above, we could 
simplify
the system~$L$ of Fig.~\ref{fig-L} by substituting all $x_C$ variables
with constants. Then, \eqref{eq:gl:min}~and~\eqref{eq:gl:max} can be 
eliminated, \eqref{eq:gl:exp}~becomes a linear constraint,
and \eqref{eq:gl:var}~the only quadratic constraint. Thus, the system~$L$  
can be transformed into a quadratic program $L_{z}$ in which the quadratic constraint is negative semi-definite
with rank 1 (see
\ifthenelse{\isundefined{\techreport}}{\cite{techrepLics13}}{Appendix~\ref{app:global-approx-matrix}}),
and hence approximated in polynomial time
\cite{Vavasis92}. 
Since we do not know the precise number $z$ we try different
candidates $\bar z$, namely we approximate the value (to the precision $\frac{\varepsilon}{2}$) of $L_{\bar z}$
for all numbers $\bar z$ between $\min_{a\in A} r(a)$ and $\max_{a\in A} r(a)$ 
that are a multiple of $\tau = \frac{\varepsilon}{8 \max \{N,1\}}$ where $N$ is the maximal absolute value of an assigned reward.
If any $L_{\bar z}$
has a solution lower than $u-\frac{\varepsilon}{2}$, we output ``yes'', otherwise we output ``no''. The correctness
of the algorithm is proved in
\ifthenelse{\isundefined{\techreport}}{\cite{techrepLics13}}{Appendix~\ref{app:global-approx-cor}}.

Note that if we \emph{knew} the constant~$z$ we would even get that the approximation
problem can be solved in polynomial time (assuming that the number of digits 
in $z$ is polynomial in the size of the problem instance). Unfortunately,
our proof of Item~2 does not give a procedure for computing~$z$, and
we cannot even conclude that $z$~is rational.
We conjecture that the constant~$z$ \emph{can} actually be chosen as 
a rational number with small number of digits (which would immediately lower
the complexity of strategy existence to~$\NP$ using the results of~\cite{Vavasis90} for solving negative semi-definite quadratic programs). 
Also note that Remark~\ref{remark1} and Theorem~\ref{thm:global} immediately 
yield the following result.

\begin{corollary}
The approximate Pareto curve for global variance can be computed in 
pseudo-polynomial time.
\end{corollary}

\section{Local variance}
\label{sec-local}
In this section we analyse the problem for local variance. As before,
we start by showing the lower bounds for memory needed by strategies, and then provide an upper bound together
with an algorithm computing a Pareto optimal strategy.
As in the case of global variance, Pareto optimal strategies require both 
randomization and memory, however, in contrast to global variance where for
unichain MDPs deterministic memoryless strategies are sufficient we show 
(in the following example) that for local
variance both memory and randomization is required even for unichain MDPs.
\begin{figure} 
\begin{center}
 \begin{tikzpicture}[every node/.style={inner sep=0.5mm}][font=\small]

\node[state] at (0,0) (s1) {$s_1$};
\node[distr] at (1,0.5) (s1a) {};
\node[distr] at (1,-0.5) (s1b) {};
\node[state] at (2,0) (s2) {$s_2$};
\node[distr] at (1,0) (s2c) {};

\draw[trarr] (s1) |- (s1a) node[midway,left] {$a, 0$} -| (s2);
\draw[trarr] (s1) |- (s1b) node[midway,left] {$b, 2$} -| (s2);
\draw[trarr] (s2) -- (s2c) node[midway,above] {$c, 2$} -- (s1);

\end{tikzpicture}
 \end{center}
\caption{An MDP showing that Pareto optimal strategies need randomization/memory for local and hybrid variance.\label{fig:unichain-local}}
\end{figure}
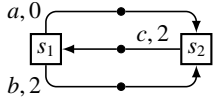  
\begin{example}\label{ex:local-mem}
Consider the MDP from Figure~\ref{fig:unichain-local}
and consider a strategy $\sigma$ that in the first step in $s_1$ makes a random choice uniformly between $a$ and $b$,
and then, whenever the state $s_1$ is revisited,
it chooses the action that was chosen in the first step.
The expected mean-payoff under such strategy is $0.5\cdot 2+0.5\cdot 1 = 1.5$
and the variance is
$\Big(0.5\cdot \big(0.5\cdot(0-1)^2 + 0.5\cdot(2-1)^2\big)\Big) + \Big(0.5\cdot (2-2)^2\Big)
 =0.5$.
 We show that the point $(1.5, 0.5)$ cannot be achieved by any memoryless randomized strategy $\sigma'$. %
Given $x\in \{a,b,c\}$, denote by $f(x)$ the frequency of the action $x$ under $\sigma'$.
Clearly, $f(c)=0.5$ and $f(b) = 0.5 - f(a)$. If $f(a)< 0.2$, then the mean-payoff $\Ex{\sigma'}{s_1}{\lraname} = 2\cdot (f(c) + f(b)) = 2 - 2f(a)$ is greater than $1.6$. Assume that $0.2\leq f(a)\leq 0.5$.
Then $\Ex{\sigma'}{s_1}{\lraname} \leq 1.6$ but
the variance is at least $0.64$ (see
\ifthenelse{\isundefined{\techreport}}{\cite{techrepLics13}}{Appendix~\ref{app:local-example}}
for computation).
Insufficiency of deterministic history-dependent strategies is proved using the same
equations and the fact
that there is only one run under such a strategy.

Thus have shown that memory and randomization is needed to achieve
a non-Pareto point $(1.55, 0.6)$.
The need of memory and randomization to achieve Pareto points will follow later from the fact that there always exist
Pareto optimal strategies.
\end{example}

In the remainder of this section we prove the following.
\begin{theorem}\label{thm:local-np-alg}
If there is a strategy $\zeta$ satisfying
$(\Ex{\zeta}{s_0}{\lraname}, \Ex{\zeta}{s_0}{\lrvlname})\leq (u,v)$
then there is a 3-memory strategy with the same properties.
The problem whether such a strategy exists belongs to $\mathbf{NP}$. Moreover, Pareto optimal strategies always exist.
\end{theorem}
We start by proving that 3-memory stochastic update strategies achieve all achievable points wrt. local variance.
\begin{proposition}\label{prop:local-main}
For every strategy $\zeta$ there is a 3-memory stochastic-update strategy $\sigma$ satisfying 
\[
(\Ex{\sigma}{s_0}{\lraname}, \Ex{\sigma}{s_0}{\lrvlname})\quad \leq\quad (\Ex{\zeta}{s_0}{\lraname}, \Ex{\zeta}{s_0}{\lrvlname})
\]
Moreover, the three memory elements of $\sigma$, say $m_1,m_2,m'_2$, satisfy the following:
\begin{itemize}
\item The memory element $m_1$ is initial, $\sigma$ may randomize in $m_1$ and may stochastically update its memory either to $m_2$, or to $m'_2$.
\item In $m_2$ and $m'_2$ the strategy $\zeta$ behaves {\em deterministically} and never changes its memory.
\end{itemize}
\end{proposition}
\begin{proof}
By Lemma~\ref{lemma:stay-mec} $\sum_{C\in \Mec(G)}\ \Prb(R_C)=1$, and
\begin{multline*}
(\Ex{\zeta}{s_0}{\lraname},\Ex{\zeta}{s_0}{\lrvlname})\\\quad=\Big(\sum_{C\in \Mec(G)}\!\!\!\! \Prb(R_C)\cdot \Ex{\zeta}{s_0}{\lraname\mid R_C},\sum_{C\in \Mec(G)}\!\!\!\! \Prb(R_C)\cdot \Ex{\zeta}{s_0}{\lrvlname\mid R_C}\Big).
\end{multline*}
In what follows we sometimes treat each MEC $C$ as a standalone MDP obtained by restricting $G$ to $C$. Then, for example, $C^{\kappa}$ denotes the Markov chain obtained by applying the strategy $\kappa$ to the component $C$.

The next proposition formalizes the main idea of our proof:
\begin{proposition}\label{prop:strong-opt-local}
Let $C$ be a MEC.
There are two frequency functions $f_C:C\rightarrow \Rset$ and $f'_C:C\rightarrow \Rset$ on $C$, and a number $p_C\in [0,1]$ such that the following holds
\begin{multline*}
p_C\cdot (\lraname[f_C],\lrvname[f_C])+(1-p_C)\cdot
 (\lraname[f'_C],\lrvname[f'_C])\\
  \leq (\Ex{\zeta}{s_0}{\lraname|R_C}, \Ex{\zeta}{s_0}{\lrvlname|R_C})\,.
\end{multline*}
\end{proposition}
The proposition is proved in
\ifthenelse{\isundefined{\techreport}}{\cite{techrepLics13}}{Appendix~\ref{app-strong-opt}},
where we first show that it follows
from a relaxed version of the proposition which gives us, for any $\varepsilon>0$, frequency functions $f_\varepsilon$ and $f'_\varepsilon$ and number $p_\varepsilon$ such that
\begin{multline*}
p_\varepsilon\cdot (\lraname[f_\varepsilon],\lrvname[f_\varepsilon])+(1-p_\varepsilon)\cdot
 (\lraname[f'_\varepsilon],\lrvname[f'_\varepsilon])\\
  \leq (\Ex{\zeta}{s_0}{\lraname|R_C}, \Ex{\zeta}{s_0}{\lrvlname|R_C}) + (\varepsilon, \varepsilon)\,.
\end{multline*}
Then we show that the weaker version holds by showing that there are runs $\omega$ from which we can extract
the frequency functions $f_\varepsilon$ and $f_\varepsilon'$. The selection of runs is rather involved, since it is not clear a priori which runs to pick or even how
to extract the frequencies from them (note that the naive approach of considering the average ratio of taking a given action $a$ does not work, since the averages might not be defined).

Proposition~\ref{prop:strong-opt-local} implies that any expected mean payoff and local variance achievable on a MEC $C$
can be achieved by a composition of two memoryless randomized strategies giving precisely the frequencies of actions specified by $f_C$ and $f'_C$ (note that $\lrvname[f_C]$ and $\lrvname[f'_C]$ may not be equal to the expected local variance of such strategies, but we show that the ``real'' expected local variance cannot be larger). 
By further selecting BSCCs of these strategies and using some de-randomization tricks we obtain, for every MEC $C$, two memoryless deterministic strategies $\pi_C$ and $\pi'_C$ and a constant $h_C$ such that for every
$s\in \MECstate{C}$ the value of
$h_C (\Ex{\pi_C}{s}{\lraname}, \Ex{\pi_C}{s}{\lrvlname})+(1-h_C) (\Ex{\pi'_C}{s}{\lraname}, \Ex{\pi'_C}{s}{\lrvlname})$
is equal to a fixed $(u',v')$ (since both $C^{\pi_C}$ and $C^{\pi'_C}$ have only one BSCC) satisfying $(u',v') \leq (\Ex{\zeta}{s_0}{\lraname|R_C}, \Ex{\zeta}{s_0}{\lrvlname|R_C})$. We define two memoryless deterministic strategies $\pi$ and $\pi'$ that in every $C$ behave as $\pi_C$ and $\pi'_C$, respectively.
Details of the steps above are
\ifthenelse{\isundefined{\techreport}}{in \cite{techrepLics13}}{ postponed to Appendix~\ref{app-local-main}}.

Using similar arguments as in~\cite{BBCFK:MDP-two-views} (that in turn depend on results of~\cite{EKVY:multi-objectives}) one may show that there is a $2$-memory stochastic update strategy $\sigma'$, with two memory locations $m_1,m_2$, satisfying the following properties: In $m_1$, the strategy $\sigma'$ may randomize and may stochastically update its memory to $m_2$. In $m_2$, the strategy $\sigma'$ never changes its memory. Most importantly, the probability that $\sigma'$ updates its memory from $m_1$ to $m_2$ in a given MEC $C$ is equal to $\Pr{\zeta}{s_0}{R_C}$.

We modify the strategy $\sigma'$ to the desired 3-memory $\sigma$ by splitting the memory element $m_2$ into two elements $m_2,m'_2$. Whenever $\sigma'$ updates to $m_2$, the strategy $\sigma$ further chooses randomly whether to update either to $m_2$ (with prob. $h_C$), or to $m'_2$ (with prob. $1-h_C$). Once in $m_2$ or $m'_2$, the strategy $\sigma$ never changes its memory and plays according to $\pi$ or $\pi'$, respectively.
For every MEC $C$ we have
$\Prb^{\sigma}_{s_0}(\text{update to } m_2 \text{ in }C)=\Prb(R_C)\cdot h_C$ and
$\Prb^{\sigma}_{s_0}(\text{update to } m'_2 \text{ in }C)=\Prb(R_C)\cdot (1-h_C)$. Thus we get
\begin{equation}
(\Ex{\zeta}{s_0}{\lraname},\Ex{\zeta}{s_0}{\lrvlname})
 = (\Ex{\sigma}{s_0}{\lraname},\Ex{\sigma}{s_0}{\lrvlname})\label{eqn-t2}
\end{equation}
as shown in
\ifthenelse{\isundefined{\techreport}}{\cite{techrepLics13}}{Appendix~\ref{app-eqn-t2}}.
\end{proof}
Proposition~\ref{prop:local-main} combined with results of~\cite{BBCFK:MDP-two-views} allows us to finish the proof of Theorem~\ref{thm:local-np-alg}.
\begin{proof}[Proof (of Theorem~\ref{thm:local-np-alg})]
Intuitively, the non-deterministic polynomial time algorithm works as follows: First, guess two memoryless deterministic strategies $\pi$ and $\pi'$. Verify whether there is a 3-memory stochastic update strategy $\sigma$ with memory elements $m_1,m_2,m'_2$ which in $m_2$ behaves as $\pi$, and in $m'_2$ behaves as $\pi'$ such that
$(\Ex{\sigma}{s_0}{\lraname}, \Ex{\sigma}{s_0}{\lrvlname}) \leq  (u,v)$. Note that it suffices to compute the probability distributions chosen by $\sigma$ in the memory element $m_1$ and the probabilities of updating to $m_2$ and $m'_2$.
This can be done by a reduction to the controller synthesis problem for two dimensional mean-payoff objectives studied in~\cite{BBCFK:MDP-two-views}.

More concretely, we construct a new MDP $G[\pi,\pi']$ with 
\begin{itemize}
\item the set of states $S':=\{s_{in}\} \cup (S\times \{m_1,m_2,m'_2\})$

(Intuitively, the $m_1,m_2,m'_2$ correspond to the memory elements of $\sigma$.)
\item the set of actions%
  \footnote{To keep the presentation simple, here we do not require that every action is enabled in at most one step.}
  $A\cup \{[\pi],[\pi'], \mathit{default}\}$

\item the mapping $\mathit{Act}'$ defined by $\mathit{Act}'(s_{in}) = \{[\pi],[\pi'],\mathit{default}\}$,
 $\mathit{Act}'((s,m_1))=\mathit{Act}(s)\cup \{[\pi],[\pi']\}$ 
 and $\mathit{Act}'((s,m_2))=\mathit{Act}'((s,m'_2))=\{\mathit{default}\}$

(Intuitively, the actions $[\pi]$ and $[\pi']$ simulate the update of the memory element $m_2$ and to $m'_2$, respectively, in $\sigma$. As $\sigma$ is supposed to behave in a fixed way in $m_2$ and $m'_2$, we do not need to simulate its behavior in these states in $G[\pi,\pi']$. Hence, the $G[\pi,\pi']$ just loops under the action $\mathit{default}$ in the states $(s,m_2)$ and $(s,m'_2)$. The action $\mathit{default}$ is also used in the initial state to denote that the initial memory element is $m_1$.)
\item the probabilistic transition function $\delta'$ defined as follows:
  \begin{itemize}
  \item $\delta'(s_{in})(\mathit{default})((s_0,m_1))=\delta(s_{in},[\pi])((s_0,m_2))=\delta(s_{in},[\pi'])((s_0,m'_2))=1$ for $a\in A$ and $t\in S$
  \item $\delta'((s,m_1),a)((t,m_1))=\delta(s,a)(t)$ for $a\in A$ and $t\in S$
  \item $\delta'((s,m_1),[\pi])((s,m_2))=\\\delta'((s,m_1),[\pi'])((s,m'_2))=1$ 
  \item $\delta'((s,m_2),\mathit{default})((s,m_2))=\\\delta'((s,m'_2),\mathit{default})((s,m'_2))=1$
  \end{itemize}
\end{itemize}
We define a vector of rewards $\vec{r}:S'\rightarrow \Rset^2$ as follows: $\vec{r}((s,m_2)):=(\Ex{\pi}{s}{\lraname}, \Ex{\pi}{s}{\lrvlname})$ and $\vec{r}((s,m'_2)):=(\Ex{\pi'}{s}{\lraname}, \Ex{\pi'}{s}{\lrvlname})$ and
$\vec{r}(s_{in})=\vec{r}((s,m_1)):=(\max_{a\in A} r(a)+1,(\max_{a\in A} r(a)-\min_{a\in A} r(a))^2+1)$.
(Here the rewards are chosen in such a way that no (Pareto) optimal scheduler can stay in the states of the form $(s,m_1)$ with positive probability.) Note that $\vec{r}$ can be computed in polynomial time using standard algorithms for computing mean-payoff in Markov chains~\cite{Norris:book}.

In
\ifthenelse{\isundefined{\techreport}}{\cite{techrepLics13}}{Appendix~\ref{app-local-np-alg}}
we show that if there is a strategy $\zeta$ for $G$ such that
$(\Ex{\zeta}{s_0}{\lraname}, \Ex{\zeta}{s_0}{\lrvlname})\leq (u,v)$, then there is a (memoryless randomized) strategy $\rho$ in $G[\pi,\pi']$ such that $(\Ex{\rho}{s_{in}}{\lraname^{\vec{r}_1}},\Ex{\rho}{s_{in}}{\lraname^{\vec{r}_2}})\leq (u,v)$. Also, we show that such $\rho$ can be computed in polynomial time using results of~\cite{BBCFK:MDP-two-views}.
Finally, it is straightforward to move the second component of the states of $G[\pi,\pi']$ to the memory of a stochastic update strategy which gives a 3-memory stochastic update strategy $\sigma$ for $G$ with the desired properties. Thus a non-deterministic polynomial time algorithm works as follows: (1) guess $\pi,\pi'$ (2) construct $G[\pi,\pi']$ and $\vec{r}$ (3) compute $\rho$ (if it exists). As noted above, $\rho$ can be transformed to the 3-memory stochastic update strategy $\sigma$ in polynomial time.

Finally, we can show that Pareto optimal strategies exist by a reasoning similar to the one used in global variance.
\end{proof}

Theorem~\ref{thm:local-np-alg} and Remark~\ref{remark1} give the following corollary.

\begin{corollary}
The approximate Pareto curve for local variance can be computed in 
exponential time.
\end{corollary}

\section{Hybrid variance}
\label{sec-hybrid}

We start by showing that memory or randomization is needed for 
Pareto optimal strategies in unichain MDPs for hybrid variance; and then 
show that both memory and randomization is required for hybrid variance for
general MDPs.

\begin{example}
Consider again the MDP from Fig.~\ref{fig:unichain-local}, and any memoryless deterministic strategy. There are in fact two of these. One, which choses $a$ in $s_1$,
yields the variance $1$, and the other, which chooses $b$ in $s_1$, yields the expectation $2$.

However, a memoryless randomized strategy $\sigma$ which randomizes uniformly between $a$ and $b$
yields the expectation $1.5$ and variance
\ifthenelse{\isundefined{\techreport}}{$0.75$}{
\begin{multline*}
 \Big(0.5\cdot \big(0.5\cdot(0-1.5)^2 + 0.5\cdot(2-1.5)^2\big)\Big) + \Big(0.5\cdot (2-0.15)^2\Big)\\
 = 0.25\cdot 2.25 + 0.75\cdot 0.25 = 0.75
\end{multline*}}
which makes it incomparable to either of the memoryless deterministic strategies.
Similarly, the deterministic strategy which alternates between $a$ and $b$ on subsequent visits of
$s_1$ yields the same values as the $\sigma$ above. This gives us that memory or randomization is
needed even to achieve a non-Pareto point $(1.6, 0.8)$.
\end{example}

Before proceeding with general MDPs, we give the following proposition, which states
an interesting and important relation between the three notions of variance%
\footnote{Note that Proposition~\ref{prop:relation} does {\em not} simplify the decision problem for hybrid 
variance, since it does not imply that the algorithms for global and local variance could be combined.}.
The proposition is proved in
\ifthenelse{\isundefined{\techreport}}{\cite{techrepLics13}}{Appendix~\ref{app:relation}}.
\begin{proposition}\label{prop:relation}
Suppose $\sigma$ is a strategy under which for almost all $\omega$ the limits exists for
$\lrvhname(\omega)$, $\lraname(\omega)$, and $\lrvname(\omega)$ (i.e. the $\limsup$ in their definitions can be
swapped for $\lim$). Then
\[\Ex{\sigma}{s}{\lrvhname} = \Va{\sigma}{s}{\lraname} + \Ex{\sigma}{s}{\lrvname}\;.\]
\end{proposition}

Now we can show that both memory and randomization is needed, by extending Example~\ref{fig-MDP-mem-rand}.
\begin{example}\label{ex:hybrid-mem-ran}
Consider again the MDP from Fig.~\ref{fig-MDP-mem-rand}. Under every strategy,
every run $\omega$ satisfies $\lrvname(\omega)=0$, and
the limits for $\lraname(\omega)$, $\lrvname(\omega)$
and $\lrvhname(\omega)$ exist. Thus
$\Ex{\zeta}{s}{\lrvname} = 0$ for all $\zeta$ and by
Proposition~\ref{prop:relation} we get
$\Ex{\zeta}{s}{\lrvhname} = \Va{\zeta}{s}{\lraname}$.
Hence we can use Example~\ref{fig-MDP-mem-rand} to reason that both
memory and randomization is needed to achieve the Pareto point $(4,2)$ in
Fig.~\ref{fig-MDP-mem-rand}.
\end{example}

Now we prove the main theorem of this section.
\begin{theorem}\label{thm:hybrid}
  If there is a strategy $\zeta$ satisfying
  $(\Ex{\zeta}{s}{\lraname}, \Ex{\zeta}{s_0}{\lrvhname})
     \leq (u,v)$,
  then there is a 2-memory strategy with the same properties.
The problem whether such a strategy exists belongs to $\mathbf{NP}$,
and approximation of the answer
  can be done in polynomial time. Moreover, Pareto optimal strategies always exist.
\end{theorem}

We start by proving that 2-memory stochastic update strategies are 
sufficient for Pareto optimality wrt.\ hybrid variance. 

\begin{figure}[t]
\begin{align}
\mathbf{1}_{s_0}(s) + \sum_{a\in A} y_{a}\cdot \delta(a)(s) & = 
 \sum_{a\in \act{s}} y_{a} + y_s \hspace*{1em} \text{for all  $s\in S$}
\label{eq:yah}\\
\sum_{C\in \Mec(G)}\sum_{s\in S\cap C}y_{s} & =  1 &
\label{eq:ys1h}\\
 \sum_{s\in C} y_{s} & =  \sum_{a\in A\cap C}\!\! x_{a} \hspace*{1.1em}
   \text{for all $C\in\Mec(G)$}
\label{eq:yCh}\\
\sum_{a\in A} x_{a}\cdot \delta(a)(s) & = 
\sum_{a\in \act{s}} x_{a} \hspace*{1em} \text{for all  $s\in S$}
\label{eq:xah}\\
u &\geq \sum_{a\in A} x_{a}\cdot r(a)
\label{eq:rewh}\\
v &\geq\sum_{a\in A} x_{a}\cdot r^2(a) -\bigg(\sum_{a\in A} x_{a}\cdot r(a)\bigg)^2
\label{eq:varh}
\end{align}
\caption{The system $L_H$. (Here $\mathbf{1}_{s_0}(s)=1$ 
if $s = s_0$, and $\mathbf{1}_{s_0}(s)=0$ otherwise.)}
\label{system-LH}
\end{figure}

\begin{proposition}\label{prop:hybrid-details}
Let $s_0\in S$ and $u,v\in \Rset$.
\begin{enumerate}
\item \label{thm:hybrid-main-stop} 
If there is a strategy $\zeta$ satisfying
$(\Ex{\zeta}{s_0}{\lraname}, \Ex{\zeta}{s_0}{\lrvhname}) \leq (u,v)$,
then the system $L_H$ (Fig.~\ref{system-LH}) has a non-negative solution.
\item \label{thm:hybrid-main-ptos} If there is a non-negative solution for the 
system $L_H$ (Fig.~\ref{system-LH}), 
then there is a 2-memory stochastic-update strategy $\sigma$ satisfying 
$(\Ex{\sigma}{s_0}{\lraname}, \Ex{\sigma}{s_0}{\lrvhname}) \leq (u,v)$.
\end{enumerate}
\end{proposition}
Notice that we get the existence of Pareto optimal strategies as a side product of the above
proposition, similarly to the case of global variance.

We briefly sketch the main ingredients for the proof of 
Proposition~\ref{prop:hybrid-details}. 
We first establish the sufficiency of finite-memory strategies by showing that 
for an arbitrary strategy $\zeta$, there is a 3-memory stochastic update
strategy $\sigma$ such that 
$(\Ex{\sigma}{s_0}{\lraname}, \Ex{\sigma}{s_0}{\lrvhname}) \leq 
(\Ex{\zeta}{s_0}{\lraname}, \Ex{\zeta}{s_0}{\lrvhname})$.
The key idea of the proof of the construction of a 3-memory stochastic update
strategy $\sigma$ from an arbitrary strategy $\zeta$ is similar to the proof of 
Proposition~\ref{prop:local-main}. The details are in
\ifthenelse{\isundefined{\techreport}}{\cite{techrepLics13}}{Appendix~\ref{app-hybrid-finite}}.
We then focus on finite-memory strategies.
For a finite-memory strategy $\zeta$, the frequencies are well-defined,
and for an action $a\in A$, let
$
f(a)
\coloneqq
\lim_{\ell\to\infty}
\frac{1}{\ell} \sum_{t=0}^{\ell-1} \Pr{\zeta}{s_0}{A_t=a} 
$
denote the frequency of action $a$.
We show that setting
\(
x_a \coloneqq f(a)
\)
for all $a\in A$ satisfies  Eqns.~(\ref{eq:xah}), Eqns.~(\ref{eq:rewh}) and Eqns.~(\ref{eq:varh}) of~$L_H$.
To obtain $y_a$ and $y_s$, we define them in the same way as done 
in~\cite[Proposition~2]{BBCFK:MDP-two-views} using the results 
of~\cite{EKVY:multi-objectives}. The details are
\ifthenelse{\isundefined{\techreport}}{in \cite{techrepLics13}}{postponed to Appendix~\ref{app-hybrid-details-a}}.
This completes the proof of the first item.
The proof of the second item is as follows:
the construction of a 2-memory stochastic update strategy $\sigma$ from the 
constraints of the system $L_H$ (other than constraint of Eqns~\ref{eq:varh}) 
was presented in~\cite[Proposition~1]{BBCFK:MDP-two-views}.
The key argument to show that strategy $\sigma$ also satisfies 
Eqns~\ref{eq:varh} is obtained by establishing that for the 
strategy $\sigma$ we have:
$\Ex{\sigma}{s}{\lrvhname} = \Ex{\sigma}{s}{\lraname_{r^2}} - \Ex{\sigma}{s}{\lraname}^2$
(here $\lraname_{r^2}$ is the value of $\lraname$ w.r.t. reward function defined by $r^2(a)=r(a)^2$; the equality is shown in
\ifthenelse{\isundefined{\techreport}}{\cite{techrepLics13}}{Appendix~\ref{app-hybrid-details-b}}).
It follows immediately that Eqns~\ref{eq:varh} is satisfied.
This completes the proof of Proposition~\ref{prop:hybrid-details}.
Finally we show that for the quadratic program defined by the system 
$L_H$, the quadratic constraint satisfies the conditions of negative 
semi-definite programming with matrix of rank~1 (see
\ifthenelse{\isundefined{\techreport}}{\cite{techrepLics13}}{Appendix~\ref{app-hybrid-semidefinite}}).
Since negative semi-definite programs  can be decided in NP~\cite{Vavasis90} 
and with the additional restriction of rank~1 can be approximated in 
polynomial time~\cite{Vavasis92}, we get the complexity bounds of Theorem~\ref{thm:hybrid}.
Finally, Theorem~\ref{thm:hybrid} and Remark~\ref{remark1} give the following
result.

\begin{corollary}
The approximate Pareto curve for hybrid variance can be computed in 
pseudo-polynomial time.
\end{corollary}

\newcommand{\wh}{\widehat}
\newcommand{\ovk}{\overline}

\section{Zero variance with optimal performance}\label{sect-zero}
Now we present polynomial-time algorithms to compute the optimal
expectation that can be ensured along with zero variance. The results are
captured in the following theorem.

\begin{theorem}\label{thm:zero}
The minimal expectation that can be ensured
\begin{enumerate}
 \item with zero hybrid variance can be computed in $O((|S|\cdot|A|)^2)$ time 
using discrete graph theoretic algorithms;
 \item with zero local variance can be computed in PTIME;
 \item with zero global variance can be computed in PTIME.
\end{enumerate}
\end{theorem}

\smallskip\noindent{\bf Hybrid variance.}
The algorithm for zero hybrid variance is as follows: 
(1)~Order the rewards in an increasing sequence 
$\beta_1 < \beta_2 < \ldots < \beta_n$; 
(2)~find the least $i$ such that $A_i$ is the set of actions with reward
$\beta_i$ and it can be ensured with probability~1 (almost-surely) 
that eventually only actions in $A_i$ are visited, and output $\beta_i$; and 
(3)~if no such $i$ exists output ``NO" (i.e., zero hybrid variance cannot 
be ensured).
Since almost-sure winning for MDPs with eventually always property (i.e., 
eventualy only actions in $A_i$ are visited) can be decided in quadratic
time with discrete graph theoretic algorithm~\cite{CH12,CH11}, we obtain the first item 
of Theorem~\ref{thm:zero}. The correctness is proved in
\ifthenelse{\isundefined{\techreport}}{\cite{techrepLics13}}{Appendix~\ref{app-zero-hybrid}}.

\smallskip\noindent{\bf Local variance.}
For zero local variance, we make use of the previous algorithm.
The intuition is that to minimize the expectation with zero local variance,
a strategy $\sigma$ needs to reach states $s$ in which zero hybrid variance can be
ensured by strategies~$\sigma_s$, 
and then mimic them.
Moreover, $\sigma$ minimizes the expected value of $\lraname$
among all possible behaviours satisfying the above.
The algorithm is as follows:
(1)~Use the algorithm for zero hybrid 
variance to compute a function $\beta$ that assigns to every state $s$ 
the minimal expectation value $\beta(s)$ that can be ensured along with zero 
hybrid variance when starting in $s$, and if zero hybrid variance cannot be ensured, then 
$\beta(s)$ is assigned $+\infty$. Let $M=1 +\max_{s\in S} \beta(s)$.
(2)~Construct an MDP $\ovk{G}$ as follows: 
For each state $s$ such that $\beta(s) < \infty$ we add a state 
$\ovk{s}$ with a self-loop on it, and we add a new action $a_s$ that leads from 
$s$ to $\ovk{s}$.
(3)~Assign a reward $\beta(s) - M$ to $a_s$, and $0$ to all other actions.
Let $T=\{ a_s \mid \beta(s) < \infty \}$ be the target set of actions.
(4)~Compute a strategy that minimizes the cumulative reward and ensures 
almost-sure (probability~1) reachability to $T$ in $\ovk{G}$. 
Let $\wh{\beta}(s)$ denote the minimal expected payoff for the cumulative 
reward; and $\ovk{\beta}(s)=\wh{\beta}(s) +M$.
In
\ifthenelse{\isundefined{\techreport}}{\cite{techrepLics13}}{Appendix~\ref{app-zero-local}}
we show that $\ovk{\beta}(s)$ is the minimal expectation that can be ensured
with zero local variance, and every step of the above computation can 
be achieved in polynomial time.
This gives us the second item of Theorem~\ref{thm:zero}.

\smallskip\noindent{\bf Global variance.}
The basic intuition for zero global variance is that
we need to find the minimal number $y$ such that there is an almost-sure winning strategy 
to reach the MECs where expectation {\em exactly} $y$ can be ensured with zero variance.

The algorithm works as follows:
(1)~Compute the MEC decomposition of the MDP and let the MECs be 
$C_1,C_2,\ldots,C_n$.
(2)~For every MEC $C_i$ compute the minimal expectation $\alpha_{C_i}=\inf_{\sigma} \min_{s\in C_i} \Ex{\sigma}{s}{\lraname}$ and 
the maximal expectation $\beta_{C_i}=\sup_{\sigma} \max_{s\in C_i} \Ex{\sigma}{s}{\lraname}$ that can be ensured in the MDP induced by the 
MEC $C_i$.
(3)~Sort the values $\alpha_{C_i}$ in a non-decreasing order as
$\ell_1 \leq \ell_2 \leq \ldots \leq \ell_n$.
(4)~Find the least $i$ such that 
(a)~$\mathcal{C}_i=\{C_j \mid \alpha_{C_j} \leq \ell_i \leq \beta_{C_j} \}$ is 
the MEC's whose interval contains $\ell_i$; 
(b)~almost-sure (probability~1) reachability to the set 
$\bigcup_{C_j \in \mathcal{C}_i} C_j$ (the union of the MECs in 
$\mathcal{C}_i$) can be ensured;
and output $\ell_i$.
(5)~If no such $i$ exists, then the answer to zero global variance is ``NO"
(i.e., zero global variance cannot be ensured).
All the above steps can be computed in polynomial time. The correctness is proved in
\ifthenelse{\isundefined{\techreport}}{\cite{techrepLics13}}{Appendix~\ref{app-zero-global}},
and we obtain the last item of Theorem~\ref{thm:zero}.

\section{Conclusion}
We studied three notions of variance for MDPs 
with mean-payoff objectives: global (the standard one), local and hybrid variance.
We established a strategy complexity (i.e., the memory and randomization
required) for Pareto optimal strategies. 
For the zero variance problem, all the three cases are in PTIME.
There are several interesting open questions. 
The most interesting open questions are whether the approximation  problem 
for local variance can be solved in polynomial time, and what are the exact
complexities of the strategy existence problem.

\medskip
\noindent
{\bf Acknowledgements.}
T.~Br\'{a}zdil is supported by the Czech Science Foundation, grant
No~P202/12/P612.
K.~Chatterjee is supported by the 
Austrian Science Fund (FWF) Grant No P 23499-N23;
FWF NFN Grant No S11407-N23 (RiSE);
ERC Start grant (279307: Graph Games);
Microsoft faculty fellows award. V. Forejt is supported by a Royal Society Newton Fellowship and EPSRC project~EP/J012564/1,
and is also affiliated with FI MU Brno, Czech Republic.

{\small
\bibliographystyle{plain}
\bibliography{str-short,concur}

\begin{thebibliography}{10}

\bibitem{Altman}
E.~Altman.
\newblock {\em Constrained Markov Decision Processes (Stochastic Modeling)}.
\newblock Chapman \& Hall/CRC, 1999.

\bibitem{Billingsley:book}
P.~Billingsley.
\newblock {\em Probability and Measure}.
\newblock Wiley, 1995.

\bibitem{Boyd04}
S.~Boyd and L.~Vandenberghe.
\newblock {\em Convex Optimization}.
\newblock Cambridge Univ. Press, 2004.

\bibitem{BBCFK:MDP-two-views}
T.~Br\'{a}zdil, V.~Bro{\v{z}}ek, K.~Chatterjee, V.~Forejt, and A.~Ku{\v{c}}era.
\newblock Two views on multiple mean-payoff objectives in {Markov} decision
  processes.
\newblock In {\em Proceedings of LICS 2011}. IEEE, 2011.

\bibitem{Canny:Tarski-exist-PSPACE}
J.~Canny.
\newblock Some algebraic and geometric computations in {PSPACE}.
\newblock In {\em Proceedings of STOC'88}, pages 460--467. ACM Press, 1988.

\bibitem{CH11}
K.~Chatterjee and M.~Henzinger.
\newblock Faster and dynamic algorithms for maximal end-component decomposition
  and related graph problems in probabilistic verification.
\newblock In {\em SODA}, pages 1318--1336. SIAM, 2011.

\bibitem{CH12}
K.~Chatterjee and M.~Henzinger.
\newblock An {O}({\it n}$^{\mbox{2}}$) time algorithm for alternating
  {B}{\"u}chi games.
\newblock In {\em SODA}, pages 1386--1399. SIAM, 2012.

\bibitem{CJH04}
K.~Chatterjee, M.~Jurdzinski, and T.~Henzinger.
\newblock Quantitative stochastic parity games.
\newblock In {\em SODA}, pages 121--130. SIAM, 2004.

\bibitem{CHMH:multi-objectives}
K.~Chatterjee, R.~Majumdar, and T.~Henzinger.
\newblock Markov decision processes with multiple objectives.
\newblock In {\em Proceedings of STACS 2006}, volume 3884 of {\em LNCS}, pages
  325--336. Springer, 2006.

\bibitem{Chung94}
K-J. Chung.
\newblock Mean-variance tradeoffs in an undiscounted {MDP}: The unichain case.
\newblock {\em Operations Research}, 42:184--188, 1994.

\bibitem{CY:MDP-regular-TAC}
C.~Courcoubetis and M.~Yannakakis.
\newblock Markov decision processes and regular events.
\newblock {\em {IEEE} Transactions on Automatic Control}, 43(10):1399--1418,
  1998.

\bibitem{derman1970finite}
C.~Derman.
\newblock {\em Finite state Markovian decision processes}.
\newblock Mathematics in science and engineering. Academic Press, 1970.

\bibitem{EKVY:multi-objectives}
K.~Etessami, M.~Kwiatkowska, M.~Vardi, and M.~Yannakakis.
\newblock Multi-objective model checking of {Markov} decision processes.
\newblock {\em Logical Methods in Computer Science}, 4(4):1--21, 2008.

\bibitem{FKL89}
J.~A. Filar, L.C.M. Kallenberg, and H-M. Lee.
\newblock Variance-penalize {Markov} decision processes.
\newblock {\em Math. of Oper. Research}, 14:147--161, 1989.

\bibitem{FKP12}
V.~Forejt, M.~Kwiatkowska, and D.~Parker.
\newblock Pareto curves for probabilistic model checking.
\newblock In {\em Proc. of ATVA'12}, volume 7561 of {\em LNCS}, pages 317--332.
  Springer, 2012.

\bibitem{ICML2011Mannor_156}
S.~Mannor and J.~Tsitsiklis.
\newblock Mean-variance optimization in {M}arkov decision processes.
\newblock In {\em Proceedings of ICML-11}, pages 177--184, New York, NY, USA,
  June 2011. ACM.

\bibitem{Norris:book}
J.R.{} Norris.
\newblock {\em {Markov} Chains}.
\newblock Cambridge University Press, 1998.

\bibitem{Puterman:book}
M.L.{} Puterman.
\newblock {\em {Markov} Decision Processes}.
\newblock Wiley, 1994.

\bibitem{Royden88}
H.~L. Royden.
\newblock {\em {Real analysis}}.
\newblock Macmillan, New York, 3rd edition, 1988.

\bibitem{Sobel82}
M.~J. Sobel.
\newblock The variance of discounted {MDP's}.
\newblock {\em Journal of Applied Probability}, 19:794--802, 1982.

\bibitem{Sobel94}
M.~J. Sobel.
\newblock Mean-variance tradeoffs in an undiscounted {MDP}.
\newblock {\em Operations Research}, 42:175--183, 1994.

\bibitem{Vavasis90}
S.~A. Vavasis.
\newblock Quadratic programming is in {NP}.
\newblock {\em Information Processing Letters}, 36(2):73 -- 77, 1990.

\bibitem{Vavasis92}
S.~A. Vavasis.
\newblock Approximation algorithms for indefinite quadratic programming.
\newblock {\em Math. Program.}, 57(2):279--311, November 1992.

\end{thebibliography}
}

\ifthenelse{\isundefined{\techreport}}{}{
\newpage
\onecolumn
\appendix
\subsection{Proofs for Global Variance}
\label{app-global}
\subsubsection{Obtaining values $y_\kappa$ for $\kappa \in S\cup A$ in Item 1 of Proposition~\ref{prop:global-lp}}
\label{app-global-stay-mecs}
Let $G$ be an MDP, and let $G'$ be obtained from $G$ by adding a state $d_s$ for every state $s\in S$, and an action $a_s$ that leads to $d_s$ from $s$.
\begin{lemma}
Let $\sigma$ be a strategy for $G$. Then there is a strategy $\bar\sigma$ in $G'$ such that
$\Pr{\sigma}{s_{in}}{\staymec{C}} = \Pr{\bar\sigma}{s_{in}}{\bigcup_{s\in C}\reach(d_s)}$.
\end{lemma}
\begin{proof}
We give a proof by contradiction.
Let $C_1,\ldots C_n$ be all MECs of $G$, and let $X\subseteq \Rset^n$ be the set of all points $(x_1,\ldots,x_n)$ for which
there is a strategy $\sigma'$ in $G'$ such that $\Pr{\sigma'}{s_{in}}{\bigcup_{s\in C_i}\reach(d_s)} \ge x_i$ for all $1\le i \le n$.
Let $(y_1,\ldots,y_n)$ be the numbers such that $\Pr{\sigma}{s_{in}}{R_{C_i}} = y_i$ for all $1\le i \le n$. For contradiction, suppose
$(y_1,\ldots,y_n) \not\in X$. By \cite[Theorem 3.2]{EKVY:multi-objectives} the set $X$ can be described as a set of solutions of a linear program, and hence it is convex.
By separating hyperplane theorem (see e.g. ~\cite{Boyd04}) there are non-negative weights $w_1,\ldots,w_n$ such that 
$\sum_{i=0}^n y_i\cdot w_i > \sum_{i=0}^n x_i\cdot w_i$ for every $(x_1,\ldots, x_n)\in X$.

We define a reward function $r$ by
$r(a)=w_i$ for an action $a$ from $C_i$, where $1\le i\le n$, and $r(a)=0$ for actions not in any MEC. Observe that the mean payoff of 
any run that eventually stays in a MEC $C_i$ is $w_i$, and so the expected mean payoff w.r.t. $r$ under $\sigma$ is 
$\sum_{i=0}^n y_i\cdot w_i$. Because memoryless deterministic strategies suffice for maximizing the expected mean payoff, there is also a memoryless deterministic strategy
$\hat\sigma$ for $G$ that yields expected mean payoff w.r.t. $r$ equal to $z\ge\sum_{i=0}^n y_i\cdot w_i$. We now define a strategy $\bar\sigma$ for $G'$
to mimic $\hat\sigma$ until a BSCC is reached, and when a BSCC is reached, say along a path $w$, the strategy $\bar\sigma$
takes the action $a_{\last(w)}$.
Let $x_i=\Pr{\bar\sigma}{s_{in}}{\bigcup_{s\in C_i}\reach(d_s)}$. Due to the construction of $\bar\sigma$ we have 
$x_i = \Pr{\hat\sigma}{s_{in}}{R_{C_i}}$: this follows because once a BSCC is reached on a path $w$, every run $\omega$ extending $w$
has an infinite suffix containing only the states of the MEC containing the state $\last(w)$.
Hence $\sum_{i=0}^n x_i\cdot w_i = z$. However, by the choice of the weights $w_i$ we get that $(x_1,\ldots,x_n)\not\in X$,
and hence a contradiction, because $\bar\sigma$ witnesses that $(x_1,\ldots,x_n)\in X$.
\end{proof}

Let $\zeta$ be the strategy from Item 1. of Proposition~\ref{prop:global-lp}. By the above lemma there is a strategy $\zeta'$ for $G'$ such that 
$\Pr{\zeta}{s_{in}}{\staymec{C}} = \Pr{\zeta'}{s_{in}}{\bigcup_{s\in C}\reach(d_s)}$. Since $G'$ satisfies the conditions
of~\cite[Theorem 3.2]{EKVY:multi-objectives}, we get a solution $\bar{y}$ to the linear program of \cite[Figure 3]{EKVY:multi-objectives} where for all
$C$ we have $\sum_{s\in C\cap S} \bar y_{d_s} = \Pr{\zeta}{s_{in}}{\staymec{C}}$. This solution gives us a solution to the
Inequalities \ref{eq:gl:ya} -- \ref{eq:gl:pos}
of the linear system $L$
of Figure~\ref{fig-L} by $y_t := \bar y_{d_t}$ for all $t\in S$, and $y_a = \bar y_{(s,a)}$ for all $a$
(note that the state $s$ is given uniquely as the state in which $a$ is enabled). Because $\bar{y}_{d_s} = y_t$, we get 
the required property that
$\sum_{t\in C \cap S} y_{t} =
\sum_{t\in C \cap S} y_{d_t} =
\Pr{\zeta}{s_{in}}{\staymec{C}}$.

\subsubsection{Proof of Lemma~\ref{lem-x_c-almost-surely}}\label{app-global-x_c-almost-surely}
Given a memoryless strategy $\sigma$ and an action $a$, we use $f_\sigma(a) = \Ex{\sigma}{s}{\lim_{i\rightarrow\infty} \frac{1}{i} I_a(A_i)}$ (where $I_a(a)=1$ and $I_a(b)=0$ for $a\neq b$) the frequency of action $a$.

Let $\sigma_1$ and $\sigma_2$ be memoryless deterministic strategies that
minimize and maximize the expectation, respectively, and only yield one BSCC for any initial state. Let $\sigma'$ be arbitrary memoryless randomized strategy that
visits every action in $C$ with nonzero frequency (such strategy clearly exists). We define the strategy $\sigma_{z_C}$
as follows. If $z_C = \sum_{a\in C\cap A} f_{\sigma'}(a) \cdot r(a)$, then $\sigma_{z_C}=\sigma'$. If $z_C > \sum_{a\in C\cap A} f_{\sigma'}(a) \cdot r(a)$,
then, because also $z_C \le \sum_{a\in C\cap A} f_{\sigma_2}(a) \cdot r(a)$, there must be a number $p\in (0,1]$
such that
\[
z_C= p\cdot \Big(\sum_{a\in C\cap A} f_{\sigma'}(a) \cdot r(a)\Big) + (1-p)\cdot\Big(\sum_{a\in C\cap A} f_{\sigma_2}(a) \cdot r(a)\Big)
\]
We define numbers $z_a=p\cdot f_{\sigma'}(a) + (1-p) \cdot f_{\sigma_2}(a)$ for all $a\in C\cap A$. Observe that we have, for any $s\in C$
\begin{eqnarray*}
\sum_{a\in C\cap A} z_{a}\cdot \delta(a)(s) &=& \sum_{a\in C\cap A}\Big(p\cdot f_{\sigma'}(a)\cdot\delta(a)(s) + (1-p) \cdot f_{\sigma_2}(a) \cdot\delta(a)(s)\Big)\\
&=& p\cdot\Big(\sum_{a\in C\cap A} f_{\sigma'}(a)\cdot \delta(a)(s)\Big) + (1-p) \cdot \Big(\sum_{a\in C\cap A} f_{\sigma_2}(a) \cdot \delta(a)(s)\Big)\\
&=& p\cdot\Big(\sum_{a\in \mathit{Act}(s)} f_{\sigma'}(a)\Big) + (1-p) \cdot \Big(\sum_{a\in \mathit{Act}(s)} f_{\sigma_2}(a)\Big)\\ 
&=& \sum_{a\in \mathit{Act}(s)}\Big(p\cdot f_{\sigma'}(a) + (1-p) \cdot f_{\sigma_2}(a)\Big)
\end{eqnarray*}
Hence, there is a memoryless randomized strategy $\sigma_{z_C}$ which visits $a$ with frequency $z_a$, hence giving the expectation
\[
 \Big(\sum_{a\in C\cap A} p\cdot f_{\sigma'}(a) \cdot r(a)\Big) + \Big(\sum_{a\in C\cap A} (1-p)\cdot f_{\sigma_2}(a) \cdot r(a)\Big)=\\
 p\cdot \Big(\sum_{a\in C\cap A} f_{\sigma'}(a) \cdot r(a)\Big) + (1-p)\cdot\Big(\sum_{a\in C\cap A} f_{\sigma_2}(a) \cdot r(a)\Big)
 = z_C
\]
For $z_C < \sum_{a\in C\cap A} f_{\sigma'}(a) \cdot r(a)$
we proceed similarly, this time combining $\sigma_C$ with $\sigma_1$ instead of $\sigma_2$.
\subsubsection{Showing that $\Va{\zeta}{s}{\lraname} \geq \Va{\zeta'}{s}{\lraname}$}\label{app-global-variance-improves}
Since by law of total variance $\Var(Z) = \Exp(\Var(Z|Y)) + \Var(\Exp(Z|Y))$ for all random variables $Y$, $Z$ we have for $\sigma\in \{\zeta,\zeta'\}$:
\[
 \Va{\sigma}{s}{\lraname} = \Big(\sum_{C\in \Mec{G}} \Pr{\sigma}{s}{\staymec{C}} \cdot \Va{\sigma}{s}{\lraname | \staymec{C}}\Big) + \Var(X)
\]
where $X$ is the random variable which to every MEC C assigns $\Ex{\sigma}{s}{\lraname |\staymec{C}}$.
Note that these random variables are equal for both $\zeta$ and $\zeta'$, and so also the second summands in the equation above are equal for $\zeta$ and $\zeta'$.
In the first summand, all the values 
$\Va{\zeta}{s}{\lraname | \staymec{C}}$ are nonnegative, while 
$\Va{\zeta'}{s}{\lraname | \staymec{C}}$ are zero. Hence the variance can only decrease when we go from $\zeta$ to $\zeta'$.
  
\subsubsection{From $\hat\sigma$ to $\sigma$}
In the construction of $\sigma$ we employ the following technical lemma.
\begin{lemma}\label{lemma:global-closer-preserves-variance}
Let $A$ be a finite set, $X,Y:A \rightarrow \Rset$ be random variables, $a_1, a_2\in A$ and $d>0$ a number satisfying the following:
\begin{itemize}
 \item For all $a\not\in \{a_1,a_2\}$: $X(a)=Y(a)$.
 \item $Y(a_1) \le Y(a_2)$
 \item $X(a_1) + d = Y(a_1)$
 \item $X(a_2) - \frac{\Prb(a_1)}{\Prb(a_2)}\cdot d = Y(a_2)$
\end{itemize}
Then $\Exp(X) = \Exp(Y)$ and $\Var(X) \ge \Var(Y)$.
\end{lemma}
\begin{proof}
Let us fix the following notation:
\begin{align*}
 \mu&= \Exp(X) &
 e_1&= X(a_1) &
 e_2&= X(a_2) &
 e_c&= \Exp(X\mid A\setminus\{a_1,a_2\})\\
 p_1&= \Prb(a_1) &
 p_2&= \Prb(a_2) &
 p_c&= \Prb(A\setminus\{a_1,a_2\})\\
\end{align*}

For expectation, we have 

{\small \begin{eqnarray*}
\Exp(X)
 &=&\Exp(X\mid A\setminus\{a_1,a_2\})\cdot p_c + \Exp(X\mid a_1)\cdot p_1 + \Exp(X\mid a_2)\cdot p_2\\
 &=&\Exp(Y\mid A\setminus\{a_1,a_2\})\cdot p_c + (\Exp(Y\mid a_1)-d)\cdot p_1 + (\Exp(Y\mid a_2) + \frac{p_1}{p_2}\cdot d)\cdot p_2\\
 &=&\Exp(Y\mid A\setminus\{a_1,a_2\})\cdot p_c + \Exp(Y\mid a_1)\cdot p_1 + \Exp(Y\mid a_2)\cdot p_2\\
 &=&\Exp(Y).
\end{eqnarray*}}

For variance, we need to show that
{\small \[
 \Exp((X - \mu)^2\mid A\setminus\{a_1,a_2\})\cdot p_c +
 \Exp((X - \mu)^2\mid a_1)\cdot p_1 + 
 \Exp((X - \mu)^2\mid a_2)\cdot p_2
 \ge
 \Exp((Y - \mu)^2\mid A\setminus\{a_1,a_2\})\cdot p_c+
 \Exp((Y - \mu)^2\mid a_1)\cdot p_1 + 
 \Exp((Y - \mu)^2\mid a_2)\cdot p_2
\]
}
which boils down to showing that
\[
 \Exp((X - \mu)^2\mid a_1)\cdot p_1 + 
 \Exp((X - \mu)^2\mid a_2)\cdot p_2
\ge
 \Exp((Y - \mu)^2\mid a_1)\cdot p_1 + 
 \Exp((Y - \mu)^2\mid a_2)\cdot p_2
\]
We have 
{\small
\begin{eqnarray*}
\Exp((Y - \mu)^2\mid a_1)\cdot p_1 +  \Exp((Y - \mu)^2\mid a_2)\cdot p_2
&=&p_1\cdot(e_1+d - \mu)^2 + p_2\cdot(e_2 - \frac{p_1}{p_2}\cdot d -\mu)^2\\
&=&p_1\cdot((e_1+d)^2 - 2\cdot(e_1+d)\cdot \mu + \mu^2)\\
&&+p_2\cdot((e_2 - \frac{p_1}{p_2}\cdot d)^2 - 2\cdot (e_2 - \frac{p_1}{p_2}\cdot d)\cdot\mu + \mu^2)\\
&=&p_1\cdot(e_1^2+ 2\cdot e_1\cdot d + d^2 - 2\cdot(e_1+d)\cdot \mu + \mu^2)\\
&& + p_2\cdot(e_2^2 - 2\cdot e_2\cdot \frac{p_1}{p_2}\cdot d + \frac{p_1^2}{p_2^2}\cdot d^2 - 2\cdot (e_2 - \frac{p_1}{p_2}\cdot d)\cdot\mu + \mu^2)\\
&=&p_1\cdot((e_1 - \mu)^2 + d^2 + 2\cdot e_1\cdot d - 2\cdot d\cdot \mu)\\
&&+ p_2\cdot((e_2 - \mu)^2 - 2\cdot e_2\cdot \frac{p_1}{p_2}\cdot d + \frac{p_1^2}{p_2^2}\cdot d^2 + 2\cdot \frac{p_1}{p_2}\cdot d\cdot\mu)\\
&=& p_1\cdot\Exp((X-\mu)^2\mid a_1) + p_2\cdot\Exp((X-\mu)^2\mid a_2)\\
&&+p_1\cdot(d^2 + 2\cdot e_1\cdot d - 2\cdot d\cdot \mu)+p_2\cdot(-2\cdot e_2\cdot \frac{p_1}{p_2}\cdot d + \frac{p_1^2}{p_2^2}\cdot d^2 + 2\cdot \frac{p_1}{p_2}\cdot d\cdot\mu)\\
\end{eqnarray*}
}
and so we need to show that the term on the last line is not positive. It is equal to

\[
p_1\cdot d^2 + p_1\cdot2\cdot e_1\cdot d - p_1\cdot2\cdot d\cdot \mu
- 2\cdot e_2\cdot p_1\cdot d + \frac{p_1^2}{p_2}\cdot d^2 + 2\cdot p_1\cdot d\cdot\mu
=p_1\cdot d^2 + p_1\cdot2\cdot (e_1-e_2)\cdot d + \frac{p_1^2}{p_2}\cdot d^2
\]
and hence we need to show that $d + 2(e_1-e_2) + \frac{p_1}{p_2}\cdot d$ is not positive, which is the case, because
by the assumption we have $(e_2-e_1) = Y(a_2) +\frac{p_1}{p_2}\cdot d - (Y(a_1) - d) \ge d+\frac{p_1}{p_2}\cdot d$.
\end{proof}

Let $\hat\sigma$ be the strategy from page~\pageref{page:hatsigma},
i.e. for every MEC $C$ there is a number $x_C$ such that $\lra{\omega} = x_C$ for almost every run from $R_C$.
Let us fix arbitrary $z$, and let $\violating{z}{\sigma}$ be the set of all the MECs which
satisfy:
\begin{itemize}
 \item If $\alpha_C > z$, then $x_C\neq\alpha_C$.
 \item If $\beta_C < z$, then $x_C\neq\beta_C$.
 \item Otherwise (if $\alpha_C \le z \le \beta_C$) we have $x_C\neq z$.
\end{itemize}

We create a sequence of strategies $\sigma_0,\sigma_1\ldots$ and numbers $z_0,z_1,\ldots$ by starting with $\sigma_0=\hat\sigma$, $z_0=z$ and creating
$\sigma_{k+1}$ and $z_{k+1}$ from $\sigma_k$ and $z_k$ as follows, finishing the sequence with a desired strategy $\sigma$. First, until possible, we repeat the following step.

  If there are MECs $C_i$ and $C_j$ in $\violating{z_k}{\sigma_k}$ such that $x_{C_i} < z$ and $x_{C_j} > z$, denote
  $p=\frac{\Pr{\sigma_k}{s}{R_{C_i}}}{\Pr{\sigma_k}{s}{R_{C_j}}}$
  and pick the maximal $d$ such that
  $d \le x_{C_i} - \max\{z,\alpha_{C_i}\}$ and $p\cdot d \le \min\{z,\beta_{C_j}\} - x_{C_j}$.
  We construct a $2$-memory
  strategy $\sigma_{k+1}$ that preserves the probabilities of $\sigma_k$ to reach each of the MECs, satisfies
  $\Ex{\sigma_{k+1}}{s}{\lraname\mid R_C} = \Ex{\sigma_k}{s}{\lraname\mid R_C}$ and $\Va{\sigma_{k+1}}{s}{\lraname\mid R_C}=0$ for every MEC $C$ different from
  $C_i$ and $C_j$, and also satisfies $\Ex{\sigma_{k+1}}{s}{\lraname\mid R_{C_i}}=v_{C_i} + d$ and $\Ex{\sigma_{k+1}}{s}{\lraname\mid R_{C_i}}=v_{C_j} - p\cdot d$.
  We also define $z_{k+1}=z_k$.
  By Lemma~\ref{lemma:global-closer-preserves-variance} the resulting strategy $\sigma_{k+1}$
  satisfies $\Ex{\sigma_{k+1}}{s}{\lraname} = \Ex{\sigma_k}{s}{\lraname}$ and $\Va{\sigma_{k+1}}{s}{\lraname} \le \Va{\sigma_k}{s}{\lraname}$.
  Also, $\violating{z_{k+1}}{\sigma_{k+1}}\subsetneq \violating{z_k}{\sigma_k}$, because
  one of the MECs $C_i$ and $C_j$ does not satisfy the defining condition of ${\cal C}$ and no new MEC satisfies it.

  Once it is not possible to perform the above, we either got $\violating{z_{k+1}}{\sigma_{k+1}}=\emptyset$ (in which case we put $\sigma=\sigma_{k+1}$
  and we are done) or exactly one of the following takes place:
  there is a MEC $C$ in $\violating{z_{k+1}}{\sigma_{k+1}}$ such that $x_C > z$ or there is a MEC $C$ in $\violating{z_{k+1}}{\sigma_{k+1}}$ such that $x_C < z$.
  Depending on which of these two happen, we continue building the sequence of strategies and numbers using one of the following items, until possible.
\begin{itemize}
 \item Suppose there is a MEC $C$ in $\violating{z_k}{\sigma_k}$ such that $x_C > z$. Let $\exact{z_k}{\sigma_k}$
  be the set of all MECs $C'$ such that $\Ex{\sigma_k}{s}{\lraname\mid R_{C'}}=z$ and $z\neq\beta_{C'}$,
  and let $p=\frac{\sum_{C'\in \exact{z_k}{\sigma_k}} \Pr{\sigma}{s}{R_{C'}}}{\Pr{\sigma}{s}{R_C}}$.
  Let us pick a maximal $d$ such that
  $p\cdot d \le x_C - \max\{z + p\cdot d,\alpha_C\}$ and
  $d \le \min\{\alpha_{C'} \mid C'\in D\} - z$. We construct a strategy $\sigma_{k+1}$
  so that it satisfies $\Va{\sigma_{k+1}}{s}{\lraname\mid R_{C'}}=0$ for every MEC $C'$,
  $\Ex{\sigma_{k+1}}{s}{\lraname\mid R_{C'}} = \Ex{\sigma_{k}}{s}{\lraname\mid R_{C'}}$ for every MEC $C'\not\in \exact{z_k}{\sigma_k} \cup \{C\}$
  and also satisfies $\Ex{\sigma_{k+1}}{s}{\lraname\mid R_{C}}=v_{C} - p\cdot d$ and
  $\Ex{\sigma'}{s}{\lraname\mid R_{C'}}=v_{C''} + d$
  for all $C'\in \exact{z_k}{\sigma_k}$. By Lemma~\ref{lemma:global-closer-preserves-variance} the resulting strategy
  satisfies $\Ex{\sigma_{k+1}}{s}{\lraname} = \Ex{\sigma_k}{s}{\lraname}$ and $\Va{\sigma'}{s}{\lraname} \le \Va{\sigma_k}{s}{\lraname}$.

  One of the following also takes place:
  \begin{itemize}
   \item $\violating{z_{k+1}}{\sigma_{k+1}}\subsetneq\violating{z_{k+1}}{\sigma_{k+1}}$, because $C\not\in\violating{z_{k+1}}{\sigma_{k+1}}$.
   \item $\violating{z_{k+1}}{\sigma_{k+1}}=\violating{z_{k+1}}{\sigma_{k+1}}$ and $\exact{z_{k+1}}{\sigma_{k+1}}\subsetneq\exact{z_{k+1}}{\sigma_{k+1}}$
  \end{itemize}
  We set $z_{k+1} = z_k$ and continue, if possible.
 \item If there is a MEC $C$ such that $x_C < z$ we proceed similarly as in the above item.
\end{itemize}
Note that the above procedure eventually terminates, because in every step either $\violating{z_{i+1}}{\sigma_{i+1}}\subseteq \violating{z_{i}}{\sigma_{i}}$,
and for $m =|\Mec(G)|$ we have $\violating{z_{i+m}}{\sigma_{i+m}} \subsetneq \violating{z_{i+1}}{\sigma_{i+1}}$, because if $\violating{z_{i+1}}{\sigma_{i+1}}= \violating{z_{i}}{\sigma_{i}}$, then $\exact{z_{i+1}}{\sigma_{i+1}}\subsetneq \exact{z_{i}}{\sigma_{i}}$ and $|\exact{\cdot}{\cdot}|\le m$.

\subsubsection{Solving $L_{\hat z}$ in polynomial time.}
\label{app:global-approx-matrix}

\begin{lemma}
Let $n \in \Nset$ and $m_i \in \Nset$ for every $1 \leq i \leq n$.
For all $1 \leq i \leq n$ and $1 \leq j \leq m_i$, we use
$\tp{i,j}$ to denote the index $j + \sum_{\ell=1}^{i-1} m_\ell$.
Consider a function $f : \Rset^k \rightarrow \Rset$, where
$k = \sum_{i=1}^n m_i$, of the form 
\[
  f(\vec{v}) \quad = \quad \left( \sum_{i=1}^n  
     \bigg( \vec{c}_i^2 \cdot \sum_{j=1}^{m_i} \vec{v}_{\tp{i,j}} \bigg) \right )
  \quad - \quad  
  \left( \sum_{i=1}^n \bigg( \vec{c}_i \cdot  
         \sum_{j=1}^{m_i} \vec{v}_{\tp{i,j}}\bigg) \right)^2
\]
where $\vec{c} \in \Rset^n$. Then
$f(\vec{v})$ can be written as
$
   f(\vec{v}) \ = \ \vec{v}^T Q\, \vec{v} \ +\  \vec{d}^T \vec{v}
$
where $Q$ is a negative semi-definite matrix of rank~$1$ and
$\vec{d} \in \Rset^k$.
Consequently, $f(\vec{v})$ is concave and $Q$ has exactly one eigenvalue.
\end{lemma}
\begin{proof}
  Observe that every vector $\vec{u} \in \Rset^k$ can be written as
  $\vec{u}^T = (\vec{u}_{\tp{1,1}},\ldots,\vec{u}_{\tp{1,m_1}},\ \cdots, 
  \vec{u}_{\tp{n,1}},\ldots,\vec{u}_{\tp{1,m_n}})$. 
  Let $Q$ be $k \times k$ matrix where 
  $
     Q_{\tp{i,j},\tp{i',j'}} = - (c_{i'} \cdot c_{i})
  $. 
  Then
  \[ 
     (Q\, \vec{v})_{\tp{i,j}} \quad = \quad 
         \sum_{i'=1}^{n}\sum_{j'=1}^{m_{i'}} Q_{\tp{i,j},\tp{i',j'}} 
               \cdot \vec{v}_{\tp{i',j'}} \quad = \quad 
         - \sum_{i'=1}^n\sum_{j'=1}^{m_{i'}} (c_{i'} \cdot c_i) \vec{v}_{\tp{i',j'}}
  \]
  and consequently
  \[
     \vec{v}^T Q\, \vec{v} \quad = \quad
         - \sum_{i=1}^n \sum_{j=1}^{m_{i}} \vec{v}_{\tp{i,j}} \cdot
         \left( 
           \sum_{i'=1}^n \sum_{j'=1}^{m_{i'}} (c_{i'} \cdot c_i) \vec{v}_{\tp{i',j'}} 
         \right)
      \quad = \quad
         - \sum_{i=1}^n \sum_{i'=1}^n (c_i \cdot c_{i'}) \cdot
           \sum_{j=1}^{m_{i}} \vec{v}_{\tp{i,j}} \cdot 
           \sum_{j'=1}^{m_{i'}} \vec{v}_{\tp{i',j'}}
       \quad= \quad  - \left( \sum_{i=1}^n \bigg( \vec{c}_i \cdot  
         \sum_{j=1}^{m_i} \vec{v}_{i,j}\bigg) \right)^2
  \]
  Hence, $f(\vec{v}) \ = \ \vec{v}^T Q\, \vec{v} \ +\  \vec{d}^T \vec{v}$,
  where $\vec{d}_{\tp{i,j}} = c_i^2$. Let $\vec{u}\in \Rset^k$ be a 
  (fixed) vector such that $\vec{u}_{\tp{i,j}} = - c_{i}$. Then the
  $\tp{i',j'}$-th column of $Q$ is equal to 
  $c_{i'} \cdot \vec{u}$, which means that the rank of $Q$ is $1$. 
  The matrix $Q$ is negative semi-definite because 
  $\vec{v}^T Q\, \vec{v} \leq 0$ for every $\vec{v}\in \Rset^k$. 
\end{proof}

\subsubsection{Correctness of the approximation algorithm.}
\label{app:global-approx-cor}
Assume there is a strategy $\sigma$ such that
$(\Ex{\sigma}{s}{\lraname}, \Va{\sigma}{s}{\lraname}) \leq (u-\varepsilon,v-\varepsilon)$, and let $z$
be the number from Item~2, and let us fix a valuation $\bar y_\kappa$ for the variables $y_\kappa$ where $\kappa \in S\cup A$
from equations of the system $L$ (see Figure~\ref{fig-L}). Let $\bar z$ be a number between the minimal and the maximal assigned reward
that is a multiple of $\tau$, and
which satisfies $|z-\bar z| < \tau$. Such a number must exist. We show that the system $L_{\bar z}$ has a solution.
The valuation $\bar y_\kappa$ can be applied to the system $L_{\bar z}$, and
we get
\begin{eqnarray*}
 \sum_{C \in \Mec(G)} x_{C,\bar z} \cdot \sum_{t\in S\cap C} y_t
  & = & \Big(\sum_{C \in \Mec(G)} x_{C,z} \cdot \sum_{t\in S\cap C} y_t\Big) + \Big(\sum_{C \in \Mec(G)} (x_{C, \bar z} - x_{C, z}) \cdot \sum_{t\in S\cap C} y_t\Big)\\
  & \le & (u-\varepsilon) + \Big(\sum_{C \in \Mec(G)} \tau \cdot \sum_{t\in S\cap C} y_t\Big)\\
  & \le & (u-\varepsilon) + \Big(\sum_{C \in \Mec(G)} \tau \cdot \sum_{t\in S\cap C} y_t\Big)\\
  & \le & (u-\varepsilon) + \tau \le u
\end{eqnarray*}
For variance, we have that 
\begin{eqnarray*}
 \left(\sum_{C \in \Mec(G)} x_{C,\bar z}^2 \cdot \sum_{t\in S\cap C}y_t\right)
 & = & 
 \left(\sum_{C \in \Mec(G)}  \big(x_{C,z} + (x_{C, \bar z} - x_{C,z})\big)^2 \cdot \sum_{t\in S\cap C}y_t\right)\\
 & = & 
 \left(\sum_{C \in \Mec(G)} x_{C,z}^2 \cdot \sum_{t\in S\cap C}y_t\right) + 
  \left(\sum_{C \in \Mec(G)} (2\cdot x_{C,z} \cdot (x_{C,\bar z} - x_{C,z}) + (x_{C,\bar z} - x_{C,z})^2) \cdot \sum_{t\in S\cap C}y_t\right)\\
 & \le & 
 \left(\sum_{C \in \Mec(G)} x_{C,z}^2 \cdot \sum_{t\in S\cap C}y_t\right) + 
  \left(\sum_{C \in \Mec(G)} (2\cdot x_{C,z} \cdot \tau + \tau^2) \cdot \sum_{t\in S\cap C}y_t\right)\\
 & \le & 
 \left(\sum_{C \in \Mec(G)} x_{C,z}^2 \cdot \sum_{t\in S\cap C}y_t\right) + 
  \left(\sum_{C \in \Mec(G)} (2\cdot N \cdot \tau + \tau^2) \cdot \sum_{t\in S\cap C}y_t\right)\\
 & \le & 
 \left(\sum_{C \in \Mec(G)} x_{C,z}^2 \cdot \sum_{t\in S\cap C}y_t\right) + 
  2\cdot N \cdot \tau + \tau^2
\end{eqnarray*}
and
\begin{eqnarray*}
    \left(\sum_{C \in \Mec(G)} x_{C,\bar z} \cdot \sum_{t \in S \cap C}y_t\right)^2
 & = & 
    \left(\sum_{C \in \Mec(G)} \Big(x_{C,z} + (x_{C,\bar z} - x_{C,z})\Big) \cdot \sum_{t \in S \cap C}y_t\right)^2\\
 & = & 
    \left(\Big(\sum_{C \in \Mec(G)} x_{C,z} \cdot \sum_{t \in S \cap C}y_t\Big) +\Big(\sum_{C \in \Mec(G)}(x_{C,\bar z} - x_{C,z}) \cdot \sum_{t \in S \cap C}y_t\Big)\right)^2\\
 & \ge & 
    \left(\Big(\sum_{C \in \Mec(G)} x_{C,z} \cdot \sum_{t \in S \cap C}y_t\Big) - \Big(\sum_{C \in \Mec(G)}\tau \cdot \sum_{t \in S \cap C}y_t\Big)\right)^2\\
 & = & 
    \left(\Big(\sum_{C \in \Mec(G)} x_{C,z} \cdot \sum_{t \in S \cap C}y_t\Big) - \tau \right)^2\\
 & = & 
    \left(\sum_{C \in \Mec(G)} x_{C,z} \cdot \sum_{t \in S \cap C}y_t\right)^2
     - 2\cdot\Big(\sum_{C \in \Mec(G)} x_{C,z} \cdot \sum_{t \in S \cap C}y_t\Big)\cdot \tau + \tau^2\\
 & \ge & 
    \left(\sum_{C \in \Mec(G)} x_{C,z} \cdot \sum_{t \in S \cap C}y_t\right)^2
     - 2\cdot N\cdot \tau + \tau^2
\end{eqnarray*}
and so we get
\begin{eqnarray*}
 \left(\sum_{C \in \Mec(G)} \hat x_{C,\bar z}^2 \cdot \sum_{t\in S\cap C}y_t\right) - 
    \left(\sum_{C \in \Mec(G)} x_{C,\bar z} \cdot \sum_{t \in S \cap C}y_t\right)^2
 &\le&
 \left(\sum_{C \in \Mec(G)} \hat x_{C,z}^2 \cdot \sum_{t\in S\cap C}y_t\right)
  - \left(\sum_{C \in \Mec(G)} x_{C,z} \cdot \sum_{t \in S \cap C}y_t\right)^2\\
 &&\quad
  + 2\cdot N \cdot \tau + \tau^2 + 2\cdot N\cdot \tau + \tau^2\\
 &\le& v - \varepsilon + \varepsilon \le v
\end{eqnarray*}
Hence we have shown that there is a solution for $L_{\bar z}$, and so the algorithm returns ``yes''.

On the other hand, if there is no strategy
such that $(\Ex{\sigma}{s}{\lraname},\Va{\sigma}{s}{\lraname}) \leq (u,v)$, then the algorithm clearly returns ``no''.

\subsection{Proofs for Local Variance}
\label{app-local}
\subsubsection{Computation for Example~\ref{ex:local-mem}}\label{app:local-example}
We have
\begin{eqnarray*}
 \Ex{\sigma'}{s_1}{\lrvname} &=&  f(a)(0-\Ex{\sigma'}{s_1}{\lraname})^2 +  (f(b)+f(c))(2-\Ex{\sigma'}{s_1}{\lraname})^2\\
 &=& f(a)(-2+2f(a)))^2 + (1 - f(a))(2f(a))^2\\
 &=& 4f(a) - 8f(a)^2 + 4f(a)^3 + 4f(a)^2 - 4f(a)^3\\
 &=& 4f(a) - 4f(a)^2 \ge 0.64
\end{eqnarray*}

Throughout this section we use the following three simple lemmas. The first one allows us to reduce convex combinations of two-dimensional vectors (typically vectors consisting of the mean-payoff and variance) to combinations of just two vectors.
\begin{lemma}\label{lem:approx-two}
Let $(a_1,b_1), (a_2,b_2),\ldots, (a_m,b_m)$ be a sequence of points in $\Rset^2$ and $c_1,c_2,\ldots , c_m \in (0,1]$ satisfy $\sum_{i=1}^{m} c_i=1$. Then there are two vectors $(a_k,b_k)$ and $(a_{\ell},b_{\ell})$ and a number $p\in [0,1]$ such that
\[
\sum_{i=1}^{m} c_i (a_i,b_i)\quad  \geq \quad p (a_k,b_k) + (1-p) (a_{\ell},b_{\ell})
\]
\end{lemma}
\begin{proof}
Denote by $(x,y)$ the point $\sum_{i=1}^{m} c_i (a_i,b_i)$ and by $H$ the set $\{(a_i,b_i)\mid 1\leq i\leq m\}$.
If all the points of $H$ lie in the same line, then clearly there must be some $(a_k,b_k)\leq (x,y)$.
Assume that this is not true. Then the convex hull $\mathcal{C}(H)$ of $H$ is a convex polygon whose vertices are some of the points of $H$. Consider a point $(x',y)$ where $x'=\min\{z\mid z\leq x, (z,y)\in \mathcal{C}(H)\}$. The point $(x',y)$ lies on the boundary of $\mathcal{C}(H)$ and thus, as $\mathcal{C}(H)$ is a convex polygon, $(x',y)$ lies on the line segment between two vertices, say $(a_k,b_k),(a_{\ell},b_{\ell})$, of $\mathcal{C}(H)$. Thus there is $p\in [0,1]$ such that
\[(x',y)=p(a_k,b_k)+(1-p)(a_{\ell},b_{\ell})\leq (x,y)=\sum_{i=1}^{m} c_i (a_i,b_i)\,.\]
This finishes the proof.
\end{proof}
\noindent
The following lemma shows how to minimize the mean square deviation (to which our notion of variance is a special case).
\begin{lemma}\label{lem:min-var}
Let $a_1,\ldots,a_m\in \Rset$ such that $\sum_{i=0}^{m} a_i = 1$, let $r_1,\ldots,r_m\in \Rset$ and let us consider the following function of one real variable:
\[
V(x)=\sum_{i=1}^m a_i \left(r_i - x\right)^2
\]
Then the function $V$ has a unique minimum in $\sum_{i=1}^m a_i r_i$.
\end{lemma}
\begin{proof}
By taking the first derivative of $V$ we obtain
\[
\frac{\delta V}{\delta x} = -2\cdot  \sum_{i=1}^m a_i \left(r_i - x\right)
  = -2\cdot \left(\sum_{i=1}^m a_i r_i\right)+2x
\]
Thus $\frac{\delta{V}}{\delta x}(x)=0$ iff $x=\sum_{i=1}^m a_i r_i$.
Moreover, by taking the second derivative we obtain $\frac{\delta^2 V}{\delta x^2}=2>0$, and thus $\sum_{i=1}^m a_i r_i$ is a minimum.
\end{proof}
The following lemma shows that frequencies of actions determine (in some cases) the mean-payoff as well as the variance.
\begin{lemma}\label{lem:freq-var}
Let $\mu$ be a memoryless strategy and let $D$ be a BSCC of $G^{\mu}$. Consider frequencies of individual actions $a\in \BSCCact{D}$ when starting in a state $s\in \BSCCstate{D}$: $\Ex{\mu}{s}{\lraname^{I_a}}$ where $I_a$ assigns $1$ to $a$ and $0$ to all other actions (note that the values do not depend on which $s$ we choose).
Then $\Ex{\mu}{s}{\lraname^{I_a}}$ determine uniquely all of $\Ex{\mu}{s}{\lraname}$, $\Ex{\mu}{s}{\lrvhname}$, and $\Ex{\mu}{s}{\lrvlname}$ as follows: 
\[
\Ex{\mu}{s}{\lraname}=\sum_{a\in A} r(a)\cdot \Ex{\mu}{s}{\lraname^{I_a}}
\qquad
\text{and}
\qquad
\Ex{\mu}{s}{\lrvhname}=\Ex{\mu}{s}{\lrvlname}=\sum_{a\in A} (r(a)-\Ex{\mu}{s}{\lraname})^2\cdot \Ex{\mu}{s}{\lraname^{I_a}}
\]
\end{lemma}
\begin{proof}
We have
\[
\Ex{\mu}{s}{\lraname} = \Ex{\mu}{s}{\lim_{i\rightarrow \infty} \frac{1}{i}\cdot \sum_{j=1}^i r(A_j)} 
   =  \Ex{\mu}{s}{\lim_{i\rightarrow \infty} \frac{1}{i}\cdot \sum_{j=1}^i \sum_{a\in A} r(a) I_a(A_j)} 
   =  \sum_{a\in A} r(a)\cdot \Ex{\mu}{s}{\lim_{i\rightarrow \infty} \frac{1}{i}\cdot \sum_{j=1}^i I_a(A_j)} 
   =  \sum_{a\in A} r(a)\cdot \Ex{\mu}{s}{\lraname^{I_a}}
\]
and
\begin{multline*}
\Ex{\mu}{s}{\lrvhname} = 
\Ex{\mu}{s}{\lim_{i\rightarrow \infty} \frac{1}{i}\cdot \sum_{j=1}^i (r(A_j)-\Ex{\mu}{s}{\lraname})^2} 
  = 
  \Ex{\mu}{s}{\lim_{i\rightarrow \infty} \frac{1}{i}\cdot \sum_{j=1}^i \sum_{a\in A} (r(a)-\Ex{\mu}{s}{\lraname})^2\cdot I_a(A_j)} \\
  = \sum_{a\in A} (r(a)-\Ex{\mu}{s}{\lraname})^2 \cdot 
  \Ex{\mu}{s}{\lim_{i\rightarrow \infty} \frac{1}{i}\cdot \sum_{j=1}^i I_a(A_j)} 
  = \sum_{a\in A} (r(a)-\Ex{\mu}{s}{\lraname})^2\cdot \Ex{\mu}{s}{\lraname^{I_a}}
\end{multline*}
Finally, it is easy to see that the local and hybrid variance coincide in BSCCs since almost all runs have the same frequencies of actions. This gives us the result for the local variance.
\end{proof}

\medskip
\noindent
\subsubsection{Proof of Proposition~\ref{prop:strong-opt-local}.}\label{app-strong-opt}
We obtain the proof from the following slightly weaker version.
\begin{proposition}\label{prop:eps-opt-local}
Let us fix a MEC $C$ and let $\varepsilon>0$. There are two frequency functions $f_{\varepsilon}:\MECact{C}\rightarrow [0,1]$ and $f'_{\varepsilon}:\MECact{C}\rightarrow [0,1]$, and a number $p_{\varepsilon}\in [0,1]$ such that:
\[
p_{\varepsilon}\cdot (\lraname[f_{\varepsilon}],\lrvname[f_\varepsilon])+(1-p_{\varepsilon})\cdot
 (\lraname[f'_{\varepsilon}],\lrvname[f'_\varepsilon])\quad \leq \quad
(\Ex{\zeta}{s_0}{\lraname}, \Ex{\zeta}{s_0}{\lrvlname})+(\varepsilon,\varepsilon)
\]
\end{proposition}
\noindent
Before we prove Proposition~\ref{prop:eps-opt-local}, let us show that it indeed implies Proposition~\ref{prop:strong-opt-local}.
There is a sequence $\varepsilon_1,\varepsilon_2,\ldots$, two functions $f_C$ and $f'_C$, and $p_C\in [0,1]$ such that as $n\rightarrow \infty$ 
\begin{itemize}
\item $\varepsilon_n\rightarrow 0$ 
\item $f_{\varepsilon_n}$ converges pointwise to $f_C$
\item $f'_{\varepsilon_n}$ converges pointwise to $f'_C$
\item $p_{\varepsilon_n}$ converges to $p_C$
\end{itemize}
It is easy to show that $f_C$ as well as $f'_C$ are frequency functions. Moreover, 
as 
\[
\lim_{n\rightarrow \infty} (\Ex{\zeta}{s_0}{\lraname}, \Ex{\zeta}{s_0}{\lrvlname})+(\varepsilon_n,\varepsilon_n)=
(\Ex{\zeta}{s_0}{\lraname}, \Ex{\zeta}{s_0}{\lrvlname})
\]
and
\[
\lim_{n\rightarrow \infty}  p_{\varepsilon_n}\cdot (\lraname[f_{\varepsilon_n}],\lrvname[f_{\varepsilon_n}])+(1-p_{\varepsilon_n})\cdot
 (\lraname[f'_{\varepsilon_n}],\lrvname[f'_{\varepsilon_n}])
  =  p_C\cdot (\lraname[f_C],\lrvname[f_C])+(1-p_C)\cdot
 (\lraname[f'_C],\lrvname[f'_C])
\]
we obtain
\[
p_C\cdot (\lraname[f_C],\lrvname[f_C])+(1-p_C)\cdot
 (\lraname[f'_C],\lrvname[f'_C])\le (\Ex{\zeta}{s_0}{\lraname}, \Ex{\zeta}{s_0}{\lrvlname})
\]
This finishes a proof of Proposition~\ref{prop:strong-opt-local}. It remains to prove Proposition~\ref{prop:eps-opt-local}.

\medskip\noindent
\begin{proof}[Proof of Proposition~\ref{prop:eps-opt-local}.]
Given $\ell,k\in \Zset$ we denote by $A^{\ell,k}$ the set of all runs $\omega\in R_C$ such that
\[
(\ell\cdot \varepsilon,k\cdot \varepsilon) \quad \leq \quad (\lraname(\omega),\lrvlname(\omega))\quad < \quad (\ell\cdot \varepsilon,k\cdot \varepsilon)+(\varepsilon,\varepsilon)
\]
Note that 
\[
\sum_{\ell,k\in \Zset} \Prb_{s_0}^{\zeta}(A^{\ell,k}|R_C)\cdot (\ell\cdot \varepsilon,k\cdot \varepsilon) \quad \leq \quad
(\Ex{\zeta}{s_0}{\lraname|R_C}, \Ex{\zeta}{s_0}{\lrvlname|R_C})
\]
By Lemma~\ref{lem:approx-two}, there are $\ell,k,\ell',k'\in \Zset$ and $p\in [0,1]$ such that $\Prb_{s_0}^{\zeta}(A^{\ell,k}|R_C)>0$ and $\Prb_{s_0}^{\zeta}(A^{\ell',k'}|R_C)>0$ and 
\begin{equation}\label{eq:approx-two}
p\cdot (\ell\cdot \varepsilon,k\cdot \varepsilon)+(1-p)\cdot (\ell'\cdot \varepsilon,k'\cdot \varepsilon)\leq \sum_{\ell,k\in \Zset} \Prb_{s_0}^{\zeta}(A^{\ell,k}|R_C)\cdot (\ell\cdot \varepsilon,k\cdot \varepsilon)\leq
(\Ex{\zeta}{s_0}{\lraname|R_C}, \Ex{\zeta}{s_0}{\lrvlname|R_C})
\end{equation}
Let us concentrate on $(\ell\cdot \varepsilon,k\cdot \varepsilon)$ and construct a frequency function $f$ on $C$ such that
\[
 (\lraname[f],\lrvname[f])
\quad \leq \quad (\ell\cdot \varepsilon,k\cdot \varepsilon)+(\varepsilon,\varepsilon)
\]
Intuitively, we obtain $f$ as a vector of frequencies of individual actions on an appropriately chosen run of $R_C$. Such frequencies determine the average and variance close to $\ell\cdot \varepsilon$ and $k\cdot \varepsilon$, respectively. We have to deal with some technical issues, mainly with the fact that the frequencies might not be well defined for almost all runs (i.e. the corresponding limits might not exist). This is solved by a careful choice of subsequences as follows.
\begin{claim}\label{claim:subsequence-local}
For every run $\omega\in R_C$ %
there is a sequence of numbers $T_1[\omega],T_2[\omega],\ldots$ such that
all the following limits are defined:
\[
\lim_{i\rightarrow \infty} \frac{1}{T_i[\omega]} \sum_{j=1}^{T_i[\omega]} r(A_j(\omega))\quad = \quad \lraname(\omega)
\qquad
\text{and}
\qquad
\lim_{i\rightarrow \infty} \frac{1}{T_i[\omega]} \sum_{j=1}^{T_i[\omega]} (r(A_j(\omega))-\lraname(\omega))^2\quad \le \quad \lrvlname(\omega)
\]
and for every action $a\in A$ there is a number $f_{\omega}(a)$ such that
\[
\lim_{i\rightarrow \infty} \frac{1}{T_i[\omega]} \sum_{j=1}^{T_i[\omega]} I_a(A_j(\omega))\quad = \quad f_{\omega}(a)
\]
(Here $I_a(A_j(\omega))=1$ if $A_j(\omega)=a$, and $I_a(A_j(\omega))=0$ otherwise.)

Moreover, for almost all runs $\omega$ of $R_C$ we have that $f_{\omega}$ is a frequency function on $C$ and that 
$f_{\omega}$ determines $(\lraname(\omega),\lrvlname(\omega))$, i.e.,
$\lraname(\omega)=\lraname(f_{\omega})$ and $\lrvlname(\omega)\ge\lrvname(f_{\omega})$.
\end{claim}
\begin{proof}
We start by taking a sequence $T'_1[\omega],T'_2[\omega],\ldots$ such that
\[
\lim_{i\rightarrow \infty} \frac{1}{T'_i[\omega]} \sum_{j=1}^{T'_i[\omega]} r(A_j(\omega))\quad = \quad \lraname(\omega)
\]
Existence of such a sequence follows from the fact that every sequence of real numbers has a subsequence which converges to the lim sup of the original sequence.

Now we extract a subsequence $T''_1[\omega],T''_2[\omega],\ldots$ of $T'_1[\omega],T'_2[\omega],\ldots$ such that 
\begin{equation}\label{eq:subsequence-local}
\lim_{i\rightarrow \infty} \frac{1}{T''_i[\omega]} \sum_{j=1}^{T''_i[\omega]} (r(A_j(\omega))-\lraname(\omega))^2\quad \le \quad \lrvlname(\omega)
\end{equation}
using the same argument.

Now assuming an order on actions, $a_1,\ldots,a_m$, we define $T^{k}_1[\omega],T^{k}_2[\omega],\ldots$ for $0\leq k\leq m$ so that $T^{0}_1[\omega],T^{0}_2[\omega],\ldots$ is the sequence $T''_1[\omega],T''_2[\omega],\ldots$, and
every $T^{k+1}_1[\omega],T^{k+1}_2[\omega],\ldots$ is a subsequence of $T^{k}_1[\omega],T^{k}_2[\omega],\ldots$ 
such that the following limit exists (and is equal to a number $f_{\omega}(a_{k+1})$)
\[
\lim_{i\rightarrow \infty} \frac{1}{T^{k+1}_i[\omega]} \sum_{j=1}^{T^{k+1}_i[\omega]} I_{a_{k+1}}(A_j(\omega))
\]
We take $T^{m}_1[\omega],T^{m}_2[\omega],\ldots$ to be the desired sequence $T_1[\omega],T_2[\omega],\ldots$.

Now we have to prove that $f_{\omega}$ is a frequency function on $C$ for almost all runs of $R_C$. Clearly, $0\leq f_{\omega}(a)\leq 1$ for all $a\in \MECact{C}$. Also,
\[
\sum_{a\in \MECact{C}} f_{\omega}(a)=\sum_{a\in \MECact{C}} \lim_{i\rightarrow \infty} \frac{1}{T_i[\omega]} \sum_{j=1}^{T_i[\omega]} I_{a}(A_j(\omega))=\lim_{i\rightarrow \infty} \frac{1}{T_i[\omega]} \sum_{j=1}^{T_i[\omega]} \sum_{a\in \MECact{C}}I_{a}(A_j(\omega))
=\lim_{i\rightarrow \infty} \frac{1}{T_i[\omega]} \sum_{j=1}^{T_i[\omega]} 1=1
\]
To prove the third condition from the definition of frequency functions, we invoke the law of large numbers (SLLN) \cite{Billingsley:book}. Given a run $\omega$, an action $a$, a state $s$ and $k\geq 1$, define
\[
N^{a,s}_k(\omega)=\begin{cases}
  1 & \text{ $a$ is executed at least $i$ times, and $s$ is visited just after the $i$-th execution of $a$; }\\
  0 & \text{ otherwise.}
\end{cases}
\]
By SLLN and by the fact that in every step the distribution on the next states depends just on the chosen action, for almost all runs $\omega$ the following limit is defined and the equality holds whenever $f_\omega(a) > 0$:
\[
\lim_{j\rightarrow \infty} \frac{\sum_{k=1}^j N^{a,s}_k(\omega)}{j} = \delta(a)(s)
\]
We obtain
\begin{eqnarray*}
\sum_{a\in \MECact{C}} f_{\omega}(a)\cdot \delta(a)(s) & = &  \sum_{a\in \MECact{C}} \lim_{i\rightarrow \infty} \frac{1}{T_i[\omega]} \sum_{j=1}^{T_i[\omega]} I_a(A_j(\omega))\cdot   
  \lim_{i\rightarrow \infty} \frac{1}{i}\sum_{k=1}^{i} N^{a,s}_k(\omega) \\
& = & \sum_{a\in \MECact{C}} \lim_{i\rightarrow \infty} \frac{1}{T_i[\omega]} \sum_{j=1}^{T_i[\omega]} I_a(A_j(\omega))\cdot \lim_{i\rightarrow \infty} \frac{1}{\sum_{j=1}^{T_i[\omega]} I_a(A_j(\omega))}\sum_{k=1}^{\sum_{j=1}^{T_i[\omega]} I_a(A_j(\omega))} N^{a,s}_k(\omega) \\
& = & \sum_{a\in \MECact{C}} \lim_{i\rightarrow \infty} \frac{1}{T_i[\omega]}\sum_{k=1}^{\sum_{j=1}^{T_i[\omega]} I_a(A_j(\omega))} N^{a,s}_k(\omega) \\
& = & \lim_{i\rightarrow \infty} \frac{1}{T_i[\omega]}\sum_{a\in \MECact{C}} \sum_{k=1}^{\sum_{j=1}^{T_i[\omega]} I_a(A_j(\omega))} N^{a,s}_k(\omega) \\
& = & \lim_{i\rightarrow \infty} \frac{1}{T_i[\omega]}\sum_{j=1}^{T_i[\omega]} I_s(S_j(\omega)) \\
& = & \lim_{i\rightarrow \infty} \frac{1}{T_i[\omega]}\sum_{j=1}^{T_i[\omega]} \sum_{a\in \mathit{Act}(s)} I_a(A_j(\omega)) \\
& = & \sum_{a\in \mathit{Act}(s)}\lim_{i\rightarrow \infty} \frac{1}{T_i[\omega]}\sum_{j=1}^{T_i[\omega]}  I_a(A_j(\omega)) \\
& = & \sum_{a\in \mathit{Act}(s)} f_{\omega}(a)
\end{eqnarray*}
Here $S_j(\omega)$ is the $j$-th state of $\omega$, and $I_s(t)=1$ for $s=t$ and $I_s(t)=0$ otherwise.
\begin{eqnarray*}
\lraname(\omega) & = & \lim_{i\rightarrow \infty} \frac{1}{T_i[\omega]} \sum_{j=1}^{T_i[\omega]} r(A_j(\omega)) \\
& = & \lim_{i\rightarrow \infty} \frac{1}{T_i[\omega]} \sum_{j=1}^{T_i[\omega]} \sum_{a\in \MECact{C}} I_a(A_j(\omega))\cdot r(a) \\
& = & \sum_{a\in \MECact{C}} r(a)\cdot \lim_{i\rightarrow \infty} \frac{1}{T_i[\omega]} \sum_{j=1}^{T_i[\omega]} I_a(A_j(\omega)) \\
& = & \sum_{a\in \MECact{C}} r(a)\cdot f_{\omega}(a)\\
& = & \lraname[f_{\omega}]
\end{eqnarray*}
\begin{eqnarray*}
\lrvlname(\omega) & \ge & \lim_{i\rightarrow \infty} \frac{1}{T_i[\omega]} \sum_{j=1}^{T_i[\omega]} (r(A_j(\omega))-\lraname(\omega))^2 \\
& = & \lim_{i\rightarrow \infty} \frac{1}{T_i[\omega]} \sum_{j=1}^{T_i[\omega]} \sum_{a\in \MECact{C}} I_a(A_j(\omega))\cdot (r(a)-\lraname(\omega))^2 \\
& = & \sum_{a\in \MECact{C}} (r(a)-\lraname(\omega))^2 \cdot \lim_{i\rightarrow \infty} \frac{1}{T_i[\omega]} \sum_{j=1}^{T_i[\omega]} I_a(A_j(\omega)) \\
& = & \sum_{a\in \MECact{C}} (r(a)-\lraname(\omega))^2 \cdot f_{\omega}(a)\\
& = & \lrvname[f_{\omega}]
\end{eqnarray*}
\end{proof}
\noindent
Now pick an arbitrary run $\omega$ of $A^{k,\ell}$ such that $f_{\omega}$ is a frequency function. Then
\[
(\lraname(f_{\omega}),\lrvname(f_{\omega}))\le(\lraname(\omega),\lrvlname(\omega))\leq (\ell\cdot \varepsilon,k\cdot\varepsilon)+(\varepsilon,\varepsilon)
\]
Similarly, for $\ell',k'$ we obtain $f'_{\omega}$ such that 
\[
(\lraname(f'_{\omega}),\lrvname(f'_{\omega}))\le(\lraname(\omega),\lrvlname(\omega))\leq (\ell'\cdot \varepsilon,k'\cdot\varepsilon)+(\varepsilon,\varepsilon)
\]
This together with the equation~(\ref{eq:approx-two}) from page~\pageref{eq:approx-two} proves Proposition~\ref{prop:eps-opt-local}:
\begin{align*}
p\cdot (\lraname(f_{\omega}),\lrvname(f_{\omega})) +(1-p)\cdot  (\lraname(f'_{\omega}),\lrvname(f'_{\omega})) &
   \leq p\cdot \left((\ell\cdot \varepsilon,k\cdot \varepsilon)+(\varepsilon,\varepsilon)\right)+(1-p)\cdot \left((\ell'\cdot \varepsilon,k'\cdot \varepsilon)+(\varepsilon,\varepsilon)\right) \\
  & \leq  (\Ex{\zeta}{s_0}{\lraname|R_C}, \Ex{\zeta}{s_0}{\lrvlname|R_C})+(\varepsilon,\varepsilon)
\end{align*}
This finishes the proof of Proposition~\ref{prop:eps-opt-local}.
\end{proof}

\subsubsection{Details for proof of Proposition~\ref{prop:local-main}}\label{app-local-main}
We have
\[
\Ex{\zeta}{s_0}{\lraname}=\sum_{C\in \Mec(G)} \Prb(R_C)\cdot \Ex{\zeta}{s_0}{\lraname\mid R_C}
\qquad
\text{and}
\qquad
\Ex{\zeta}{s_0}{\lrvlname}=\sum_{C\in \Mec(G)} \Prb(R_C)\cdot \Ex{\zeta}{s_0}{\lrvlname\mid R_C}
\]
Here $\Ex{\zeta}{s_0}{\lraname\mid R_C}$ and $\Ex{\zeta}{s_0}{\lraname\mid R_C}$ are conditional expectations of $\lraname$ and $\lrvlname$, respectively, on runs of $R_C$.
Thus
\begin{equation}\label{eq:ev-complete}
(\Ex{\zeta}{s_0}{\lraname}, \Ex{\zeta}{s_0}{\lrvlname})\quad = \quad \sum_{C\in \Mec(G)} \Prb(R_C)\cdot \left(\Ex{\zeta}{s_0}{\lraname\mid R_C},\Ex{\zeta}{s_0}{\lrvlname\mid R_C}\right)
\end{equation}

We define memoryless strategies $\kappa$ and $\kappa'$ in $C$ as follows: Given $s\in \MECstate{C}$ such that $\sum_{b\in A(s)} f_C(b)>0$ and $a\in A(s)$, we put
\[
\kappa(s)(a)=f_C(a)\ /\ \sum_{b\in A(s)} f_C(b)
\qquad
\text{and}
\qquad
\kappa'(s)(a)=f_C(a)\ /\ \sum_{b\in A(s)} f_C(b)
\]
In the remaining states $s$ the strategy $\kappa$ (or $\kappa'$) behaves as a memoryless deterministic strategy reaching $\{s\in \MECstate{C}\mid \sum_{b\in \act{s}} f_C(b)>0\}$ (or $\{s\in \MECstate{C}\mid \sum_{b\in \act{s}} f'_C(b)>0\}$, resp.) with probability one.

Given a BSCC $D$ of $C^{\kappa}$ (or $D'$ of $C^{\kappa'}$), we write $f_C(D)=\sum_{a\in \BSCCact{D}} f_C(a)$ (or $f'_C(D')=\sum_{a\in \BSCCact{D'}} f'_C(a)$, resp.)

Denoting by $L$ the tuple $(\Ex{\zeta}{s_0}{\lraname|R_C}, \Ex{\zeta}{s_0}{\lrvlname|R_C})$ we obtain
\begin{eqnarray*}
L & = & p_C\cdot (\lraname[f_C],\lrvname[f_C]) + (1-p_C)\cdot
 (\lraname[f'_C],\lrvname[f'_C]) \\
  & = & \sum_{D\in \BSCCuni{C^{\kappa}}} p_C\cdot f_C(D)\cdot \left(\sum_{a\in \BSCCact{D}} \frac{f_C(a)}{f_C(D)}\cdot r(a),\sum_{a\in \BSCCact{D}} \frac{f_C(a)}{f_C(D)}\cdot (r(a)-\lraname[f_C])^2\right) \\
  & & +\ \sum_{D\in \BSCCuni{C^{\kappa'}}} (1-p_C)\cdot f'_C(D)\cdot \left(\sum_{a\in \BSCCact{D}} \frac{f'_C(a)}{f'_C(D)}\cdot r(a),\sum_{a\in \BSCCact{D}} \frac{f'_C(a)}{f'_C(D)}\cdot (r(a)-\lraname[f'_C])^2\right) \\
  & \geq & \sum_{D\in \BSCCuni{C^{\kappa}}} p_C\cdot f_C(D)\cdot \left(\sum_{a\in \BSCCact{D}} \frac{f_C(a)}{f_C(D)}\cdot r(a),\sum_{a\in \BSCCact{D}} \frac{f_C(a)}{f_C(D)}\cdot (r(a)-\sum_{b\in \BSCCact{D}} \frac{f_C(b)}{f_C(D)}\cdot r(b))^2\right) \\
  & & +\ \sum_{D\in \BSCCuni{C^{\kappa'}}} (1-p_C)\cdot f'_C(D)\cdot \left(\sum_{a\in \BSCCact{D}} \frac{f'_C(a)}{f'_C(D)}\cdot r(a),\sum_{a\in \BSCCact{D}} \frac{f'_C(a)}{f'_C(D)}\cdot (r(a)-\sum_{b\in \BSCCact{D}} \frac{f'_C(b)}{f'_C(D)}\cdot r(b))^2\right) \\
  & = & \sum_{D\in \BSCCuni{C^{\kappa}}} p_C\cdot f_C(D)\cdot \left(\Exp_D(\lraname),\Exp_D(\lrvlname)\right)
    + \sum_{D\in \BSCCuni{C^{\kappa'}}} (1-p_C)\cdot f'_C(D)\cdot \left(\Exp_{D}(\lraname),\Exp_{D}(\lrvlname)\right) \\
\end{eqnarray*}
Here $\Exp_D(\lraname)$ and $\Exp_D(\lrvlname)$ denote the expected mean-payoff and the expected local variance, resp., on almost all runs of either $C^{\kappa}$ or $C^{\kappa'}$ initiated in any state of $D$ (note that almost all such runs have the same mean-payoff and the local variance due to ergodic theorem).
Note that the second equality follows from the fact that $f_C(a)>0$ (or $f'_C(a)>0$) iff $a\in \BSCCact{D}$ for a BSCC $D$ of $C^{\kappa}$ (or of $C^{\kappa'}$). The third inequality follows from Lemma~\ref{lem:min-var}. The last equality follows from Lemma~\ref{lem:freq-var} and the fact that $f_C(a)/f_C(D)$ is the frequency of firing $a$ on almost all runs initiated in $D$.

By~Lemma~\ref{lem:approx-two}, there are two components $D,D'\in \BSCCuni{C^{\kappa}}\cup \BSCCuni{C^{\kappa'}}$ and $0\leq d_C\leq 1 $ such that
\[
L  \quad \geq \quad d_C \cdot \left(\Exp_D(\lraname),\Exp_D(\lrvlname)\right) 
 + (1-d_C)\cdot \left(\Exp_{D'}(\lraname),\Exp_{D'}(\lrvlname)\right) 
\]
In what follows we use the following definition: Let $\nu$ be a memoryless randomized strategy on a MEC $C$ and let $K$ be a BSCC of $C^{\nu}$. We say that a strategy $\mu_K$ is {\em induced} by $K$ if 
\begin{enumerate}
\item $\mu_K(s)(a)=\nu(s)(a)$ for all $s\in \BSCCstate{K}$ and $a\in \BSCCact{K}$
\item in all $s\in S\smallsetminus (\BSCCstate{K})$ the strategy $\mu_K$ corresponds to a memoryless deterministic strategy which reaches a state of $K$ with probability one
\end{enumerate}
(Note that the above definition is independent of the strategy $\nu$ once it generates the same BSCC $K$.)

The strategies $\mu_D$ and $\mu_{D'}$ induced by $D$ and $D'$, resp., generate single-BSCC Markov chains $C^{\mu_D}$ and $C^{\mu_{D'}}$ satisfying for every state $s\in C\cap S$ the following
\begin{eqnarray*}
L &= & (\Ex{\zeta}{s_0}{\lraname|R_C}, \Ex{\zeta}{s_0}{\lrvlname|R_C}) \\
  &\geq & d_C \cdot \left(\Exp_D(\lraname),\Exp_D(\lrvlname)\right) 
 + (1-d_C)\cdot \left(\Exp_{D'}(\lraname),\Exp_{D'}(\lrvlname)\right) \\
  & = & d_C\cdot (\Ex{\mu_D}{s}{\lraname}, \Ex{\mu_D}{s}{\lrvlname}) + (1-d_C)\cdot (\Ex{\mu_{D'}}{s}{\lraname}, \Ex{\mu_{D'}}{s}{\lrvlname}) \\
  & = &  d_C\cdot (\Ex{\mu_D}{s}{\lraname}, \Ex{\mu_D}{s}{\lrvhname}) + (1-d_C)\cdot (\Ex{\mu_{D'}}{s}{\lraname}, \Ex{\mu_{D'}}{s}{\lrvhname}) 
\end{eqnarray*}
Here the last equality follows from the fact that almost all runs in $C^{\mu_{D}}$ (and also in $C^{\mu_{D'}}$) have the same mean-payoff. Thus for almost all runs the local variance is equal to the hybrid one. This shows that in $C$, a convex combination of two memoryless (possibly randomized) strategies is sufficient to optimize the mean-payoff and the local variance. 

Now we show that these strategies may be even deterministic.
\begin{claim}\label{prop:MR-MD}
Let $s\in S$.
There are {\em memoryless deterministic} strategies $\chi_1,\chi_2,\chi'_1,\chi'_2$ in $C$, each generating a single BSCC, and numbers $0\leq \nu,\nu'\leq 1$ such that
\[
(\Ex{\mu_D}{s}{\lraname}, \Ex{\mu_D}{s}{\lrvhname}) \ge  \nu\cdot (\Ex{\chi_1}{s}{\lraname}, \Ex{\chi_1}{s}{\lrvhname})+
(1-\nu)\cdot (\Ex{\chi_2}{s}{\lraname}, \Ex{\chi_2}{s}{\lrvhname})
 \geq  \nu\cdot (\Ex{\chi_1}{s}{\lraname}, \Ex{\chi_1}{s}{\lrvlname})+
(1-\nu)\cdot (\Ex{\chi_2}{s}{\lraname}, \Ex{\chi_2}{s}{\lrvlname})
\]
and
\[
(\Ex{\mu_{D'}}{s}{\lraname}, \Ex{\mu_{D'}}{s}{\lrvhname}) \ge \nu'\cdot (\Ex{\chi'_1}{s}{\lraname}, \Ex{\chi'_1}{s}{\lrvhname})+
(1-\nu')\cdot (\Ex{\chi'_2}{s}{\lraname}, \Ex{\chi'_2}{s}{\lrvhname})
 \geq \nu'\cdot (\Ex{\chi'_1}{s}{\lraname}, \Ex{\chi'_1}{s}{\lrvlname})+
(1-\nu')\cdot (\Ex{\chi'_2}{s}{\lraname}, \Ex{\chi'_2}{s}{\lrvlname})
\]
\end{claim}
\begin{proof}
It suffices to concentrate on $\mu_D$. 
By~\cite{derman1970finite}, $\Ex{\mu_D}{s_0}{\lraname^{I_a}}$ is equal to a convex combination of the values $\Ex{\iota_i}{s_0}{\lraname^{I_a}}$ for some memoryless deterministic strategies $\iota_1,\ldots,\iota_m$, i.e. there are $\gamma_1,\ldots,\gamma_m > 0$ such that $\sum_{i=1}^m \gamma_i=1$ and $\sum_{i=1}^m \gamma_i \cdot \Ex{\iota_i}{s_0}{\lraname^{I_a}} = \Ex{\mu_D}{s_0}{\lraname^{I_a}}$.
For all $1\le i \le m$ and $D\in\BSCCuni{C^{\iota_i}}$ denote
$\iota_{i,D}$ a memoryless deterministic strategy such that $\iota_{i,D}(s)=\iota_{i}(s)$ on all $s\in D\cap S$, and on other states $\iota_{i,D}$ is defined so that
$D\cap S$ is reached with probability 1, independent of the starting state.
For all $a\in D\cap A$ we have $\Ex{\iota_{i,D}}{s_0}{\lraname^{I_a}} = \Pr{\iota_i}{s_0}{\reach(D)}\cdot \Ex{\mu_D}{s_0}{\lraname^{I_a}}$,
while for $a\not\in D\cap A$ we have $\Ex{\iota_{i,D}}{s_0}{\lraname^{I_a}} = 0$. Hence
$\sum_{i=1}^m\sum_{D\in\BSCCuni{C^{\iota_i}}} \gamma_i \cdot\Pr{\iota_i}{s_0}{\reach(D)} \cdot \Ex{\iota_{i,D}}{s_0}{\lraname^{I_a}} = \Ex{\iota_i}{s_0}{\lraname^{I_a}}$.
Since $\sum_{i=1}^m\sum_{D\in\BSCCuni{C^{\iota_i}}} \gamma_i \cdot\Pr{\iota_i}{s_0}{\reach(D)}=1$, we apply Lemma~\ref{lem:approx-two} and get
there are two memoryless deterministic single-BSCC strategies $\chi_1,\chi_2$ and $0\leq \nu\leq 1$ such that
\[
\Ex{\mu_D}{s_0}{\lraname^{I_a}}=\nu \Ex{\chi_1}{s_0}{\lraname^{I_a}}+ (1-\nu) \Ex{\chi_2}{s_0}{\lraname^{I_a}}
\]
which together with Lemma~\ref{lem:freq-var} implies that
\begin{eqnarray*}
\Ex{\mu_D}{s}{\lraname} & = & \sum_{a\in A} r(a)\cdot \Ex{\mu_D}{s}{\lraname^{I_a}} \\
  & = & \sum_{a\in A} r(a)\cdot \left(\nu \Ex{\chi_1}{s}{\lraname^{I_a}}+(1-\nu) \Ex{\chi_2}{s}{\lraname^{I_a}}\right) \\
  & = & \nu \sum_{a\in A} r(a)\cdot \Ex{\chi_1}{s}{\lraname^{I_a}} + (1-\nu)\sum_{a\in A} r(a)\cdot \Ex{\chi_2}{s}{\lraname^{I_a}} \\
  & = & \nu \Ex{\chi_1}{s}{\lraname}+(1-\nu) \Ex{\chi_2}{s}{\lraname}
\end{eqnarray*}
and
\begin{eqnarray*}
\Ex{\mu_D}{s}{\lrvhname} & = & \sum_{a\in A} (r(a)-\Ex{\mu_D}{s}{\lraname})^2\cdot \Ex{\mu_D}{s}{\lraname^{I_a}} \\
  & = & \sum_{a\in A} (r(a)-\Ex{\mu_D}{s}{\lraname})^2\cdot (\nu\Ex{\chi_1}{s}{\lraname^{I_a}}+(1-\nu)\Ex{\chi_2}{s}{\lraname^{I_a}}) \\
  & = & \nu \sum_{a\in A} (r(a)-\Ex{\mu_D}{s}{\lraname})^2\cdot \Ex{\chi_1}{s}{\lraname^{I_a}} +
  (1-\nu) \sum_{a\in A} (r(a)-\Ex{\mu_D}{s}{\lraname})^2\cdot \Ex{\chi_2}{s}{\lraname^{I_a}} \\
  & \geq & \nu \sum_{a\in A} (r(a)-\Ex{\chi_1}{s}{\lraname})^2\cdot \Ex{\chi_1}{s}{\lraname^{I_a}} +
  (1-\nu) \sum_{a\in A} (r(a)-\Ex{\chi_2}{s}{\lraname})^2\cdot \Ex{\chi_2}{s}{\lraname^{I_a}} \\
  & = & \nu \Ex{\chi_1}{s}{\lrvhname} +
  (1-\nu) \Ex{\chi_2}{s}{\lrvhname}
\end{eqnarray*}
Here the inequality follows from Lemma~\ref{lem:min-var}. So
\[
(\Ex{\mu_D}{s}{\lraname},\Ex{\mu_D}{s}{\lrvhname})\geq \nu (\Ex{\chi_1}{s}{\lraname},\Ex{\chi_1}{s}{\lrvhname}) + (1-\nu) (\Ex{\chi_2}{s}{\lraname},\Ex{\chi_2}{s}{\lrvhname})
\]
Finally, we show that $\Ex{\chi_1}{s}{\lrvhname}\geq \Ex{\chi_1}{s}{\lrvlname}$. Since $\chi_1$ has a single BSCC, almost all runs have the same mean payoff. Hence, $\Ex{\chi_1}{s}{\lrvhname}= \Ex{\chi_1}{s}{\lrvlname}$.
\end{proof}
\noindent
By Claim~\ref{prop:MR-MD},
 \begin{eqnarray*}
 L & \geq & d_C\cdot (\Ex{\mu_D}{s}{\lraname}, \Ex{\mu_D}{s}{\lrvhname}) +(1-d_C)\cdot
  (\Ex{\mu_{D'}}{s}{\lraname}, \Ex{\mu_{D'}}{s}{\lrvhname}) \\
   & \geq & d_C\cdot \nu\cdot (\Ex{\chi_1}{s}{\lraname}, \Ex{\chi_1}{s}{\lrvlname})
    +  d_C\cdot (1-\nu) \cdot (\Ex{\chi_2}{s}{\lraname}, \Ex{\chi_2}{s}{\lrvlname}) \\
   & & +\  (1-d_C) \cdot \nu' \cdot (\Ex{\chi'_1}{s}{\lraname}, \Ex{\chi'_1}{s}{\lrvlname})
    + (1-d_C) \cdot (1-\nu') \cdot (\Ex{\chi'_2}{s}{\lraname}, \Ex{\chi'_2}{s}{\lrvlname})
 \end{eqnarray*}
and so by Lemma~\ref{lem:approx-two}, there are $\pi_C, \pi'_C \in \{\chi_1,\chi_2,\chi'_1,\chi'_2\}$ and a number $h_C$ such that
\begin{eqnarray*}
L &= & (\Ex{\zeta}{s_0}{\lraname|R_C}, \Ex{\zeta}{s_0}{\lrvlname|R_C}) \\
  & \ge & h_C\cdot (\Ex{\pi_C}{s}{\lraname}, \Ex{\pi_C}{s}{\lrvlname}) + (1-h_C)\cdot (\Ex{\pi'_C}{s}{\lraname}, \Ex{\pi'_C}{s}{\lrvlname})
\end{eqnarray*}
Define memoryless deterministic strategies $\pi$ and $\pi'$ in $G$ so that for every $s\in S$ and $a\in A$ we have
$\pi(s)(a):=\pi_C(s)(a)$ and $\pi'(s)(a):=\pi'_C(s)(a)$ for $s\in \MECstate{C}$.

\subsubsection{Proof of Equation~(\ref{eqn-t2})}\label{app-eqn-t2}
We have
\begin{eqnarray*}
\lefteqn{(\Ex{\zeta}{s_0}{\lraname},\Ex{\zeta}{s_0}{\lrvlname}) }\\& = & \!\!\Big(\!\!\!\!\sum_{C\in \Mec(G)}\!\!\!\! \Pr{\zeta}{s_0}{R_C}\cdot \Ex{\zeta}{s_0}{\lraname\mid R_C},
   \!\!\!\!\sum_{C\in \Mec(G)}\!\!\!\! \Pr{\zeta}{s_0}{R_C}\cdot \Ex{\zeta}{s_0}{\lrvlname\mid R_C}\Big) \\
  & \geq  & \!\!\Big(\!\!\!\!\sum_{C\in \Mec(G)}\!\!\!\! \Pr{\sigma}{s_0}{R_C}{\cdot} h_C {\cdot} \Exp^{\pi}_{s[C]}[\lraname] + 
 \Pr{\sigma}{s_0}{R_C}{\cdot} (1{-}h_C) {\cdot} \Exp^{\pi'}_{s[C]}[\lraname], \\
  & &\; \sum_{C\in \Mec(G)}\!\!\!\! \Pr{\sigma}{s_0}{R_C}{\cdot} h_C {\cdot} \Exp^{\pi}_{s[C]}[\lrvlname] + 
 \Pr{\sigma}{s_0}{R_C}{\cdot} (1{-}h_C) {\cdot} \Exp^{\pi'}_{s[C]}[\lrvlname]\Big)\\
  & = &\!\! (\Ex{\sigma}{s_0}{\lraname},\Ex{\sigma}{s_0}{\lrvlname})
\end{eqnarray*}
Here $s[C]$ is an arbitrary state of $\MECstate{C}$.
\subsubsection{Proof of Theorem~\ref{thm:local-np-alg}}\label{app-local-np-alg}
First, we show that if there is $\zeta$ in $G$ such that
$(\Ex{\zeta}{s_0}{\lraname}, \Ex{\zeta}{s_0}{\lrvlname})\leq (u,v)$, then there is a strategy $\rho$ in $G[\pi,\pi']$ such that $(\Ex{\rho}{s_{in})}{\lraname^{r_1}},\Ex{\rho}{s_{in}}{\lraname^{r_2}})\leq (u,v)$. 
Consider the 3-memory stochastic update strategy $\sigma$ from Proposition~\ref{prop:local-main} satisfying
$(\Ex{\sigma}{s_0}{\lraname}, \Ex{\sigma}{s_0}{\lrvlname})\leq (u,v)$. Define a memoryless strategy $\rho$ in $G[\pi,\pi']$ that mimics $\sigma$ as follows (we denote the only memory element of $\rho$ by $\bullet$):
\begin{itemize}
\item $\rho(s_{in},\bullet)(\mathit{default}) = \alpha(m_1)$,  $\rho(s_{in},\bullet)([\pi]) = \alpha(m_2)$, $\rho(s_{in},\bullet)([\pi']) = \alpha(m'_2)$,
\item $\rho((s,m_1),\bullet)(a)=\sigma_n(s,m_1)(a)\cdot \sigma_u(a,s,m_1)(m_1)$  for all $a\in A$
\item $\rho((s,m_1),\bullet)(\pi)=\sigma_u(a,s,m_1)(m_2)$
\item $\rho((s,m_1),\bullet)(\pi')=\sigma_u(a,s,m_1)(m'_2)$
\item $\rho((s,m_2),\bullet)(\mathit{default})=\rho((s,m'_2),\bullet)(\mathit{default})=1$
\end{itemize}
It is straightforward to verify that 
\[
(\Ex{\sigma}{s_0}{\lraname}, \Ex{\sigma}{s_0}{\lrvlname})\quad = \quad (\Ex{\rho}{s_{in}}{\lraname^{r_1}}, \Ex{\rho}{s_{in}}{\lraname^{r_2}})\quad \leq \quad (u,v)
\]

\noindent
Second, we show that if there is $\rho'$ in $G[\pi,\pi']$ satisfying 
$(\Ex{\rho'}{s_{in}}{\lraname^{r_1}}, \Ex{\rho'}{s_{in}}{\lraname^{r_2}}) \leq (u,v)$, then
there is the desired 3-memory stochastic update strategy $\sigma$ in $G$. Moreover, we show that existence of such $\sigma$ is decidable in polynomial time and also that the strategy is computable in polynomial time (if it exists).

By~\cite{BBCFK:MDP-two-views}, there is a 2-memory stochastic update strategy $\sigma'$ for $G[\pi,\pi']$ such that 
\[(\Ex{\sigma'}{s_{in}}{\lraname^{r_1}}, \Ex{\sigma'}{s_{in}}{\lraname^{r_2}})\leq (u,v)\]
Moreover, existence of such $\sigma'$ is decidable in polynomial time and also $\sigma'$ is computable in polynomial time (if it exists). We show how to transform, in polynomial time, the strategy $\sigma'$ to the desired $\sigma$.

In~\cite{BBCFK:MDP-two-views}, the strategy $\sigma'$ is constructed using a memoryless deterministic strategy $\xi$ on $G[\pi,\pi']$ as follows: The strategy $\sigma'$ has two memory elements, say $n_1,n_2$. In $n_1$ the strategy $\sigma'$ behaves as a memoryless randomized strategy. After updating (stochastically) its memory element to $n_2$, which may happen {\em only} in a BSCC of $G[\pi,\pi']^{\xi}$, the strategy $\sigma'$ behaves as $\xi$ and no longer updates its memory. 
Note that if $\sigma'$ changes its memory element while still being in states of the form $(s,m_1)$ then from this moment on the second component is always $m_1$. However, such a strategy may be improved by moving to $(s,m_2)$ (or to $(s,m'_2)$) when its memory changes to $n_2$ because the values of $\vec{r}$ in states of the form $(s,m_1)$ are so large that moving to any state with $m_2$ or $m'_2$ in the second component is better than staying in them. Obviously, there are only polynomially many improvements of this kind and all of them can be done in polynomial time.

So we may safely assume that the strategy $\sigma'$ stays in $n_1$ on states of $\{(s,m_1)\mid s\in S\}$, i.e. behaves as a memoryless randomized strategy on these states. We define the 3-memory stochastic update strategy $\sigma$ on $G$ with memory elements $m_1,m_2,m'_2$  which in the memory element $m_1$ mimics the behavior of $\sigma'$ on states of the form $(s,m_1)$. Once $\sigma'$ chooses the action $[\pi]$ (or $[\pi']$) the strategy $\sigma$ changes its memory element to $m_2$ (or to $m'_2$) and starts playing according to $\pi$ (or to $\pi'$, resp.)
 
 Formally, we define
\begin{itemize}
\item $\alpha(m_1)=\sigma'_n(s_{in},n_1)(\mathit{default})$, $\alpha(m_1)=\sigma'_n(s_{in},n_1)([\pi])$ and $\alpha(m_1)=\sigma'_n(s_{in},n_1)([\pi'])$
\item $\sigma_n(s,m_1)(a)=\sigma'_n((s,m_1),n_1)(a)\  / \ \sum_{b\in A} \sigma'_n((s,m_1),n_1)(b)$  for all $a\in A$
\item $\sigma_u(a,s,m_1)(m_1)=\sum_{b\in A} \sigma'_n((s,m_1),n_1)(b)$
\item $\sigma_u(a,s,m_1)(m_2)=\sigma'_n(a,(s,m_1),n_1)([\pi])$
\item $\sigma_u(a,s,m_1)(m'_2)=\sigma'_n(a,(s,m_1),n_1)([\pi'])$
\end{itemize}
It is straightforward to verify that 
\[
(\Ex{\sigma}{s_0}{\lraname}, \Ex{\sigma}{s_0}{\lrvlname})\quad = \quad (\Ex{\sigma'}{s_{in}}{\lraname}, \Ex{\sigma'}{s_{in}}{\lrvlname})\quad \leq \quad (u,v)
\]

\subsection{Proofs for Hybrid Variance}
\label{app-hybrid}
\subsubsection{Proof of Proposition \ref{prop:relation}}
\label{app:relation}
We have
\begin{eqnarray*}
\Ex{\sigma}{s}{\lrvname}
&=&\Ex{\sigma}{s}{\lim_{n\rightarrow \infty} \frac{1}{n}\sum_{i=0}^{n-1} \big( r(A_i) - \lraname \big)^2}\\
&=&\Ex{\sigma}{s}{\lim_{n\rightarrow \infty} \frac{1}{n}\sum_{i=0}^{n-1} \big(r(A_i)^2 - 2\cdot r(A_i) \cdot\lraname^2 + \lraname^2 \big)}\\
&=&\Ex{\sigma}{s}{\lim_{n\rightarrow \infty} \frac{1}{n}\sum_{i=0}^{n-1} r(A_i)^2} -
   \Ex{\sigma}{s}{\lim_{n\rightarrow \infty} \frac{1}{n}\sum_{i=0}^{n-1} 2\cdot r(A_i) \cdot\lraname} +
   \Ex{\sigma}{s}{\lim_{n\rightarrow \infty} \frac{1}{n}\sum_{i=0}^{n-1} \lraname^2}\\
&=& \Ex{\sigma}{s}{\lim_{n\rightarrow \infty} \frac{1}{n}\sum_{i=0}^{n-1} r(A_i)^2} -
   2\cdot \Ex{\sigma}{s}{\lim_{n\rightarrow \infty} \lraname \cdot \frac{1}{n}\sum_{i=0}^{n-1} r(A_i)}\cdot
   \Ex{\sigma}{s}{\lim_{n\rightarrow \infty} \frac{1}{n}\sum_{i=0}^{n-1} \lraname^2}\\
&=& \Ex{\sigma}{s}{\lim_{n\rightarrow \infty} \frac{1}{n}\sum_{i=0}^{n-1} r(A_i)^2} -
   2\cdot \Ex{\sigma}{s}{ \lraname^2}+
   \Ex{\sigma}{s}{ \lraname^2}\\
&=& \Ex{\sigma}{s}{\lim_{n\rightarrow \infty} \frac{1}{n}\sum_{i=0}^{n-1} r(A_i)^2} -
   \Ex{\sigma}{s}{ \lraname^2}
\end{eqnarray*}
and
\begin{eqnarray*}
\Ex{\sigma}{s}{\lrvhname}
&=&\Ex{\sigma}{s}{\lim_{n\rightarrow \infty}  \frac{1}{n}\sum_{i=0}^{n-1} \left( r(A_i(\pat)) - \Ex{\sigma}{s}{\lraname} \right)^2}\\
&=& \Ex{\sigma}{s}{\lim_{n\rightarrow \infty}  \frac{1}{n}\sum_{i=0}^{n-1} r(A_i)^2}
  - \Ex{\sigma}{s}{\lim_{n\rightarrow \infty}  \frac{1}{n}\sum_{i=0}^{n-1} 2\cdot r(A_i) \cdot \Ex{\sigma}{s}{\lraname}}
   + \Ex{\sigma}{s}{\lim_{n\rightarrow \infty}  \frac{1}{n}\sum_{i=0}^{n-1} \Ex{\sigma}{s}{\lraname}^2}\\
&=&\Ex{\sigma}{s}{\lim_{n\rightarrow \infty}  \frac{1}{n}\sum_{i=0}^{n-1} r(A_i)^2}
   - 2\cdot \Ex{\sigma}{s}{\lraname}^2
   + \Ex{\sigma}{s}{\lraname}^2\\
&=&\Ex{\sigma}{s}{\lim_{n\rightarrow \infty}  \frac{1}{n}\sum_{i=0}^{n-1} r(A_i)^2}
   - \Ex{\sigma}{s}{\lraname}^2\\
\end{eqnarray*}
and so
\begin{eqnarray*}
\Va{\sigma}{s}{\lraname} + \Ex{\sigma}{s}{\lrvname}
&=& \Ex{\sigma}{s}{ \lraname^2} - \Ex{\sigma}{s}{\lraname}^2
    + \Ex{\sigma}{s}{\lim_{n\rightarrow \infty} \frac{1}{n}\sum_{i=0}^{n-1} r(A_i)^2} -
   \Ex{\sigma}{s}{ \lraname^2}\\
&=& \Ex{\sigma}{s}{\lim_{n\rightarrow \infty} \frac{1}{n}\sum_{i=0}^{n-1} r(A_i)^2} - \Ex{\sigma}{s}{\lraname}^2 = \Ex{\sigma}{s}{\lrvhname}
\end{eqnarray*}

\subsubsection{Obtaining 3-memory strategy $\sigma$.}
\label{app-hybrid-finite}

Let us fix a MDP $G=(S,A,\mathit{Act},\delta)$. We prove the following proposition.

\begin{proposition}\label{prop:hybrid-finite-strategy}
Let $s_0\in S$ and $u,v\in \Rset$. 
If there is a strategy $\zeta$ satisfying
\[
(\Ex{\zeta}{s_0}{\lraname}, \Ex{\zeta}{s_0}{\lrvhname})\quad \leq\quad (u,v);
\]
then there exists a 3-memory strategy $\sigma$ satisfying
\[
(\Ex{\sigma}{s_0}{\lraname}, \Ex{\sigma}{s_0}{\lrvhname})\quad \leq \quad (u,v).
\]
\end{proposition}

Intuitively the proof will resemble the proof of Proposition~\ref{prop:local-main}, and given an arbitrary 
strategy $\zeta$ with $\Ex{\zeta}{s_0}{\lraname}=u$, we will mimic the proof for the local 
variance replacing the quantity $(r(A_j(\pat))-\lra{\pat})^2$ by $(r(A_j(\pat)-u)^2$ 
appropriately. Formally, Proposition~\ref{prop:hybrid-finite-strategy} is a consequence of Lemma~\ref{lemm:hybrid-finite}.

\begin{figure}\small
\begin{align}
\mathbf{1}_{s_0}(s) + \sum_{a\in A} y_{a}\cdot \delta(a)(s) & = 
 \sum_{a\in \mathit{Act}(s)} y_{a} + y_s \hspace*{1em} \text{for all  $s\in S$}
\label{eq:yahf}\\
\sum_{s\in S}y_{s} & =  1 &
\label{eq:ys1hf}\\
 \sum_{s\in C} y_{s} & =  \sum_{a\in A\cap C} x_{a} + \sum_{a\in A\cap C} x'_{a} \hspace*{1.1em}
   \text{for all $C\in \Mec(G)$}
\label{eq:yChf}\\
\sum_{a\in A} x_{a}\cdot \delta(a)(s) & = 
\sum_{a\in \mathit{Act}(s)} x_{a} \hspace*{1em} \text{for all  $s\in S$}
\label{eq:xahf}\\
\sum_{a\in A} x'_{a}\cdot \delta(a)(s) & = 
\sum_{a\in \mathit{Act}(s)} x'_{a} \hspace*{1em} \text{for all  $s\in S$}
\label{eq:xaprimehf}\\
u & = \sum_{C \in \Mec(G)} \left(\sum_{a\in A\cap C} x_{a}\cdot r(a) + \sum_{a\in A\cap C} x'_{a}\cdot r(a)\right)
\label{eq:av-rewhf}\\
v & = \sum_{C \in \Mec(G)}\left(\sum_{a\in A\cap C} x_{a}\cdot (r(a)-u)^2 + 
\sum_{a\in A\cap C} x'_{a}\cdot (r(a)-u)^2\right) \\
x_a & \geq 0 \hspace*{1em} \text{for all  $a\in A$} \\
x'_a & \geq 0 \hspace*{1em} \text{for all  $a\in A$} 
\label{eq:var-rewhf}
\end{align}
\caption{System $L_H^\zeta$ of linear inequalities. Here $u$ and $v$ are treated as constants (see~Lemma~\ref{lemm:hybrid-finite}). 
We define $\mathbf{1}_{s_0}(s)=1$ if $s = s_0$, and $\mathbf{1}_{s_0}(s)=0$ otherwise.}
\label{system-Lzeta}
\end{figure}

\begin{lemma}\label{lemm:hybrid-finite}
Let us fix $s_0\in S$ and $u,v\in \Rset$.
\begin{enumerate}
\item Consider an arbitrary strategy $\zeta$ such that 
$(\Ex{\zeta}{s_0}{\lraname}, \Ex{\zeta}{s_0}{\lrvhname}) = (u,v)$.
Then the system $L_H^\zeta$ (Figure~\ref{system-Lzeta}) has a non-negative solution.
\item If there is a non-negative solution for the system $L_H^\zeta$ (Figure~\ref{system-Lzeta}), 
then there is a 3-memory stochastic-update strategy $\sigma$ satisfying 
$(\Ex{\sigma}{s_0}{\lraname}, \Ex{\sigma}{s_0}{\lrvhname}) = (u,v)$.
\end{enumerate}
\end{lemma}

We start with the proof of the first item of Lemma~\ref{lemm:hybrid-finite}.
We have
\[
\Ex{\zeta}{s_0}{\lraname}=\sum_{C\in \Mec(G)} \Prb(R_C)\cdot \Ex{\zeta}{s_0}{\lraname\mid R_C}
\qquad\text{and}
\qquad
\Ex{\zeta}{s_0}{\lrvhname}=\sum_{C\in \Mec(G)} \Prb(R_C)\cdot \Ex{\zeta}{s_0}{\lrvhname\mid R_C}
\]
and thus
\begin{equation}
\left(\sum_{C\in \Mec(G)} \Prb(R_C)\cdot \Ex{\zeta}{s_0}{\lraname\mid R_C},\sum_{C\in \Mec(G)} \Prb(R_C)\cdot \Ex{\zeta}{s_0}{\lrvhname\mid R_C}\right) = (u,v)\,.
\end{equation}
Let $C$ be a MEC and consider a frequency function $f$ on $C$.
Given $u$ and $f$, define $\lraname[f]:=\sum_{a\in C} f(a)\cdot r(a)$ 
and $\lrvhname[f,u] :=\sum_{a\in C} f(a)\cdot (r(a)-u)^2$.
\begin{proposition}\label{prop:strong-opt-local-h}
Let us fix a MEC $C$.
There are two frequency functions $f_C:C\rightarrow \Rset$ and $f'_C:C\rightarrow \Rset$ on $C$, 
and a number $p_C\in [0,1]$ such that the following holds
\[
p_C\cdot (\lraname[f_C],\lrvhname[f_C,u])+(1-p_C)\cdot
 (\lraname[f'_C],\lrvhname[f'_C,u])=(\Ex{\zeta}{s_0}{\lraname|R_C}, \Ex{\zeta}{s_0}{\lrvhname|R_C})
\]
\end{proposition}
\noindent
We first argue that Proposition~\ref{prop:strong-opt-local-h} gives us a solution of $L_H^{\zeta}$. 
Indeed, given $a\in A$ (or $s\in S$) denote by $C(a)$ (or $C(s)$) the MEC containing $a$ (or $s$).
For every $a\in A$ put 
\[
x_a=\Prb(R_{C(a)})\cdot p_{C(a)}\cdot f_{C(a)}(a)
\qquad
\text{and} 
\qquad
x'_a=\Prb(R_{C(a)})\cdot (1-p_{C(a)})\cdot f'_{C(a)}(a)
\] 
For every action $a\in A$ which does not belong to any MEC put $x_a=x'_a=0$. 
(1)~We have the following equality for $u$, i.e.,
\begin{eqnarray*}
u & = & \sum_{C\in \Mec(G)} \Prb(R_C)\cdot \Ex{\zeta}{s_0}{\lraname\mid R_C} \\
  & = & \sum_{C\in \Mec(G)} \Prb(R_C)\cdot (p_C\cdot \lraname[f_C]+(1-p_C)\cdot \lraname[f'_C]) \\
  & = & \sum_{C\in \Mec(G)} \Prb(R_C)\cdot p_C\cdot \lraname[f_C]+\sum_{C\in \Mec(G)} \Prb(R_C)\cdot(1-p_C)\cdot \lraname[f'_C]) \\
  & = & \sum_{C\in \Mec(G)} \Prb(R_C)\cdot p_C\cdot \sum_{a\in C} f_C(a)\cdot r(a)
   +\sum_{C\in \Mec(G)} \Prb(R_C)\cdot(1-p_C)\cdot \sum_{a\in C} f'_C(a)\cdot r(a) \\
  & = & \sum_{C\in \Mec(G)} \sum_{a\in C} \Prb(R_C)\cdot p_C\cdot f_C(a)\cdot r(a)
   +\sum_{C\in \Mec(G)} \sum_{a\in C} \Prb(R_C)\cdot(1-p_C)\cdot f'_C(a)\cdot r(a) \\
  & = & \sum_{C\in \Mec(G)} \left(\sum_{a\in C} x_a\cdot r(a)+
  \sum_{a\in C} x'_a\cdot r(a)\right) \\
\end{eqnarray*}
and
(2)~the following equality for $v$:
\begin{eqnarray*}
v & = & \sum_{C\in \Mec(G)} \Prb(R_C)\cdot \Ex{\zeta}{s_0}{\lrvhname\mid R_C} \\
  & = & \sum_{C\in \Mec(G)} \Prb(R_C)\cdot (p_C\cdot \lrvhname[f_C,u]+(1-p_C)\cdot \lrvhname[f'_C,u]) \\
  & = & \sum_{C\in \Mec(G)} \Prb(R_C)\cdot p_C\cdot \lrvhname[f_C,u]+\sum_{C\in \Mec(G)} \Prb(R_C)\cdot(1-p_C)\cdot \lrvhname[f'_C,u]) \\
  & = & \sum_{C\in \Mec(G)} \Prb(R_C)\cdot p_C\cdot \sum_{a\in C} f_C(a)\cdot (r(a)-u)^2
    +\sum_{C\in \Mec(G)} \Prb(R_C)\cdot(1-p_C)\cdot \sum_{a\in C} f'_C(a)\cdot (r(a)-u)^2) \\
  & = & \sum_{C\in \Mec(G)} \sum_{a\in C} \Prb(R_C)\cdot p_C\cdot f_C(a)\cdot (r(a)-u)^2
    +\sum_{C\in \Mec(G)} \sum_{a\in C} \Prb(R_C)\cdot(1-p_C)\cdot f'_C(a)\cdot (r(a)-u)^2 \\
  & = & \sum_{C\in \Mec(G)} \left(\sum_{a\in C} x_a\cdot (r(a)-u)^2 + \sum_{a\in C} x'_a\cdot (r(a)-u)^2\right)
\end{eqnarray*}
The appropriate values for $y_a,y_s$ can be found in the same way as in the 
proof of~\cite[Proposition~2]{BBCFK:MDP-two-views}.

It remains to prove Proposition~\ref{prop:strong-opt-local-h}.
As for the proof for local variance, we obtain the proposition from the following slightly weaker 
version
\begin{proposition}\label{prop:eps-opt-hybrid}
Let us fix a MEC $C$ and let $\varepsilon>0$. There are two frequency functions 
$f_{\varepsilon}:C\rightarrow [0,1]$ and $f'_{\varepsilon}:C\rightarrow [0,1]$, and a number $p_{\varepsilon}\in [0,1]$ such that:
\[
p_{\varepsilon}\cdot (\lraname[f_{\varepsilon}],\lrvhname[f_\varepsilon,u])+(1-p_{\varepsilon})\cdot
 (\lraname[f'_{\varepsilon}],\lrvhname[f'_\varepsilon,u])\quad \leq \quad
(\Ex{\zeta}{s_0}{\lraname}, \Ex{\zeta}{s_0}{\lrvhname})+(\varepsilon,\varepsilon)
\]
\end{proposition}
\noindent
As before Proposition~\ref{prop:eps-opt-hybrid} implies Proposition~\ref{prop:strong-opt-local-h} as follows:
There is a sequence $\varepsilon_1,\varepsilon_2,\ldots$, two functions $f_C$ and $f'_C$, and $p_C\in [0,1]$ such that as $n\rightarrow \infty$ 
\begin{itemize}
\item $\varepsilon_n\rightarrow 0$ 
\item $f_{\varepsilon_n}$ converges pointwise to $f_C$
\item $f'_{\varepsilon_n}$ converges pointwise to $f'_C$
\item $p_{\varepsilon_n}$ converges to $p_C$
\end{itemize}
It is easy to show that $f_C$ as well as $f'_C$ are frequency functions. Moreover, 
as 
\[
\lim_{n\rightarrow \infty} (\Ex{\zeta}{s_0}{\lraname}, \Ex{\zeta}{s_0}{\lrvhname})+(\varepsilon_n,\varepsilon_n)=
(\Ex{\zeta}{s_0}{\lraname}, \Ex{\zeta}{s_0}{\lrvhname})
\]
and
\[
\lim_{n\rightarrow \infty} p_{\varepsilon_n}\cdot (\lraname[f_{\varepsilon_n}],\lrvhname[f_{\varepsilon_n},u])+(1-p_{\varepsilon_n})\cdot
 (\lraname[f'_{\varepsilon_n}],\lrvhname[f'_{\varepsilon_n},u])
  =  p_C\cdot (\lraname[f_C],\lrvhname[f_C,u])+(1-p_C)\cdot
 (\lraname[f'_C],\lrvhname[f'_C,u])
\]
we obtain
\[
p_C\cdot (\lraname[f_C],\lrvhname[f_C,u])+(1-p_C)\cdot
 (\lraname[f'_C],\lrvhname[f'_C,u])=(\Ex{\zeta}{s_0}{\lraname}, \Ex{\zeta}{s_0}{\lrvhname})
\]

\paragraph{Proof of Proposition~\ref{prop:eps-opt-hybrid}.}
The proof is exactly the same as proof of Proposition~\ref{prop:eps-opt-local}.
Given $\ell,k\in \Zset$ we denote by $A_H^{\ell,k}$ the set of all runs $\omega\in R_C$ such that
\[
(\ell\cdot \varepsilon,k\cdot \varepsilon) \quad \leq \quad (\lraname(\omega),\lrvhname(\omega))\quad < \quad (\ell\cdot \varepsilon,k\cdot \varepsilon)+(\varepsilon,\varepsilon)
\]
Note that 
\[
\sum_{\ell,k\in \Zset} \Prb_{s_0}^{\zeta}(A_H^{\ell,k}|R_C)\cdot (\ell\cdot \varepsilon,k\cdot \varepsilon) \quad \leq \quad
(\Ex{\zeta}{s_0}{\lraname|R_C}, \Ex{\zeta}{s_0}{\lrvhname|R_C})
\]
By Lemma~\ref{lem:approx-two}, there are $\ell,k,\ell',k'\in \Zset$ and $p\in [0,1]$ such that $\Prb_{s_0}^{\zeta}(A_H^{\ell,k}|R_C)>0$ and $\Prb_{s_0}^{\zeta}(A_H^{\ell',k'}|R_C)>0$ and 
\begin{equation}\label{eq:approx-two-h}
p\cdot (\ell\cdot \varepsilon,k\cdot \varepsilon)+(1-p)\cdot (\ell'\cdot \varepsilon,k'\cdot \varepsilon)\leq \sum_{\ell,k\in \Zset} \Prb_{s_0}^{\zeta}(A_H^{\ell,k}|R_C)\cdot (\ell\cdot \varepsilon,k\cdot \varepsilon)\leq
(\Ex{\zeta}{s_0}{\lraname|R_C}, \Ex{\zeta}{s_0}{\lrvhname|R_C})
\end{equation}
Let us focus on $(\ell\cdot \varepsilon,k\cdot \varepsilon)$ and construct a frequency function $f$ on $C$ such that
\[
 (\lraname[f],\lrvhname[f,u])
\quad \leq \quad (\ell\cdot \varepsilon,k\cdot \varepsilon)+(\varepsilon,\varepsilon)
\]
The construction is identical to the proof of the corresponding proposition for 
local variance.
\begin{claim}
For every run $\omega\in R_C$ %
there is a sequence of numbers $T_1[\omega],T_2[\omega],\ldots$ such that
all the following limits are defined:
\[
\lim_{i\rightarrow \infty} \frac{1}{T_i[\omega]} \sum_{j=1}^{T_i[\omega]} r(A_j(\omega))\quad = \quad \lraname(\omega)
\qquad
\text{and}
\qquad
\lim_{i\rightarrow \infty} \frac{1}{T_i[\omega]} \sum_{j=1}^{T_i[\omega]} (r(A_j(\omega))-u)^2\quad \le \quad \lrvhname(\omega)
\]
and for every action $a\in A$ there is a number $f_{\omega}(a)$ such that
\[
\lim_{i\rightarrow \infty} \frac{1}{T_i[\omega]} \sum_{j=1}^{T_i[\omega]} I_a(A_j(\omega))\quad = \quad f_{\omega}(a)
\]
(Here $I_a(A_j(\omega))=1$ if $A_j(\omega)=a$, and $I_a(A_j(\omega))=0$ otherwise.)

Moreover, for almost all runs $\omega$ of $R_C$ we have that $f_{\omega}$ is a frequency function on $C$ and that 
$f_{\omega}$ determines $(\lraname(\omega),\lrvhname(\omega))$, i.e.,
$\lraname(\omega)=\lraname(f_{\omega})$ and $\lrvhname(\omega)\ge\lrvhname(f_{\omega},u)$.
\end{claim}
\begin{proof}
The proof is identical to the proof of Claim~\ref{claim:subsequence-local}, we only substitute the equation~(\ref{eq:subsequence-local})
with
\begin{equation}
\tag{\ref{eq:subsequence-local}a}
\lim_{i\rightarrow \infty} \frac{1}{T_i[\omega]} \sum_{j=1}^{T_i[\omega]} (r(A_j(\omega))-u)^2\quad \le \quad \lrvhname(\omega)
\end{equation}
and then instead of proving $\lrvlname(\omega)= \lrvname[f_{\omega}]$ we use the equality
\begin{eqnarray*}
\lrvhname(\omega) & \ge & \lim_{i\rightarrow \infty} \frac{1}{T_i[\omega]} \sum_{j=1}^{T_i[\omega]} (r(A_j(\omega))-u)^2 \\
& = & \lim_{i\rightarrow \infty} \frac{1}{T_i[\omega]} \sum_{j=1}^{T_i[\omega]} \sum_{a\in C} I_a(A_j(\omega))\cdot (r(a)-u)^2 \\
& = & \sum_{a\in C} (r(a)-u)^2 \cdot \lim_{i\rightarrow \infty} \frac{1}{T_i[\omega]} \sum_{j=1}^{T_i[\omega]} I_a(A_j(\omega)) \\
& = & \sum_{a\in C} (r(a)-u)^2 \cdot f_{\omega}(a)\\
& = & \lrvhname[f_{\omega},u]
\end{eqnarray*}
The desired result follows.
\end{proof}

Now pick an arbitrary run $\omega$ of $A_H^{k,\ell}$ such that $f_{\omega}$ is a frequency function. Then
\[
(\lraname(f_{\omega}),\lrvhname(f_{\omega},u))\le(\lraname(\omega),\lrvhname(\omega))\leq (\ell\cdot \varepsilon,k\cdot\varepsilon)+(\varepsilon,\varepsilon)
\]
Similarly, for $\ell',k'$ we obtain $f'_{\omega}$ such that 
\[
(\lraname(f'_{\omega}),\lrvhname(f'_{\omega},u))\le(\lraname(\omega),\lrvhname(\omega))\leq (\ell'\cdot \varepsilon,k'\cdot\varepsilon)+(\varepsilon,\varepsilon)
\]
This together with equation~(\ref{eq:approx-two-h}) from page~\pageref{eq:approx-two-h} gives the desired result:
\begin{align*}
p\cdot (\lraname(f_{\omega}),\lrvhname(f_{\omega},u)) +(1-p)\cdot  (\lraname(f'_{\omega}),\lrvhname(f'_{\omega},u))
 & \leq p\cdot \left((\ell\cdot \varepsilon,k\cdot \varepsilon)+(\varepsilon,\varepsilon)\right)+(1-p)\cdot \left((\ell'\cdot \varepsilon,k'\cdot \varepsilon)+(\varepsilon,\varepsilon)\right) \\
  & \leq  (\Ex{\zeta}{s_0}{\lraname|R_C}, \Ex{\zeta}{s_0}{\lrvhname|R_C})+(\varepsilon,\varepsilon)
\end{align*}
This finishes the proof of the first item of Lemma~\ref{lemm:hybrid-finite}.

We continue with the proof of the second item of Lemma~\ref{lemm:hybrid-finite}.
Assume that the system $L_H^\zeta$ has a solution $\bar{y}_a,\bar{x}_a,\bar{x}'_a$ for every $a\in A$.
We define two memoryless strategies $\kappa$ and $\kappa'$ as follows: Given $s\in S$ and $a\in \act{s}$, we define
\[
\kappa(s)(a)=\bar{x}_a\ /\ \sum_{b\in \act{s}} \bar{x}_b
\qquad
\text{and}
\qquad
\kappa'(s)(a)=\bar{x}'_a\ /\ \sum_{b\in \act{s}} \bar{x}'_b
\]
respectively.

Using similar arguments as in~\cite{BBCFK:MDP-two-views} it can be shown that there is a 3-state stochastic update strategy $\xi$ with memory elements $m_1,m_2,m'_2$ satisfying the following: 
A run of $G^{\xi}$ starts in $s_0$ with a fixed initial distribution on memory elements. In $m_1$ the strategy plays according to a fixed memoryless strategy until the memory changes either to $m_2$, or to $m'_2$.
In $m_2$ (or in $m'_2$), the strategy $\xi$ plays according to $\kappa$ (or according to $\kappa'$, resp.) and never changes its memory element. The key ingredient is that for every BSCC $D$ of $G^{\kappa}$ we have that
\[
\Prb^{\xi}_{s_0}(\text{switch to }\kappa\text{ in }D)\quad = \quad \sum_{a\in D\cap A} \bar{x}_a\quad =:\quad \bar{x}_D
\]
and for every BSCC $D'$ of $G^{\kappa'}$ we have that
\[
\Prb^{\xi}_{s_0}(\text{switch to }\kappa'\text{ in }D')\quad = \quad \sum_{a\in D'\cap A} \bar{x}'_a\quad =:\quad \bar{x}'_{D'}
\]
Here $\Prb^{\xi}_{s_0}(\text{switch to }\kappa\text{ in }D')$ (or $\Prb^{\xi}_{s_0}(\text{switch to }\kappa'\text{ in }D')$) is the probaibility that $\xi$ switches its state to $m_2$ (or to $m'_2$) in one of the states of $D$ (or $D'$).

Given a BSCC $D$ of $G^{\xi}$, almost all runs $\omega$ of $G^{\xi}_{s_0}$ that stay in $D$ with the memory element $m_2$ have the frequency of $a\in D\cap A$ equal to $\bar{x}_a / \bar{x}_D$. Thus $\lraname(\omega)=\sum_{a\in D\cap A} \bar{x}_a / \bar{x}_D\cdot r(a)$. Similarly, if the BSCC is $D'$ and the memory element is $m'_2$, then $\lraname(\omega)=\sum_{a\in D'\cap A}\bar{x}'_a / \bar{x}'_{D'}\cdot r(a)$. Thus we have the following 
desired equalities: (1)~Equality for $u$
\begin{eqnarray*}
\Ex{\xi}{s_0}{\lraname} & = & \sum_{D\text{ is a BSCC of }G^{\kappa}} \Prb^{\xi}_{s_0}(\text{switch to }\kappa\text{ in }D)\cdot \sum_{a\in D\cap A} \bar{x}_a / \bar{x}_D \cdot r(a) +\\
 & & \quad + \sum_{D'\text{ is a BSCC of }G^{\kappa'}} \Prb^{\xi}_{s_0}(\text{switch to }\kappa'\text{ in }D')\cdot \sum_{a\in D'\cap A} \bar{x}'_a / \bar{x}'_{D'} \cdot r(a) \\
 & = & \sum_{C\in \Mec(G)} \Big(\sum_{a\in C\cap A} \bar{x}_a\cdot r(a) + \sum_{a\in C\cap A} \bar{x}'_a\cdot r(a)\Big)\\
 & = & u; 
\end{eqnarray*}
and (2)~Equality for $v$
\begin{eqnarray*}
\Ex{\xi}{s_0}{\lrvhname} & = & 
\sum_{D\text{ is a BSCC of }G^{\kappa}} \Prb^{\xi}_{s_0}(\text{switch to }\kappa\text{ in }D)\cdot \sum_{a\in D\cap A} \bar{x}_a / \bar{x}_D \cdot (r(a)-\Ex{\xi}{s_0}{\lraname})^2 \\
& & \quad + \sum_{D'\text{ is a BSCC of }G^{\kappa'}} \Prb^{\xi}_{s_0}(\text{switch to }\kappa'\text{ in }D')\cdot \sum_{a\in D'\cap A} \bar{x}'_a / \bar{x}'_{D'} \cdot (r(a)-\Ex{\xi}{s_0}{\lraname})^2 \\
& = & \sum_{D\text{ is a BSCC of }G^{\kappa}} \Prb^{\xi}_{s_0}(\text{switch to }\kappa\text{ in }D)\cdot \sum_{a\in D\cap A} \bar{x}_a / \bar{x}_D \cdot (r(a)-u)^2 \\
& & \quad + \sum_{D'\text{ is a BSCC of }G^{\kappa'}} \Prb^{\xi}_{s_0}(\text{switch to }\kappa'\text{ in }D')\cdot \sum_{a\in D'\cap A} \bar{x}'_a / \bar{x}'_{D'} \cdot (r(a)-u)^2 \\
& = & \sum_{C \in\Mec(G)}\left(\sum_{a\in C\cap A} \bar{x}_a \cdot (r(a)-u)^2 + \sum_{a\in C\cap A} \bar{x}'_a \cdot (r(a)-u)^2\right) \\
& = & v;
\end{eqnarray*}
The desired result follows.

\subsubsection{First item of Proposition~\ref{prop:hybrid-details} supposing finite-memory strategies exist}
\label{app-hybrid-details-a}
Let $\zeta$ be a strategy such that the following 
two conditions hold:
\[
(1)\ \Ex{\zeta}{s_0}{\lraname} = \ovk{u} \leq u; \qquad
(2)\ \Ex{\zeta}{s_0}{\lrvhname} = \ovk{v} \leq v. 
\]
By Proposition~\ref{prop:hybrid-finite-strategy} without loss of generality 
the strategy $\zeta$ is a finite-memory strategy.
Since $\zeta$ is a finite-memory strategy, the frequencies are well-defined,
and for an action $a\in A$, let
\[
f(a)
\coloneqq
\lim_{\ell\to\infty}
\frac{1}{\ell} \sum_{t=0}^{\ell-1} \Pr{\zeta}{s_0}{A_t=a} 
\]
denote the frequency of action $a$.
We will first show that setting
\(
x_a \coloneqq f(a)
\)
for all $a\in A$ satisfies  Eqns.~(\ref{eq:xah}), Eqns.~(\ref{eq:rewh}) and Eqns.~(\ref{eq:varh}) 
of~$L_H$.

\smallskip\noindent{\em Satisfying Eqns~\ref{eq:xah}.} 
To prove that Eqns.~(\ref{eq:xah}) are satisfied, it suffices to show that
for all $s\in S$ we have 
\[
\sum_{a\in A} f(a)\cdot \delta(a)(s) = \sum_{a \in \act{s}} f(a).
\]
We establish this below:
\begin{eqnarray*}
\sum_{a\in A} f(a)\cdot \delta(a)(s) & = &
  \sum_{a\in A}\lim_{\ell\to\infty} \frac{1}{\ell} \sum_{t=0}^{\ell-1} \Pr{\zeta}{s_0}{A_t=a}
   \cdot \delta(a)(s)\\
   & = & 
   \lim_{\ell\to\infty} \frac{1}{\ell} \sum_{t=0}^{\ell-1} 
   \sum_{a\in A} \Pr{\zeta}{s_0}{A_t=a}
   \cdot \delta(a)(s) \\
   & = &
   \lim_{\ell\to\infty} \frac{1}{\ell} \sum_{t=0}^{\ell-1} 
   \Pr{\zeta}{s_0}{S_{t+1}=s} \\
   & = &
   \lim_{\ell\to\infty} \frac{1}{\ell} \sum_{t=0}^{\ell-1} 
   \Pr{\zeta}{s_0}{S_{t}=s} \\
   & = &
   \lim_{\ell\to\infty} \frac{1}{\ell} \sum_{t=0}^{\ell-1} 
   \sum_{a \in \act{s}}
   \Pr{\zeta}{s_0}{A_{t}=a} \\
   & = &
   \sum_{a \in \act{s}}
   \lim_{\ell\to\infty} \frac{1}{\ell} \sum_{t=0}^{\ell-1} 
   \Pr{\zeta}{s_0}{A_{t}=a} \\
   & = & 
   \sum_{a \in \act{s}}
    f(a)
    \;.
\end{eqnarray*}
Here the first and the seventh equality follow from the definition of $f$.
The second and the sixth equality follow from the linearity of the limit.
The third equality follows by the definition of $\delta$.
The fourth equality is obtained from the following:
\begin{eqnarray*}
\lim_{\ell\to\infty} \frac{1}{\ell} \sum_{t=0}^{\ell-1} 
\Pr{\zeta}{s_0}{S_{t+1}=s}
 -  
\lim_{\ell\to\infty} \frac{1}{\ell} \sum_{t=0}^{\ell-1} 
\Pr{\zeta}{s_0}{S_{t}=s}
& = & 
\displaystyle
\lim_{\ell\to\infty} \frac{1}{\ell} \sum_{t=0}^{\ell-1} 
(
\Pr{\zeta}{s_0}{S_{t+1}=s}
-
\Pr{\zeta}{s_0}{S_{t}=s}
)
\\
& = & 
\displaystyle
\lim_{\ell\to\infty} \frac{1}{\ell}
(
\Pr{\zeta}{s_0}{S_{\ell+1}=s}
-
\Pr{\zeta}{s_0}{S_{1}=s}
)
=0
\end{eqnarray*}

\smallskip\noindent{\em Satisfying Eqns~\ref{eq:rewh}.}
We will show that $\sum_{a \in A} f(a) \cdot r(a) = \ovk{u}$. 
\[
\sum_{a\in A} r(a) \cdot f(a)
\quad =\quad
\sum_{a \in A}r(a)\cdot
\lim_{\ell\to\infty} \frac{1}{\ell} \sum_{t=0}^{\ell-1} 
\Pr{\zeta}{s_0}{A_t=a}
\quad =\quad
\lim_{\ell\to\infty} \frac{1}{\ell} \sum_{t=0}^{\ell-1} 
\sum_{a\in A}r(a)\cdot\Pr{\zeta}{s_0}{A_t=a}
\quad =\quad
\lim_{\ell\to\infty} \frac{1}{\ell} \sum_{t=0}^{\ell-1} 
\Ex{\zeta}{s_0}{r(A_t)}
\quad=\quad \ovk{u}
\;.
\]
Here, the first equality is the definition of $f(a)$;  
the second equality follows from the linearity of the limit; 
the third equality follows by linearity of expectation;
the fourth equality involves exchanging limit and expectation 
and follows from Lebesgue Dominated convergence theorem 
(see, e.g.~\cite[Chapter~4, Section~4]{Royden88}), 
since $|r(A_t)|\leq W$, where $W=\max_{a\in A} |r(a)|$. 
The desired result follows.

\smallskip\noindent{\em Satisfying Eqns~\ref{eq:varh}.} 
We will now show the satisfaction of Eqns~\ref{eq:varh}. 
First we have that
\[
\Ex{\zeta}{s_0}{\lrvhname} =  
\Ex{\zeta}{s_0}{\lim\sup_{\ell \to \infty} \frac{1}{\ell} \sum_{t=0}^{\ell-1}(r(A_t) - \ovk{u})^2} 
 = 
\Ex{\zeta}{s_0}{\lim_{\ell \to \infty} \frac{1}{\ell} \sum_{t=0}^{\ell-1}(r(A_t) - \ovk{u})^2} 
 = 
\lim_{\ell \to \infty} \frac{1}{\ell} \Ex{\zeta}{s_0}{\sum_{t=0}^{\ell-1}(r(A_t) - \ovk{u})^2}. 
\]
The first equality is by definition; the second equality about existence of limit 
follows from the fact that $\zeta$ is a finite-memory strategy; and 
the final equality of exchange of limit and the expectation follows from 
Lebesgue Dominated convergence theorem (see, e.g.~\cite[Chapter~4, Section~4]{Royden88}), 
since $(r(A_t)-\ovk{u})^2 \leq (2\cdot W)^2$, where $W=\max_{a\in A} |r(a)|$. 
We have
\[
\begin{array}{rcl}
\displaystyle 
\lim_{\ell \to \infty} \frac{1}{\ell} \sum_{t=0}^{\ell-1} \Ex{\zeta}{s_0}{(r(A_t) - \ovk{u})^2} 
& = & 
\displaystyle 
\lim_{\ell \to \infty} \frac{1}{\ell} \sum_{t=0}^{\ell-1} 
\Big(\Ex{\zeta}{s_0}{r^2(A_t)} 
-2\cdot \ovk{u} \cdot \Ex{\zeta}{s_0}{r(A_t)}
+ \ovk{u}^2\Big) \\[2.5ex]
& = & 
\displaystyle 
\lim_{\ell \to \infty} \frac{1}{\ell} \sum_{t=0}^{\ell-1} 
\Ex{\zeta}{s_0}{r^2(A_t)} 
-2\cdot \ovk{u} \cdot 
\lim_{\ell \to \infty} \frac{1}{\ell} \sum_{t=0}^{\ell-1} 
\Ex{\zeta}{s_0}{r(A_t)}
+ \ovk{u}^2 \\[3ex]
& = & 
\displaystyle
\sum_{a\in A} r^2(a) \cdot f(a) - 
2\cdot \ovk{u} \cdot \sum_{a \in A} r(a) \cdot f(a) 
+ \ovk{u}^2 \\[3ex]
& = & 
\displaystyle
\sum_{a\in A} r^2(a) \cdot f(a) - \bigg(\sum_{a \in A} r(a) \cdot f(a)\bigg)^2 
\end{array}
\]
The first equality is by rewriting the term within the expectation and by
linearity of expectation; the second equality is by linearity of limit; 
the third equality follows by the equality to show satisfaction of 
Eqns~\ref{eq:rewh} (it follows from the equality for Eqns~\ref{eq:rewh} 
that $\lim_{\ell \to \infty} \frac{1}{\ell} \sum_{t=0}^{\ell-1} \Ex{\zeta}{s_0}{r^2(A_t)} 
=\sum_{a \in A} r^2(a) \cdot f(a)$ by simply considering the reward function 
$r^2$ instead of $r$);
and the final equality follows from the equality to prove Eqns~\ref{eq:rewh}.
Thus we have the following equality:
\[
\sum_{a\in A} r^2(a) \cdot f(a) - \bigg(\sum_{a \in A} r(a) \cdot f(a)\bigg)^2 
 = 
\lim_{\ell \to \infty} \frac{1}{\ell} \sum_{t=0}^{\ell-1} \Ex{\zeta}{s_0}{(r(A_t) - \ovk{u})^2} \\[3ex]
 =  
\Ex{\zeta}{s_0}{\lrvhname} =\ovk{v} \leq v.
\]

Now we have to set the values for $y_\chi$, $\chi\in A\cup S$,
and prove that they satisfy the rest of $L_H$ when the values $f(a)$ are 
assigned to $x_a$. 
By Lemma~\ref{lemma:stay-mec} almost every run of $G^{\zeta}$ eventually 
stays in some MEC of $G$.
For every MEC $C$ of $G$, let $y_C$ be the probability of all runs
in $G^{\zeta}$ that eventually stay in~$C$. Note that
\[
\sum_{a\in A\cap C}
f(a)
\quad =\quad
\sum_{a\in A\cap C}
\lim_{\ell\to\infty} 
\frac{1}{\ell}
\sum_{t=0}^{\ell-1} \Pr\zeta{s_0}{A_t=a}
\quad =\quad
\lim_{\ell\to\infty} 
\frac{1}{\ell}
\sum_{t=0}^{\ell-1}
\sum_{a\in A\cap C}
\Pr\zeta{s_0}{A_t=a}
\quad =\quad
\lim_{\ell\to\infty}
\frac{1}{\ell}
\sum_{t=0}^{\ell-1} \Pr\zeta{s_0}{A_t\in C} 
\quad =\quad y_C \;.
\]
Here the last equality follows from the fact that 
$\lim_{\ell\to\infty} \Pr\zeta{s_0}{A_{{\ell}}\in C}$ 
is equal to the probability of all runs in $G^{\zeta}$ that 
eventually stay in~$C$ 
(recall that almost every run stays eventually in a MEC of $G$) and
the fact that the Ces\`{a}ro sum of a convergent sequence is equal 
to the limit of the sequence.

By the previous paragraph there is $\zeta$ such that $\Pr\zeta{s_0}{\staymec{C}} = \sum_{a\in A\cap C} f(a)$,
so we can define $y_a$ and $y_s$ in the same way as done 
in~\cite[Proposition~2]{BBCFK:MDP-two-views} (this solution is based on the results 
of~\cite{EKVY:multi-objectives}; the proof is exactly the same as the proof of \cite[Proposition~2]{BBCFK:MDP-two-views}, we only skip the part in which the assignment to $x_a$s is defined).
This completes the proof of the desired result.

\subsubsection{Proof that Eqns~\ref{eq:varh} is satisfied by $\sigma$}\label{app-hybrid-details-b}
We argue that the strategy $\sigma$ from~\cite[Proposition~1]{BBCFK:MDP-two-views} satisfies Eqns~\ref{eq:varh}.
We show that for the strategy $\sigma$ we have:
$\Ex{\sigma}{s}{\lrvhname} = \Ex{\sigma}{s}{\lraname_{r^2}} - \Ex{\sigma}{s}{\lraname}^2$.
It follows immediately that Eqns~\ref{eq:varh} is satisfied.
Since $\sigma$ is a finite-memory strategy, all the limit-superior can be replaced
with limits.
Then we use the the equality from Appendix~\ref{app:relation} where we showed that
\[
\Ex{\sigma}{s}{\lrvhname}
=\Ex{\sigma}{s}{\lim_{n\rightarrow \infty}  \frac{1}{n}\sum_{i=0}^{n-1} r(A_i)^2}
   - \Ex{\sigma}{s}{\lraname}^2
\]
which is equal to
$\Ex{\sigma}{s}{\lraname_{r^2}}
   - \Ex{\sigma}{s}{\lraname}^2$.

\subsubsection{Properties of the quadratic constraints of $L_H$.}\label{app-hybrid-semidefinite}
We now establish that the quadratic constraints of $L_H$ (i.e., 
Eqns~\ref{eq:varh}) satisfies that it is a \emph{negative semi-definite} 
constraint of \emph{rank~1}. 
Let us denote by $\vec{x}$ the vector of variables $x_a$, and 
$\vec{r}$ the vector of rewards $r(a)$, for $a \in A$.
Then the quadratic constraint of Eqns~\ref{eq:varh} is specified in matrix notation 
as: $\sum_{a \in A} x_a \cdot r^2(a) - \vec{x}^T\cdot  Q \cdot \vec{x}$, 
where $\vec{x}^T$ is the transpose of $\vec{x}$, 
and  the matrix $Q$ is as follows: $Q_{ij}=r(i)\cdot r(j)$. 
Indeed, we have
$\vec{x}^T\cdot Q \cdot \vec{x} = \vec{z}^T \cdot \vec{x}$
where $\vec{z}_i = \sum_{k\in A} x_k \cdot r(i) \cdot r(k)$
and so
\begin{eqnarray*}
\vec{x}^T\cdot Q \cdot \vec{x} &=&
 \sum_{i\in A} x_i \cdot \sum_{k\in A} x_k \cdot r(i) \cdot r(k)\\
 &=& \bigg(\sum_{i\in A} (x_ir(i))^2\bigg) + \sum_{i\in A} x_i \cdot \sum_{k\in A, k\neq i} x_k \cdot r(i) \cdot r(k)\\
 &=& \bigg(\sum_{i\in A} (x_ir(i))^2\bigg) + \sum_{i\in A} \sum_{k < i} 2\cdot x_i \cdot r(i) \cdot x_k \cdot r(k)\\
 &=& \bigg(\sum_{i\in A} x_ir(i)\bigg)^2
\end{eqnarray*}
where in the last but one equality we use an arbitrary order on $A$, and where the last equality follows by multinomial theorem.

The desired properties of $Q$ are established as follows:
\begin{itemize}
\item \emph{Negative semi-definite.} 
We argue that $Q$ is a positive semi-definite matrix.
A sufficient condition to prove that $Q$ is positive semi-definite 
is to show that for all real vectors $\vec{y}$ we have 
$\vec{y}^T \cdot Q \cdot \vec{y} \geq 0$.
For any real vector $\vec{y}$ we have 
$\vec{y}^T \cdot Q \cdot \vec{y} =(\sum_{a \in A} y_a \cdot r(a))^2 \geq 0$
(as the square of a real-number is always non-negative).
It follows that Eqns~\ref{eq:varh} is a negative semi-definite constraint.

\item \emph{Rank of $Q$ is~1.} 
We now argue that rank of $Q$ is~1. 
We observe that the matrix $Q$ with $Q_{ij}=r_i \cdot r_j$ is the 
outer-product matrix of $\vec{r}$ and $\vec{r}^T$, where $\vec{r}$ and 
$\vec{r}^T$ denote the vector of rewards and its transpose, respectively,
i.e., $Q = \vec{r} \cdot \vec{r}^T$.
Since $Q$ is obtained from a single vector (and its transpose) it follows
that $Q$ has rank~1.

\end{itemize}

\newcommand{\AS}{\mathsf{AlmostSure}}
\newcommand{\Reach}{\mathsf{Reach}}

\subsection{Details for Section~\ref{sect-zero}}
Some of our algorithms will be based on the notion of almost-sure winning for 
reachability and coB\"uchi objectives.

\smallskip\noindent{\bf Almost-sure winning, reachability and 
coB\"uchi objectives.} 
An objective $\Phi$ defines a set of runs.
For a set $B \subseteq A$ of actions, we  
(i)~recall the reachability objective $\Reach(B)$ that
specifies the set of runs  $\pat=s_1 a_1 s_2 a_2\ldots$ such that for some  
$i\geq 0$ we have $a_i \in B$ (i.e., some action from $B$ is visited at least 
once); and (ii)~define the coB\"uchi objective $\coBuchi(B)$ that 
specifies the set of runs $\pat=s_1 a_1 s_2 a_2\ldots$ such that for some $i\geq 0$ for all $j \geq i$ 
we have $a_j \in B$ (i.e., actions not in $B$ are visited finitely often).
Given an objective $\Phi$, a state $s$ is an \emph{almost-sure} winning state 
for the objective if there exists a strategy $\sigma$ (called an almost-sure winning
strategy) to ensure the objective
with probability~1, i.e., $\Pr{\sigma}{s}{\Phi}=1$.
We recall some basic results related to almost-sure winning for reachability and
coB\"uchi objectives.

\begin{theorem}[\cite{CH12,CJH04}]\label{thrm:almost}
For reachability and coB\"uchi objectives whether a state is almost-sure 
winning can be decided in polynomial time (in time $O((|S|\cdot |A|)^2)$) 
using discrete graph theoretic algorithms.
Moreover, both for reachability and coB\"uchi objectives, if there is an 
almost-sure winning strategy, then there is a memoryless pure almost-sure 
winning strategy.
\end{theorem}

\smallskip\noindent{\bf Basic facts.} We will also use the following basic 
fact about \emph{finite} Markov chains. 
Given a Markov chain, and a state $s$:
(i)~\emph{(Fact~1).} 
The local variance is zero iff for every bottom scc reachable from $s$ there 
exists a reward value $r^*$ such that all rewards of the bottom scc is $r^*$.
positive.  
(ii)~\emph{(Fact~2).} 
The hybrid variance is zero iff there exists a reward value $r^*$ such that 
for every bottom scc reachable from $s$ all rewards of the bottom scc is $r^*$.
(iii)~\emph{(Fact~3).} 
The global variance is zero iff there exists a number $y$ such that 
for every bottom scc reachable from $s$ the expected mean-payoff value of 
the bottom scc is $y$.

\subsubsection{Zero Hybrid Variance}\label{app-zero-hybrid}

We establish the correctness of our algorithm with the following lemma.
\begin{lemma}\label{lemm:zero-hybrid}
Given an MDP $G=(S,A,\mathit{Act},\delta)$, a starting state $s$, and a reward function $r$,
the following assertions hold:
\begin{enumerate}
\item If $\beta$ is the output of the algorithm, then there is a strategy 
to ensure that the expectation is at most $\beta$ and the hybrid variance is zero.

\item If there is a strategy to ensure that the expectation is at most $\beta^*$ and 
the hybrid variance is zero, then the output $\beta$ of the algorithm
satisfies that $\beta \leq \beta^*$.
\end{enumerate}
\end{lemma}
\begin{proof} 
The proofs of the items are as follows:
\begin{enumerate}
\item If the output of the algorithm is $\beta$, 
then consider $A'$ to be the set of actions with reward $\beta$. 
By step (2) of the algorithm we have that there exists an almost-sure
winning strategy for the objective $\coBuchi(A')$, and by Theorem~\ref{thrm:almost}
there exists a memoryless pure almost-sure winning strategy $\sigma$ for 
the coB\"uchi objective.
Since $\sigma$ is an almost-sure winning strategy for the coB\"uchi objective, 
it follows that in the Markov chain $G_s^\sigma$ every bottom scc $C$ reachable 
from $s$ consists of reward $\beta$ only.
Thus the expectation given the strategy $\sigma$ is $\beta$, and by Fact~2 
for Markov chains the hybrid variance is zero.

\item Consider a strategy to ensure that the expectation is at most 
$\beta^*$ with hybrid variance zero.
By the results of Proposition~\ref{prop:hybrid-finite-strategy}
there is a finite-memory strategy $\sigma$ to ensure expectation $\beta^*$ 
with hybrid variance zero. 
Given the strategy $\sigma$, if there exists an action $a$ with reward other
than $\beta^*$ that appear in a bottom scc, then the hybrid variance
is greater than zero (follows from Fact~2 for Markov chains). 
Thus every bottom scc in $G_s^\sigma$ that is reachable from $s$ consists
of reward $\beta^*$ only. 
Hence $\sigma$ is also an almost-sure winning strategy from $s$ for the objective
$\coBuchi(A^*)$, where $A^*$ is the set of actions with reward $\beta^*$.
Let $\beta^*=\beta_j$, because $\beta_j$ satisfies the requirement of step (2)
of the algorithm, we get that the output of the algorithm is a number
$\beta \leq \beta^*$. 

\end{enumerate}
The desired result follows.
\end{proof}

For reader's convenience, a formal description of the algorithm is given as Algorithm~\ref{algo:hybrid}.

\begin{algorithm}[t]
\caption{\bf Zero Hybrid Variance}
\label{algo:hybrid}
{ 
\begin{tabbing}
aa \= aa \= aaa \= aaa \= aaa \= aaa \= aaa \= aaa \kill
\> {\bf Input :} An MDP $G=(S,A,\mathit{Act},\delta)$, a starting state $s$, and a reward function $r$. \\
\> {\bf Output:} A reward value $\beta$ or \texttt{NO}. \\
\> 1. Sort the reward values $r(a)$ for $a \in A$ in an increasing order $\beta_1 < \beta_2 < \ldots < \beta_n$; \\
\> 2. $i:= 1$; \\
\> 3. {\bf repeat }   \\
\>\> 3.1. Let $A_i$ be the set of actions with reward $\beta_i$;  \\
\>\> 3.2. {\bf if} there exists an almost-sure winning strategy for $\coBuchi(A_i)$ \\
\>\>\> {\bf return} $\beta_i$; \\
\>\> 3.3 {\bf if} $i=n$  \\
\>\>\> {\bf return} \texttt{NO}; \\
\>\> 3.4 $i:=i+1$; \\
\end{tabbing}
}
\end{algorithm}

\subsubsection{Zero Local Variance}\label{app-zero-local}
For a state $s$, let $\alpha(s)$ denote the minimal expectation that can be ensured 
along with zero local variance.

Our goal is to show that $\ovk{\beta}(s)=\alpha(s)$. 
We first describe the two-step computation of $\ovk{\beta}(s)$.
\begin{enumerate}
\item Compute the set of states $U$ such that there is an almost-sure winning 
strategy for the objective $\Reach(T)$.
\item Consider the sub-MDP of $\ovk{G}$ induced by the set $U$ which is 
described as follows: $(U,A,\mathit{Act}_U,\delta)$ such that for all $s \in U$ 
we have $\mathit{Act}_U(s)=\{ a \in \mathit{Act}(s)\mid\text{for all $s'$, if } 
\delta(a)(s')>0, \text{ then } s' \in U\}$. 
In the sub-MDP compute the minimal expected payoff for the cumulative reward,
and this computation is similar to computation of optimal values for MDPs
with reachability objectives and can be achieved in polynomial time with 
linear programming.
\end{enumerate}
Note that by construction every new action $a_s$ has negative reward and all other 
actions have zero reward. 
A memoryless pure almost-sure winning strategy for a state $s$ in $U$ to reach 
$T$ ensures that the expected cumulative reward is negative, and hence
$\wh{\beta}(s)<0$ for all $s \in U$.
Also observe that if $U$ is left, then almost-sure reachability to $T$ 
cannot be ensured. 
Hence any strategy that ensures almost-sure reachability to $T$ must ensure
that $U$ is not left. 
We now claim that any memoryless pure optimal strategy in the sub-MDP 
for the cumulative reward also ensures almost-sure reachability to $T$.
Consider a memoryless pure optimal strategy $\sigma$ for the cumulative reward. 
Since every state in $T_S$ is an absorbing state (state with a self-loop) 
every bottom scc $C$ in the Markov chain is either contained in $T_S$ 
or does not intersect with $T_S$.
If there is a bottom scc $C$ that does not intersect with $T_S$, 
then the expected cumulative reward in the bottom scc is zero, 
and this is a contradiction that $\sigma$ is an optimal strategy and 
for all $s \in U$ we have $\wh{\beta}(s)<0$.
It follows that every bottom scc in the Markov chain is contained in $T_S$
and hence almost-sure reachability to $T$ is ensured.
Hence it follows that $\wh{\beta}(s)$ can be computed in polynomial time,
and thus $\ovk{\beta}(s)$ can be computed in polynomial time.
In the following two lemmas we show that $\alpha(s)=\ovk{\beta}(s)$.

\begin{lemma}\label{lemm:zero-local1}
For all states $s$ we have $\alpha(s) \geq \ovk{\beta}(s)$.
\end{lemma}
\begin{proof}
We only need to consider the case when from $s$ zero local variance 
can be ensured.
Consider a strategy that ensures expectation $\alpha(s)$ along 
with zero local variance, and by the results of Proposition~\ref{prop:local-main} 
there is a witness
finite-memory strategy $\sigma^*$.
Consider the Markov chain $G_s^{\sigma^*}$.
Consider a bottom scc $C$ of the Markov chain reachable from $s$ 
and we establish the following properties:
\begin{enumerate}
\item Every reward in the bottom scc must be the same. 
Otherwise the local variance is positive (by Fact~1 for Markov chains). 

\item Let $r^*$ be the reward of the bottom scc. 
We claim that for all states $s'$ that appears in the bottom 
scc we have $\beta(s') \leq r^*$. 
Otherwise if $\beta(s')>r^*$, playing according the strategy $\sigma$ 
in the bottom scc from $s'$ we ensure zero hybrid variance with 
expectation $r^*$ contradicting that $\beta(s')$ is the minimal
expectation along with zero hybrid variance.

\end{enumerate}
It follows that in every bottom scc $C$ of the Markov chain 
the reward $r^*$ of the bottom scc satisfy that $r^* \geq \beta(s')$,
for every $s'$ that appears in $C$. 
Also observe that the strategy $\sigma^*$ ensures almost-sure reachability
to the set $T_S$ of states where zero hybrid variance can be ensured.
We construct a strategy $\sigma$ in MDP $\ovk{G}$ as follows: the strategy 
plays as $\sigma^*$ till a bottom scc is reached, and as soon as a bottom 
scc $C$ is reached at state $s'$, the strategy in $\ovk{G}$ chooses the 
action $a_{s'}$ to proceed to the state $\ovk{s}'$.
The strategy ensures that the cumulative reward in $\ovk{G}$ is at most 
$\alpha(s) - M$, i.e., $\alpha(s)-M \geq \wh{\beta}(s)$.
It follows that $\alpha(s) \geq \ovk{\beta}(s)$.
\end{proof}

\begin{lemma}\label{lemm:zero-local2}
For all states $s$ we have $\alpha(s) \leq \beta^*(s)$.
\end{lemma}
\begin{proof}
Consider a witness memoryless pure strategy $\sigma^*$ in $\ovk{G}$ that achieves 
the optimal cumulative reward value. 
We construct a witness strategy $\sigma$ for zero local variance in $G$ as
follows: play as $\sigma^*$ till the set $T$ is reached (note that $\sigma^*$ ensures
almost-sure reachability to $T$), and after $T$ is reached, if a state 
$\ovk{s}$ is reached, then switch to the memoryless pure strategy from $s$
to ensure expectation at most $\beta(s)$ with zero hybrid variance. 
The strategy $\sigma$ ensures that every bottom scc of the 
resulting  Markov chain consists of only one reward value.   
Hence the local variance is zero.
The expectation given strategy $\sigma$ is at most $\beta^*(s)$.
Hence the desired result follows.
\end{proof}

\subsubsection{Zero Global Variance}\label{app-zero-global}

The following lemma shows that in a MEC, any expectation in the interval is 
realizable with zero global variance.

\begin{lemma}\label{lemm:zero-global}
Given an MDP $G=(S,A,\mathit{Act},\delta)$, a starting state $s$, and a reward function $r$,
the following assertions hold:
\begin{enumerate}
\item If $\ell$ is the output of the algorithm, then there is a strategy 
to ensure that the expectation is at most $\ell$ and the global variance is zero.

\item If there is a strategy to ensure that the expectation is at most $\ell^*$ and 
the global variance is zero, then the output $\ell$ of the algorithm
satisfies that $\ell \leq \ell^*$.
\end{enumerate}
\end{lemma}
\begin{proof} 
The proof of the items are as follows:
\begin{enumerate}
\item If the output of the algorithm is $\ell$, 
then consider $\mathcal{C}$ to be the set of MEC's whose interval contains $\ell$. 
Let $A'=\bigcup_{C_j \in \mathcal{C}} C_j$.
By step~(4)(b) of the algorithm we have that there exists an almost-sure
winning strategy for the objective $\Reach(A')$, and by Theorem~\ref{thrm:almost}
there exists a memoryless pure almost-sure winning strategy $\sigma_R$ for 
the reachability objective.
We consider a strategy as follows: 
(i)~play $\sigma_R$ until an end-component in $\mathcal{C}$ is reached;
(ii)~once $A'$ is reached, consider a MEC $C_j$ that is reached 
and switch to the memoryless randomized strategy 
$\sigma_\ell$ of Lemma~\ref{lem-x_c-almost-surely} to ensure that every bottom 
scc obtained in $C_j$ by fixing $\sigma_\ell$ has expected mean-payoff 
exactly $\ell$ (i.e., it ensures expectation $\ell$ with zero global variance).
Since $\sigma$ is an almost-sure winning strategy for the reachability objective to the 
MECs in $\mathcal{C}$, and once the MECs are reached the strategy $\sigma_\ell$ 
ensures that every bottom scc of the Markov chain has expectation exactly 
$\ell$, it follows that the expectation is $\ell$ and the global variance is 
zero.

\item Consider a strategy to ensure that the expectation is at most 
$\ell^*$ and the global variance zero.
By the results of Theorem~\ref{thm:global} 
there is a finite-memory strategy $\sigma$ to ensure expectation $\ell^*$ 
with global variance zero. 
Given the strategy $\sigma$, consider the Markov chain $G_s^\sigma$. 
Let $\widehat{\mathcal{C}}=\{ \widehat{C} \mid \widehat{C} \text{ is a bottom scc reachable
from $s$ in } G_s^\sigma\}$.
Since the global variance is zero and the expectation is $\ell^*$, every 
bottom scc $\widehat{C} \in \widehat{\mathcal{C}}$ must have that the expectation 
is exactly $\ell^*$.
Let 
\[
\begin{array}{rcl}
\mathcal{C} & = & \{ C \mid \mbox{ $C$ is a MEC and there exists 
$\widehat{C} \in \widehat{\mathcal{C}}$ such that the associated end component} \\
& & 
\mbox{ of $\widehat{C}$ is contained in $C$} \}.
\end{array}
\]
For every $C \in \mathcal{C}$ we have $\ell^* \in [\alpha_C,\beta_C]$, where 
$[\alpha_C,\beta_C]$ is the interval of $C$. 
Moreover, the strategy $\sigma$ is also a witness almost-sure winning strategy 
for the reachability objective $\Reach(A')$, where 
$A'=\bigcup_{C \in \mathcal{C}} C$. 
Let $\ell' =\min \{ \alpha_C \mid \ell \text{ is the minimal expectation of } 
C \in \mathcal{C} \}$.
Since for every $C \in \mathcal{C}$ we have $\ell^* \in [\alpha_C,\beta_C]$, it follows
that $\ell'\leq \ell^*$.
Observe that if the algorithm checks the value $\ell'$ in step (4) (say $\ell'=\ell_i$),
then the condition in step (4)(3) is true true, as $A' \subseteq \bigcup_{C_j \in \mathcal{C}_i} C_j$ and $\sigma$ will be a witness
almost-sure winning strategy to reach $\bigcup_{C_j \in \mathcal{C}_i} C_j$. 
Thus the algorithm must retrun a value 
$\ell \leq \ell' \leq \ell^*$.
\end{enumerate}
The desired result follows.
\end{proof}

The above lemma ensures the correctness and the complexity analysis is as follows:
(i)~the MEC decomposition for MDPs can be computed in polynomial time~\cite{CH11,CH12} 
(hence step~1 is polynomial);
(ii)~the minimal and maximal expectation can be computed in polynomial time 
by linear programming to solve MDPs with mean-payoff objectives~\cite{Puterman:book} 
(thus step~2 is polynomial); and
(iii)~sorting (step~3) and deciding existence of almost-sure winning strategies for 
reachability objectives can be achieved in polynomial time~\cite{CH12,CJH04}.
It follows that the algorithm runs in polynomial time.

For reader's convenience, the formal description of the algorithm is given as Algorithm~\ref{algo:global}.
\begin{algorithm}[t]
\caption{\bf Zero Global Variance}
\label{algo:global}
{ 
\begin{tabbing}
aa \= aa \= aaa \= aaa \= aaa \= aaa \= aaa \= aaa \kill
\> {\bf Input :} An MDP $G=(S,A,\mathit{Act},\delta)$, a starting state $s$, and a reward function $r$. \\
\> {\bf Output:} A reward value $\beta$ or \texttt{NO}. \\
\> 1. Compute the MEC decomposition of the MDP and let the MECs be $C_1,C_2,\ldots,C_n$. \\ 
\> 2. For every MEC $C_i$ compute the minimal expectation $\alpha_{C_i}$ and the maximal \\ 
\>\> expectation $\beta_{C_i}$ that can be ensured in the MDP induced by the MEC $C_i$; \\
\> 3. Sort the values $\alpha_{C_i}$ in a non-decreasing order $\ell_1 \leq \ell_2 \leq \ldots \leq \ell_n$; \\
\> 4. $i:= 1$; \\
\> 5. {\bf repeat }   \\
\>\> 5.1. Let $\mathcal{C}_i=\{C_j \mid \alpha_{C_j} \leq \ell_i \leq \beta_{C_j}\}$ be the MEC's whose interval contains $\ell_i$; \\
\>\> 5.2. Let $A_i= \bigcup_{C_j \in \mathcal{C}_i} C_j$ be the union of the MEC's in $\mathcal{C}_i$; \\
\>\> 5.3. {\bf if} there exists an almost-sure winning strategy for $\Reach(A_i)$ \\
\>\>\> {\bf return} $\ell_i$; \\
\>\> 5.4 {\bf if} $i=n$  \\
\>\>\> {\bf return} \texttt{NO}; \\
\>\> 5.5 $i:=i+1$; \\
\end{tabbing}
}
\end{algorithm}

}
\end{document}